%% file: Superstructure.tex
\documentclass[runningheads]{llncs}
\usepackage{amsmath}%
\usepackage{amssymb}%
\usepackage{stmaryrd}%
\usepackage{array}%
\usepackage{epic}%
\usepackage[fixamsmath]{mathtools}
\usepackage{cancel}
\usepackage{url}
\usepackage{bigstrut}

\DeclareMathAlphabet\mathbb{U}{msb}{m}{n}
\DeclareMathAlphabet{\mathrm}    {OT1}{cmr}{m}{n}
\DeclareMathAlphabet{\mathrmbf}  {OT1}{cmr}{bx}{n}
\DeclareMathAlphabet{\mathrmit}  {OT1}{cmr}{m}{it}
\DeclareMathAlphabet{\mathrmbfit}{OT1}{cmr}{bx}{it}
\DeclareMathAlphabet{\mathsf}    {OT1}{cmss}{m}{n}
\DeclareMathAlphabet{\mathsfbf}  {OT1}{cmss}{bx}{n}
\DeclareMathAlphabet{\mathsfit}  {OT1}{cmss}{m}{sl}
\DeclareMathAlphabet{\mathtt}    {OT1}{cmtt}{m}{n}
\DeclareMathAlphabet{\mathttbf}  {OT1}{cmtt}{bx}{n}
\DeclareMathAlphabet{\mathttit}  {OT1}{cmtt}{m}{it}
\DeclareMathAlphabet{\mathpzc}   {OT1}{pzc}{m}{it}
\newcommand{\keywords}[1]{\par\addvspace\baselineskip\noindent\enspace\ignorespaces{\bfseries Keywords:\,}#1}
\newcommand{\comment}[1]{}
\begin{document}

\pagestyle{headings}
\title{The {\ttfamily ERA} of {\ttfamily FOLE}: Superstructure} 
\titlerunning{{\ttfamily FOLE-ERA} Superstructure}  
\author{Robert E. Kent}
\institute{Ontologos}
\maketitle

\begin{abstract}
This paper discusses the representation of ontologies
in the first-order logical environment {\ttfamily FOLE}.
An ontology defines the primitives 
with which to model the knowledge resources for a community of discourse.
These primitives consist of classes, relationships and properties.
An ontology uses formal axioms to constrain the interpretation of these primitives. 
In short, an ontology specifies a logical theory.
%
%
This paper continues the discussion 
of the representation 
and interpretation
of ontologies
in the first-order logical environment {\ttfamily FOLE}.
%
%
%
The formalism and semantics of (many-sorted) first-order logic 
can be developed in both
a \emph{classification form}
and
an 
\emph{interpretation form}.
%
%
Two papers,
``The {\ttfamily ERA} of {\ttfamily FOLE}: Foundation'',
defining the concept of a structure,
and 
the current paper,
defining the concept of a sound logic,
represent 
the 
\emph{classification form},
corresponding to ideas discussed in the ``Information Flow Framework''.
Two papers, 
``The {\ttfamily FOLE} Table'',
defining the concept of a relational table,
and 
``The {\ttfamily FOLE} Database'',
defining the concept of a relational database,
represent 
the \emph{interpretation form},
expanding on material found in the paper 
``Database Semantics''.
%
%
%
%
Although 
the classification form 
follows the entity-relationship-attribute data model of Chen,
the interpretation form incorporates the relational data model of Codd.
A fifth paper
``{\ttfamily FOLE} Equivalence''
proves that
the classification form 
is equivalent to
the interpretation form.
%
In general,
the {\ttfamily FOLE} representation uses a conceptual structures approach,
that is completely compatible with 
the theory of institutions, 
formal concept analysis 
and information flow.
%
\end{abstract}

\tableofcontents


\include{intro}
\include{log-env}

\include{comps}

\include{conclu}

\appendix



\end{document}

%% file: intro.tex
\newpage
\section{Introduction}\label{sec:intro}

%
\subsection{Philosophy.}\label{sub:sec:phil}

Following the theory of general systems, 
an information system consists of a collection of interconnected parts called information resources and 
a collection of part-part relationships between pairs of information resources called constraints.
Formal information systems have specifications as their information resources.
Semantic information systems have logics as their information resources.
A formal information system has an underlying distributed system with languages as component parts 
(formalism flows along language links).
A semantic information system has an underlying distributed system with structures as component parts
(formalism flows along structure links).
Hence,
semantic information systems allow information flow over a semantic multiverse.

The paper ``System Consequence'' 
gave a general and abstract solution,
at the level of logical environments,
to the interoperation of information systems via the channel theory of information flow.
Since {\ttfamily FOLE} is a logical environment,
we can apply this approach to interoperability
for information systems based on first-order logic and relational databases.  
In this paper we show that 
formal {\ttfamily FOLE} systems interoperate in a general sense 
(since the context of {\ttfamily FOLE} languages has all sums), whereas
semantic {\ttfamily FOLE} systems interoperate in a restricted sense 
(since the context of {\ttfamily FOLE} structures has sums over fixed universes).
However,
we show that distributed databases in a semantic multiverse are interoperable 
when each defines a portal into a common universe.

The ideas of conservative extensions and modular information systems
can be formulated in terms of channels and system morphisms
at the general and abstract level of logical environments.
By illustrating these ideas
in the {\ttfamily FOLE} logical environment,
we capture the idea of modular federated databases.

\subsection{Knowledge Representation}

Many-sorted (multi-sorted) first-order predicate logic 
represents a community's ``universe of discourse'' as 
a heterogeneous collection of objects
by conceptually scaling
the universe according to types.
The \emph{relational model} (Codd~\cite{codd:90}) 
is an approach for the information management of a ``community of discourse''
\footnote{Examples include:
an academic discipline;
a commercial enterprise;
library science;
the legal profession;
etc.}
using the semantics and formalism of (many-sorted) first-order predicate logic. 
%
The relational model was initially discussed in two papers:
``A Relational Model of Data for Large Shared Data Banks''
by Codd \cite{codd:70} 
and
``The Entity-Relationship Model -- Toward a Unified View of Data'' 
by Chen \cite{chen:76}. 
The relational model follows many-sorted logic
by representing data in terms of many-sorted relations, 
subsets of the Cartesian product of multiple domains. 
All data is represented horizontally in terms of tuples, 
which are grouped vertically into relations. 
A database organized in terms of the relational model 
is a called relational database.
The relational model provides a method 
for modeling the data stored in a relational database 
and for defining queries upon it. 
%

\subsection{First Order Logical Environment}\label{sub:sec:arch}


\paragraph{Basics.}

The \emph{first-order logical environment} \texttt{FOLE}
is a category-theoretic representation for 
many-sorted (multi-sorted) first-order predicate logic. 
%
\footnote{Following the original discussion of {\ttfamily FOLE} (Kent~\cite{kent:iccs2013}), 
we use 
the term \emph{mathematical context} for the concept of a category,
the term \emph{passage} for the concept of a functor, and
the term \emph{bridge} for the concept of a natural transformation.
A context represents some ``species of mathematical structure''. 
A passage is a ``natural construction on structures of one species, 
yielding structures of another species'' 
(Goguen \cite{goguen:cm91}).}
%
The relational model can naturally be represented in \texttt{FOLE}.
The {\texttt{FOLE}} approach to logic, 
and hence to databases, 
relies upon two mathematical concepts:
(1) lists and (2) classifications.
Lists represent database signatures and tuples;
classifications represent data-types and logical predicates.
{\texttt{FOLE}} 
represents the header of a database table as a list of sorts, and
represents the body of a database table as a set of tuples 
classified by the header.
The notion of a list is common in category theory.
The notion of a classification is described in two books:
``Information Flow: The Logic of Distributed Systems''
by Barwise and Seligman \cite{barwise:seligman:97} and
''Formal Concept Analysis: Mathematical Foundations''
by Ganter and Wille \cite{ganter:wille:99}.


\paragraph{Architecture.}

%


A series of papers provides a rigorous mathematical basis for {\ttfamily FOLE} 
by defining 
an architectural semantics for the relational data model,
thus providing the foundation for
the formalism and semantics of first-order logical/relational database systems.
This architecture 
consists of two hierarchies of two nodes each:
the classification hierarchy
and
the interpretation hierarchy.

%

%

\begin{itemize}
%
\item
Two papers provide a precise mathematical basis for \texttt{FOLE} classification.
The paper 
``The {\ttfamily ERA} of {\ttfamily FOLE}: Foundation''
\cite{kent:fole:era:found}
develops the notion of a \texttt{FOLE} \underline{\emph{structure}},
following the entity-relationship model of Chen~\cite{chen:76}.
This provides a basis for
the current paper 
``The {\ttfamily ERA} of {\ttfamily FOLE}: Superstructure''
\cite{kent:fole:era:supstruc},
which develops the notion of a \texttt{FOLE} \underline{\emph{sound logic}}.
\newline
\item
Two papers provide a precise mathematical basis for \texttt{FOLE} interpretation.
Both of these papers expand on material found in the paper 
``Database Semantics''
\cite{kent:db:sem}.
The paper 
``The {\ttfamily FOLE} Table''
\cite{kent:fole:era:tbl},
develops the notion of a \texttt{FOLE} \underline{\emph{table}}
following the relational model of Codd~\cite{codd:90}.
This provided a basis for the paper 
``The {\ttfamily FOLE} Database''
\cite{kent:fole:era:db},
which develops the notion of a \texttt{FOLE} relational \underline{\emph{database}}.
\end{itemize}
%


%
\begin{flushleft}
{{\setlength{\extrarowheight}{1.6pt}
{{
{\begin{tabular}[t]{l@{\hspace{20pt}}l}
{{{\begin{minipage}{230pt}
The architecture of
{\ttfamily FOLE}
is pictured briefly
on the right
and more completely in
Fig.\,1
of the preface of
the paper
\cite{kent:fole:equiv}.
This consists of two hierarchies of two nodes each.
\comment{
\\
\vspace{-16pt}
\begin{description}
\item[{The classification hierarchy}]
on the left
defines
{\ttfamily FOLE} Structures
\cite{kent:fole:era:found}
at the bottom
and {\ttfamily FOLE} Sound Logics
\cite{kent:fole:era:supstruc}
at the top.
%
\item[\emph{The interpretation hierarchy}]
on the right
defines
{\ttfamily FOLE} Tables \cite{kent:fole:era:tbl}
at the bottom
and {\ttfamily FOLE} Databases 
(this paper)
at the top.
%
\end{description}
\vspace{-5pt}
}
The paper
``{\ttfamily FOLE} Equivalence''
\cite{kent:fole:equiv}
proves that
{\ttfamily FOLE} sound logics
are equivalent to
{\ttfamily FOLE} databases.
%
\end{minipage}}}}
&
{{\begin{tabular}{c@{\hspace{5pt}}}
\setlength{\unitlength}{0.36pt}
\begin{picture}(180,120)(-50,-20)
\put(60.5,80.5){\makebox(0,0){\tiny{$\equiv$}}}
%
\put(-60,93){\makebox(0,0){\tiny{\tt{Relational}}}}
\put(-60,73){\makebox(0,0){\tiny{\tt{Calculus}}}}
\put(180,10){\makebox(0,0){\tiny{\sf{Relational}}}}
\put(180,-10){\makebox(0,0){\tiny{\sf{Algebra}}}}
\qbezier(100,87)(60,97)(20,87)
\put(20,87){\vector(-4,-1){0}}
\qbezier(20,75)(60,65)(100,75)
\put(100,75){\vector(4,1){0}}
\put(0.3,80){\makebox(0,0){\huge{$\bullet$}}}
\put(120.3,80){\makebox(0,0){\huge{$\circ$}}}
\put(1,10){\makebox(0,0){\huge{$\circ$}}}
\put(122,10){\makebox(0,0){\huge{$\circ$}}}
%
\put(0,68){\line(0,-1){40}}
\put(120,68){\line(0,-1){40}}
\put(46,-15){\scriptsize{\ttfamily FOLE}}
\put(22,-35){\scriptsize{\textsf{architecture}}}
\end{picture}
\end{tabular}}}
\end{tabular}}
}}}}
\end{flushleft}
In the relational model there are two approaches for database management:
the relational algebra,
which defines an imperative language,
and
the relational calculus, which defines a declarative language.
The paper 
``Relational Operations in \texttt{FOLE}''
\cite{kent:fole:rel:ops}
represents relational algebra
by
expressing the relational operations of database theory in a clear and implementable representation.
The relational calculus
will be represented in \texttt{FOLE} in a future paper.

\subsection{Overview}
\label{sub:sec:overvu}

The first-order logical environment {\ttfamily FOLE} (Kent~\cite{kent:iccs2013}) is a framework for 
defining the semantics and formalism of logic and databases in an integrated and coherent fashion.
Institutions in general, and logical environments in particular, 
give equivalent heterogeneous and homogeneous representations for logical systems.
{\ttfamily FOLE} is an institution, 
since ``satisfaction is invariant under change of notation".
{\ttfamily FOLE} is a logical environment, 
since ``satisfaction respects structure linkage''.
As an institution,
the architecture of {\ttfamily FOLE} consists of
languages as indexing components,
structures to represent semantic content,
specifications to represent formal content, and
logics to combine formalism with semantics.
{\ttfamily FOLE} structures are interpreted as relational/logical databases.

This paper,
which is concerned with the classification form of {\ttfamily FOLE} 
(see Fig.\,\ref{fbr:ctx}),
is presented in two parts:
the logical environment 
and
the architecture. 
%
\S\,\ref{sec:intro}
is an introduction,
which gives a brief discussion
of the 
philosophy,
knowledge representation,
basics, and 
architecture of {\ttfamily FOLE}.
%
\S\,\ref{sub:sec:log:env} discusses
the {\ttfamily FOLE} logical environment,
where
we define formulas, sequents, constraints; 
we extend interpretation and classification from entity types to formulas;
we define satisfaction for sequents and constraints; and
we show that {\ttfamily FOLE} is an institution and logical environment.
%
\S\,\ref{sub:sec:arch:comps}
develops 
the {\ttfamily FOLE} architecture,
where
we define the architectural components of specifications and logics
by developing the logical notions of entailment, consequence, residuation and soundness.
%
\S\,\ref{sec:conclu}
gives 
the conclusion and future work.
Table\,\ref{tbl:fig:tab}
lists the figures and tables in this paper.
\begin{table}
\begin{center}
{{\scriptsize{\setlength{\extrarowheight}{1.6pt}{
\begin{tabular}{l@{\hspace{20pt}}l}
\\
{\fbox{
\begin{tabular}[t]{|l@{\hspace{8pt}}l@{\hspace{2pt}:\hspace{5pt}}l|}
\hline
\S\ref{sub:sub:sec:log:obj}
& 
Fig.~\ref{fig:log:ord}
&
Logic Order
\\\cline{1-1}
\S\ref{sub:sub:sec:snd:log:flow}
& 
Fig.~\ref{fbr:ctx}
&
{\ttfamily FOLE} Superstructure
\\\hline
\end{tabular}}}
&
{\fbox{
\begin{tabular}[t]{|l@{\hspace{8pt}}l@{\hspace{2pt}:\hspace{5pt}}l|}
\hline
\S\ref{sub:sec:overvu}
& 
Tbl.~\ref{tbl:fig:tab}
&
Figures and Tables
\\\hline\hline
\S\ref{sub:sub:sec:fml} 
& 
Tbl.~\ref{tbl:syn:flow}
&
Syntactic Flow
\\
& 
Tbl.~\ref{tbl:fmla:fn}
&
Formula Function
\\\cline{1-1}
\S\ref{sub:sub:sec:constr}
& 
Tbl.~\ref{tbl:axioms}
&
Axioms
\\\hline\hline
\S\ref{sub:sub:sec:fmla:struc}  
& 
Tbl.~\ref{tbl:sem:flow}
&
Semantic Flow
\\\cline{1-1}
\S\ref{sub:sub:sec:fml:interp}
& 
Tbl.~\ref{tbl:fml:int}
&
Formula Interpretation
\\
& 
Tbl.~\ref{tbl:fml-sem:refl}
&
Formal/Semantics Reflection
\\\cline{1-1}
\S\ref{sub:sub:sec:fml:struc:obj}
& 
Tbl.~\ref{tbl:fmla:cls}
&
Formula Classification
\\\hline
\end{tabular}}}
\\ \\
\end{tabular}}}}}
\end{center}
\caption{Figures and Tables}
\label{tbl:fig:tab}
\end{table}

%% file: log-env.tex
\newpage
\section{Logical Environment}\label{sub:sec:log:env}

\subsection{Formalism.}\label{sub:sub:sec:fmlism}

\subsubsection{Formulas.}\label{sub:sub:sec:fml}
\footnote{We use concepts and notations presented in the {\ttfamily FOLE} foundation paper (Kent~\cite{kent:fole:era:found}).}
%

Let $\mathcal{S} = {\langle{R,\sigma,X}\rangle}$ be a fixed schema with 
a set of entity types $R$,
a set of sorts (attribute types) $X$
and a signature function
$R\xrightarrow{\sigma}\mathrmbf{List}(X)$.
The set of entity types $R$ is partitioned
$R = \bigcup_{{\langle{I,s}\rangle}\in\mathrmbf{List}(X)}R(I,s)$
into fibers,
where
$R(I,s){\;\subseteq\;}R$
is the fiber (subset) of all entity types with signature ${\langle{I,s}\rangle}$.
These are called ${\langle{I,s}\rangle}$-ary entity types. 
\footnote{This is a slight misnomer, 
since the signature of $r$ is $\sigma(R)={\langle{I,s}\rangle}$,
whereas the arity of $r$ is $\alpha(R)=I$.}
Here,
we follow the tuple, domain, and relation calculi from database theory,
using logical operations to extend the set of basic entity types $R$
to a set of defined entity types $\widehat{R}$ called formulas or queries.
%
%
%

%
Formulas, which are defined entity types corresponding to queries, 
are constructed by using logical connectives within a fiber and logical flow 
along signature morphisms between fibers (Tbl.~\ref{tbl:syn:flow}).
\footnote{An $\mathcal{S}$-signature morphism ${\langle{I',s'}\rangle}\xrightarrow{h}{\langle{I,s}\rangle}$
in $\mathrmbf{List}(X)$
is an arity function $I'\xrightarrow{\,h\,}I$ that preserves signature $s' = h{\,\cdot\,}s$.}
\footnote{The full version of \texttt{FOLE} 
(Kent~\cite{kent:iccs2013})
defines syntactic flow along term vectors.}
Logical connectives on formulas
express intuitive notions of natural language operations on the interpretation (extent, view) of formulas.
These connectives include:
conjunction, disjunction, negation, implication, etc.
For any signature ${\langle{I,s}\rangle}$,
let $\widehat{R}(I,s) \subseteq \widehat{R}$ 
denote the set of all formulas with this signature.
There are called ${\langle{I,s}\rangle}$-ary formulas.
The set of $\mathcal{S}$-formulas is partitioned as
$\widehat{R} = \bigcup_{{\langle{I,s}\rangle}\in\mathrmbf{List}(X)}\widehat{R}(I,s)$.
\begin{flushleft}
{{\footnotesize{\begin{minipage}{345pt}
\begin{itemize}
\item[\textsf{fiber:}]
Let ${\langle{I,s}\rangle}$ be any signature.
Any ${\langle{I,s}\rangle}$-ary entity type (relation symbol) is an 
${\langle{I,s}\rangle}$-ary formula;
that is, $R(I,s) \subseteq \widehat{R}(I,s)$.
For a pair of ${\langle{I,s}\rangle}$-ary formulas $\varphi$ and $\psi$, 
there are the following ${\langle{I,s}\rangle}$-ary formulas:
meet $(\varphi{\,\wedge\,}\psi)$, join $(\varphi{\,\vee\,}\psi)$,
implication $(\varphi{\,\rightarrowtriangle\,}\psi)$ and 
difference $(\varphi{\,\setminus\,}\psi)$. 
For ${\langle{I,s}\rangle}$-ary formula $\varphi$,
there is an ${\langle{I,s}\rangle}$-ary negation formula $\neg\varphi$.
There are top/bottom ${\langle{I,s}\rangle}$-ary formulas $\top_{{\langle{I,s}\rangle}}$ and $\bot_{{\langle{I,s}\rangle}}$.
\item[\textsf{flow:}]
Let ${\langle{I',s'}\rangle}\xrightarrow{h}{\langle{I,s}\rangle}$ be any signature morphism.
For ${\langle{I,s}\rangle}$-ary formula $\varphi$,
there are ${\langle{I',s'}\rangle}$-ary existentially/universally quantified formulas
${\scriptstyle\sum}_{h}(\varphi)$ and ${\scriptstyle\prod}_{h}(\varphi)$.
\footnote{For any index $i \in I$,
quantification for the complement inclusion signature function
${\langle{I\setminus\{i\},s'}\rangle}\xrightarrow{\text{inc}_{i}}{\langle{I,s}\rangle}$
gives the traditional syntactic quantifiers
$\forall_{i}\varphi,\exists_{i}{\varphi}$.}
For an ${\langle{I',s'}\rangle}$-ary formula $\varphi'$,
there is an ${\langle{I,s}\rangle}$-ary substitution formula
${h}^{\ast}(\varphi') = \varphi'(t)$.
\end{itemize}
\end{minipage}}}}
\end{flushleft}
\begin{table}[h]
\begin{center}
{{\footnotesize{\setlength{\extrarowheight}{2pt}{$\begin{array}{c}
{\langle{I',s'}\rangle}\xrightarrow{h}{\langle{I,s}\rangle}
\\
{\setlength{\unitlength}{0.6pt}\begin{picture}(120,60)(0,-25)
\put(0,0){\makebox(0,0){\footnotesize{$\widehat{R}(I',s')$}}}
\put(120,0){\makebox(0,0){\footnotesize{$\widehat{R}(I,s)$}}}
\put(60,20){\makebox(0,0){\scriptsize{${\scriptstyle\sum}_{{h}}$}}}
\put(62,2){\makebox(0,0){\scriptsize{${h}^{\ast}$}}}
\put(60,-22){\makebox(0,0){\scriptsize{${\scriptstyle\prod}_{{h}}$}}}
\put(85,12){\vector(-1,0){50}}
\qbezier(35,0)(43,0)(51,0)\qbezier(69,0)(77,0)(85,0)\put(85,0){\vector(1,0){0}}
\put(85,-12){\vector(-1,0){50}}
\end{picture}}
\end{array}$}}}}
\end{center}
\caption{Syntactic Flow}
\label{tbl:syn:flow}
\end{table}

\comment{
\begin{center}
{\fbox{\begin{tabular}{c}
\setlength{\unitlength}{0.78pt}
\begin{picture}(120,10)(0,-10)
\put(0,0){\makebox(0,0){\footnotesize{$\mathrmbfit{ext}_{\widehat{\mathcal{E}}}(I',s')$}}}
\put(120,0){\makebox(0,0){\footnotesize{$\mathrmbfit{ext}_{\widehat{\mathcal{E}}}(I,s)$}}}
\put(60,20){\makebox(0,0){\scriptsize{${\scriptstyle\sum}_{{h}}$}}}
\put(62,2){\makebox(0,0){\scriptsize{${h}^{\ast}$}}}
\put(60,-22){\makebox(0,0){\scriptsize{${\scriptstyle\prod}_{{h}}$}}}
\put(85,12){\vector(-1,0){50}}
\qbezier(35,0)(43,0)(51,0)\qbezier(69,0)(77,0)(85,0)\put(85,0){\vector(1,0){0}}
\put(85,-12){\vector(-1,0){50}}
\end{picture}
\end{tabular}}}
\end{center}
}

In general,
we regard formulas to be constructed entities or queries (defining views and interpretations; i.e., relations/tables),
not assertions.
Contrast this with the use of ``asserted formulas'' below.
For example, 
in a corporation data model the conjunction
$(\mathtt{Salaried}{\,\wedge\,}\mathtt{Married})$
is not an assertion, but a constructed entity type or query that defines the view
``salaried employees that are married''.
Formulas form a schema 
$\mathrmbfit{fmla}(\mathcal{S}) = {\langle{\widehat{R},\widehat{\sigma},X}\rangle}$ 
that extends $\mathcal{S}$
with $\mathcal{S}$-formulas as entity types:
with the inductive definitions above,
the set of entity types 
is extended
to a set of logical formulas
$R\xhookrightarrow{\mathrmit{inc}_{\mathcal{S}}}\widehat{R}$,
and
the entity type signature function 
is extended
to a formula signature function $\widehat{R}\xrightarrow{\;\widehat{\sigma}\;}\mathrmbf{List}(X)$
with
$\sigma = \mathrmit{inc}_{\mathcal{S}}{\;\cdot\;}\widehat{\sigma}$.

A schema morphism 
$\mathcal{S}_{2}\xRightarrow{{\langle{r,f}\rangle}}\mathcal{S}_{1}$
can be extended to a formula schema morphism 
$\mathrmbfit{fmla}(r,f) = {\langle{\hat{r},f}\rangle} :
\mathrmbfit{fmla}(\mathcal{S}_{2}) = {\langle{\widehat{R}_{2},\hat{\sigma}_{2},X_{2}}\rangle} \Longrightarrow
{\langle{\widehat{R}_{1},\hat{\sigma}_{1},X_{1}}\rangle} = \mathrmbfit{fmla}(\mathcal{S}_{1})$.
The formula function $\hat{r} : \widehat{R}_{2} \rightarrow \widehat{R}_{1}$,
which satisfies the condition
$\mathrmit{inc}_{\mathcal{S}_{2}}{\;\cdot\;}\hat{r} = r{\;\cdot\;}\mathrmit{inc}_{\mathcal{S}_{1}}$,
is recursively defined in Tbl.~\ref{tbl:fmla:fn}.
We can show, 
by induction on source formulas $\varphi_{2}\in\widehat{R}_{2}$,
that signatures are preserved
$\widehat{r}{\,\cdot\,}\widehat{\sigma}_{1} = \widehat{\sigma}_{2}{\,\cdot\,}{\scriptstyle\sum}_{f}$.
\footnote{This translation is the formal part of an ``interpretation in first-order logic'' 
(Barwise and Seligman~\cite{barwise:seligman:97}). 
The semantic part is the fiber passage of structures
$\mathrmbfit{struc}^{\curlywedge}_{{\langle{r,f}\rangle}} : \mathrmbf{Struc}(\mathcal{S}_{2})\leftarrow\mathrmbf{Struc}(\mathcal{S}_{1})$
along the schema morphism $\mathcal{S}_{2}\xRightarrow{{\langle{r,f}\rangle}}\mathcal{S}_{1}$
(see the appendix of Kent~\cite{kent:fole:era:found}).
The precise meaning of ``interpretation in first-order logic'',
as an infomorphism between truth classifications,
is given below by Prop.~\ref{prop:ins} on institutions and Prop.~\ref{prop:log:env} on logical environments.}
\begin{table}
\begin{center}
{\footnotesize{\setlength{\extrarowheight}{2pt}\begin{tabular}{|r@{\hspace{20pt}}l@{\hspace{10pt}$=$\hspace{10pt}}l|}
\multicolumn{3}{l}{\rule[-6pt]{0pt}{12pt}\textsf{fiber:} signature ${\langle{I_{2},s_{2}}\rangle}$}
\\ \hline
\textit{operator} & \multicolumn{1}{l}{} & 
\\
entity type
& $\hat{r}(r_{2})$
& $r(r_{2})$
\\
meet
& $\hat{r}(\varphi_{2}{\,\wedge_{{\langle{I_{2},s_{2}}\rangle}}\,}\psi_{2})$
& $\hat{r}(\varphi_{2}){\,\wedge_{{\scriptscriptstyle\sum}_{f}(I_{2},s_{2})}\,}\hat{r}(\psi_{2})$
\\
join
& $\hat{r}(\varphi_{2}{\,\vee_{{\langle{I_{2},s_{2}}\rangle}}\,}\psi_{2})$
& $\hat{r}(\varphi_{2}){\,\vee_{{\scriptscriptstyle\sum}_{f}(I_{2},s_{2})}\,}\hat{r}(\psi_{2})$
\\
negation
& $\hat{r}(\neg_{{\langle{I_{2},s_{2}}\rangle}}\,\varphi)$
& $\neg_{{\scriptscriptstyle\sum}_{f}(I_{2},s_{2})}\,\hat{r}(\varphi)$
\\
implication
& $\hat{r}(\varphi{\,\rightarrowtriangle_{{\langle{I_{2},s_{2}}\rangle}}\,}\psi)$
& $\hat{r}(\varphi){\,\rightarrowtriangle_{{\scriptscriptstyle\sum}_{f}(I_{2},s_{2})}\,}\hat{r}(\psi)$
\\
difference
& $\hat{r}(\varphi{\,\setminus_{{\langle{I_{2},s_{2}}\rangle}}\,}\psi)$
& $\hat{r}(\varphi){\,\setminus_{{\scriptscriptstyle\sum}_{f}(I_{2},s_{2})}\,}\hat{r}(\psi)$
\\ \hline
\multicolumn{3}{l}{}
\\
\multicolumn{3}{l}{\rule[-6pt]{0pt}{12pt}\textsf{flow:} signature morphism 
${\langle{I_{2}',s_{2}'}\rangle} \xrightarrow{h} {\langle{I_{2},s_{2}}\rangle}$}
\\ \hline
\textit{operator} & \multicolumn{1}{l}{} & 
\\
existential
& $\hat{r}({\scriptstyle\sum}_{h}(\varphi_{2}))$
& ${\scriptstyle\sum}_{h}(\hat{r}(\varphi_{2}))$ 
\\
universal
& $\hat{r}({\scriptstyle\prod}_{h}(\varphi_{2}))$
& ${\scriptstyle\prod}_{h}(\hat{r}(\varphi_{2}))$ 
\\
substitution
& $\hat{r}({h}^{\ast}(\varphi_{2}'))$
& ${h}^{\ast}(\hat{r}(\varphi_{2}'))$ 
\\ \hline
\end{tabular}}}
\end{center}
\caption{Formula Function}
\label{tbl:fmla:fn}
\end{table}
Hence,
there is an idempotent formula passage
$\mathrmbf{Sch}\xrightarrow{\mathrmbfit{fmla}}\mathrmbf{Sch}$
on schemas.

\subsubsection{Sequents.}\label{sub:sub:sec:seq}

To make an assertion about things, 
we use a sequent.
Let $\mathcal{S} = {\langle{R,\sigma,X}\rangle}$ be a schema.
A (binary) $\mathcal{S}$-sequent 
\footnote{Since {\ttfamily FOLE} formulas are not just types,
but are constructed using, 
inter alia, 
conjunction and disjunction operations, 
we can restrict attention to \emph{binary} sequents.} 
is a pair of formulas
$(\varphi,\psi){\,\in\,}\widehat{R}{\,\times}\widehat{R}$
with the same signature 
$\widehat{\sigma}(\varphi) = {\langle{I,s}\rangle} = \widehat{\sigma}(\psi)$. 
To be explicit that we are making an assertion,
we use the turnstile notation $\varphi{\;\vdash\;}\psi$ for a sequent.
Then, 
we are claiming that a specialization-generalization relationship exists between the formulas $\varphi$ and $\psi$.
A asserted sequent $\varphi{\;\vdash\;}\psi$ expresses interpretation widening,
with the interpretation (view) of $\varphi$ required to be within the interpretation (view) of $\psi$.
%
%
%
An \emph{asserted} formula $\varphi\in\widehat{R}$ 
can be identified with the sequent $\top{\;\vdash\;}\varphi$ in $\widehat{R}{\,\times}\widehat{R}$,
which asserts the universal view ``all entities'' of signature $\sigma(\varphi)={\langle{I,s}\rangle}$.
%
Hence,
from an entailment viewpoint we can say that ``formulas are sequents''.
In the opposite direction,
there is an enfolding map
$\widehat{R}{\,\times}\widehat{R}\rightarrow\widehat{R}$
that maps $\mathcal{S}$-sequents to 
$\mathcal{S}$-formulas
$({\varphi}{\;\vdash\;}{\psi})\mapsto(\varphi{\,\rightarrowtriangle\,}\psi)$.
%
The axioms (Tbl.~\ref{tbl:axioms}) make sequents into an order.
Let $\mathrmbf{Con}_{\mathcal{S}}(I,s)={\langle{\widehat{R}(I,s),\vdash}\rangle}$ denote the fiber preorder of $\mathcal{S}$-formulas.

\subsubsection{Constraints.}\label{sub:sub:sec:constr}

Sequents only connect formulas within a particular fiber: 
an $\mathcal{S}$-sequent $\varphi{\;\vdash\;}\psi$
is between two formulas with the same signature
$\widehat{\sigma}(\varphi) = {\langle{I,s}\rangle} = \widehat{\sigma}(\psi)$,
and hence between elements in the same fiber $\varphi,\psi{\;\in\;}\widehat{R}(I,s)$. 
We now define a useful notion that connects formulas across fibers.
An $\mathcal{S}$-constraint 
$\varphi'{\,\xrightarrow{h\,}\,}\varphi$
consists of a signature morphism
$\widehat{\sigma}(\varphi')={\langle{I',s'}\rangle}\xrightarrow{h}{\langle{I,s}\rangle}=\widehat{\sigma}(\varphi)$
and 
a binary sequent ${\scriptstyle\sum}_{h}(\varphi){\;\vdash\;}\varphi'$ in $\mathrmbf{Con}_{\mathcal{S}}(I',s')$,
or equivalently by the axioms of Tbl.~\ref{tbl:axioms},
a binary sequent $\varphi{\;\vdash\;}{h}^{\ast}(\varphi')$ in $\mathrmbf{Con}_{\mathcal{S}}(I,s)$.
Hence,
a constraint requires 
that the interpretation of the $h^{\mathrm{th}}$-projection of $\varphi$ be within the interpretation of $\varphi'$,
or equivalently,
that the interpretation of $\varphi$ be within the interpretation of the $h^{\mathrm{th}}$-substitution of $\varphi'$. 
\footnote{In some sense,
this formula/constraint approach to formalism turns the tuple calculus upside down,
with atoms in the tuple calculus becoming constraints here.}

Given any schema $\mathcal{S}$,
an $\mathcal{S}$-constraint ${\varphi'}\xrightarrow{h}{\varphi}$
has source formula ${\varphi'}$ and target formula ${\varphi}$.
Constraints are closed under composition:
${\varphi''}\xrightarrow{h'}{\varphi'}\xrightarrow{h}{\varphi}
={\varphi''}\xrightarrow{h'{\,\cdot\,}h}{\varphi}$.
Let $\mathrmbf{Cons}(\mathcal{S})$ denote the mathematical context, 
whose set of objects are $\mathcal{S}$-formulas 
and whose set of morphisms are $\mathcal{S}$-constraints.
This context is fibered over the projection passage
$\mathrmbf{Cons}(\mathcal{S})\rightarrow\mathrmbf{List}(X)
:\bigl(\varphi'{\,\xrightarrow{h\,}\,}\varphi\bigr)
\mapsto
\bigl({\langle{I',s'}\rangle}\xrightarrow{h}{\langle{I,s}\rangle}\bigr)$.

Sequents are special cases of constraints:
a sequent $\varphi'{\;\vdash\;}\varphi$
asserts a constraint ${\varphi}\xrightarrow{1}{\varphi'}$
that uses an identity signature morphism.
\footnote{A constraint in the fiber $\mathrmbf{Con}_{\mathcal{S}}(I,s)$
uses an identity signature morphism ${\varphi}\xrightarrow{1}{\varphi'}$,
and hence is a sequent $\varphi'{\;\vdash\;}\varphi$.}
Since an asserted formula $\varphi$ 
can be identified with the sequent $\top{\;\vdash\;}\varphi$,
it can also be identified with the constraint ${\varphi}\xrightarrow{1}{\top}$.
Thus,
from an entailment viewpoint we can say that
``formulas are sequents are constraints''.
In the opposite direction,
there are enfolding maps
that map $\mathcal{S}$-constraints to $\mathcal{S}$-formulas:
either
$\bigl({\varphi'}\xrightarrow{h}{\varphi}\bigr)\mapsto\bigl({\scriptstyle\sum}_{h}(\varphi){\,\rightarrowtriangle\,}\varphi'\bigr)$ 
with signature ${\langle{I',s'}\rangle}$, 
or
$\bigl({\varphi'}\xrightarrow{h}{\varphi}\bigr)\mapsto\bigl(\varphi{\,\rightarrowtriangle\,}{h}^{\ast}(\varphi')\bigr)$ 
with signature ${\langle{I,s}\rangle}$.

\begin{sloppypar}
Given any schema morphism $\mathcal{S}_{2}\stackrel{{\langle{r,f}\rangle}}{\Longrightarrow}\mathcal{S}_{1}$,
there is a constraint passage
$\mathrmbf{Cons}(\mathcal{S}_{2})\xrightarrow{\mathrmbfit{cons}_{{\langle{r,f}\rangle}}}\mathrmbf{Cons}(\mathcal{S}_{1})$.
%
An $\mathcal{S}_{2}$-formula $\varphi_{2}\in\mathrmbfit{fmla}(\mathcal{S}_{2})$
is mapped to 
the $\mathcal{S}_{1}$-formula $\widehat{r}(\varphi_{2})\in\mathrmbfit{fmla}(\mathcal{S}_{1})$.
An $\mathcal{S}_{2}$-constraint $\varphi_{2}'\xrightarrow{h}\varphi_{2}$ in ${\mathrmbf{Cons}}(\mathcal{S}_{2})$
with $\mathcal{S}_{2}$-enfolding $\varphi_{2}{\,\rightarrowtriangle\,}{h}^{\ast}(\varphi_{2}')$
in $\mathrmbf{Con}_{\mathcal{S}_{2}}(I_{2},s_{2})$
is mapped to 
the $\mathcal{S}_{1}$-constraint $\widehat{r}(\varphi_{2}')\xrightarrow{h}\widehat{r}(\varphi_{2})$ in ${\mathrmbf{Cons}}(\mathcal{S}_{1})$
with $\mathcal{S}_{1}$-enfolding 
$\widehat{r}(\varphi_{2}){\,\rightarrowtriangle\,}{h}^{\ast}(\widehat{r}(\varphi_{2}'))
=\widehat{r}(\varphi_{2}){\,\rightarrowtriangle\,}\widehat{r}({h}^{\ast}(\varphi_{2}'))
=\widehat{r}\bigl(\varphi_{2}{\,\rightarrowtriangle\,}{h}^{\ast}(\varphi_{2}')\bigr)$
in $\mathrmbf{Con}_{\mathcal{S}_{1}}({\scriptstyle\sum}_{f}(I_{2},s_{2}))$
using Tbl.~\ref{tbl:fmla:fn}.
The passage $\mathrmbf{Sch}{\;\xrightarrow{\mathrmbfit{cons}}\;}\mathrmbf{Cxt}$ forms an indexed context of constraints.
\end{sloppypar}

\begin{table}
\begin{center}
{\scriptsize{\setlength{\extrarowheight}{3.5pt}\begin{tabular}{|r@{\hspace{3pt}:\hspace{15pt}}l|}
\multicolumn{2}{l}{$\text{schema:}\;\mathcal{S}$} \\ 
\multicolumn{2}{l}{\rule[-6pt]{0pt}{16pt}\textsf{fiber:} signature ${\langle{I,s}\rangle}$}
\\ \hline
reflexivity 
&
$\varphi{\;\vdash\,}\varphi$
\\
transitivity
&
$\varphi{\;\vdash\;}\varphi'$ and $\varphi'{\;\vdash\;}\varphi''$ implies $\varphi{\;\vdash\;}\varphi''$
\\ \hline\hline
meet
&
$\psi{\;\vdash\;}(\varphi{\,\wedge\,}\varphi')$ 
iff 
$\psi{\;\vdash\;}\varphi$ 
and
$\psi{\;\vdash\;}\varphi'$ 
\\
\multicolumn{1}{|c}{}
&
$(\varphi{\,\wedge\,}\varphi'){\;\vdash\;}\varphi$, 
$(\varphi{\,\wedge\,}\varphi'){\;\vdash\;}\varphi'$ 
\\ \hline
join
&
$(\varphi{\,\vee\,}\varphi'){\;\vdash\;}\psi$ 
iff 
$\varphi{\;\vdash\;}\psi$ 
and
$\varphi'{\;\vdash\;}\psi$ 
\\
\multicolumn{1}{|c}{}
&
$\varphi'{\;\vdash\;}(\varphi{\,\vee\,}\varphi)$, 
$\varphi'{\;\vdash\;}(\varphi{\,\vee\,}\varphi')$ 
\\ \hline
implication
&
$(\varphi{\;\wedge\;}\varphi'){\;\vdash\;}\psi$
iff
$\varphi{\;\vdash\;}(\varphi'{\rightarrowtriangle\,}\psi)$
\\ \hline
negation
&
$\neg\,(\neg\,(\varphi)){\;\vdash\;}\varphi$ 
\\ \hline
top and bottom
&
$\bot_{{\langle{I,s}\rangle}}{\;\vdash\;}\varphi{\;\vdash\;}\top_{{\langle{I,s}\rangle}}$
\\ \hline
\multicolumn{2}{l}{\rule[-6pt]{0pt}{20pt}\textsf{flow:} signature morphism 
${\langle{I',s'}\rangle} \xrightarrow{h} {\langle{I,s}\rangle}$}
\\ \hline
${\scriptstyle\sum}_{h}$-monotonicity
&
$\varphi{\;\vdash\;}\psi$ 
implies
${\scriptstyle\sum}_{h}(\varphi){\;\vdash'\;}{\scriptstyle\sum}_{h}(\psi)$
\\
${h}^{\ast}$-monotonicity
&
$\varphi'{\;\vdash'\;}\psi'$ 
implies
${h}^{\ast}(\varphi'){\;\vdash\;}{h}^{\ast}(\psi')$
\\
${\scriptstyle\prod}_{h}$-monotonicity
&
$\varphi{\;\vdash\;}\psi$ 
implies
${\scriptstyle\prod}_{h}(\varphi){\;\vdash'\;}{\scriptstyle\prod}_{h}(\psi)$
\\ \hline\hline
adjointness
&
${\scriptstyle\sum}_{h}(\varphi){\;\vdash'\;}\varphi'$
iff
$\varphi{\;\vdash\;}{h}^{\ast}(\varphi')$
\\
\multicolumn{1}{|c}{}
&
$\varphi{\;\vdash\;}{h}^{\ast}({\scriptstyle\sum}_{h}(\varphi))$,
${\scriptstyle\sum}_{h}({h}^{\ast}(\varphi')){\;\vdash'\;}\varphi'$
\\ \hline
\multicolumn{2}{l}{\rule[-6pt]{0pt}{24pt}$\text{schema morphism:}\;\mathcal{S}_{2}\stackrel{{\langle{r,f}\rangle}}{\Longrightarrow}\mathcal{S}_{1}$} 
\\ \hline 
$\widehat{r}$-monotonicity
&
$\varphi_{2}{\,\vdash_{2}\,}\psi_{2}$ 
implies
$\widehat{r}(\varphi_{2}){\;\vdash_{1}\;}\widehat{r}(\psi_{2})$ 
\\ \hline
\end{tabular}}}
\end{center}
\caption{Axioms}
\label{tbl:axioms}
\end{table}
%

\newpage
\subsection{Semantics.}\label{sub:sub:sec:fmla:struc}  

For any structure $\mathcal{M} = {\langle{\mathcal{E},{\langle{\sigma,\tau}\rangle},\mathcal{A}}\rangle}$, 
the semantics of formulas involves both 
a formula interpretation function $\mathrmbfit{I}_{\mathcal{M}}$ defined in Tbl.~\ref{tbl:fml:int} and 
a formula classification $\widehat{\mathcal{E}}$ defined in Tbl.~\ref{tbl:fmla:cls}.
Formula interpretation is independently defined,
but formula classification depends upon formula interpretation.

\paragraph{Semantic Quantifiers.}

Both formula interpretation and formula classification
use semantic quantifiers (and substitution) in their definitions.
Here we give an intuitive expression for these.
Let $\mathcal{A}={\langle{X,Y,\models_{\mathcal{A}}}\rangle}$ be a type domain (attribute classification) and 
let $\mathcal{S} = {\langle{R,\sigma,X}\rangle}$ be a schema
with common sort set $X$. 
If ${\langle{I',s'}\rangle}\xrightarrow{h}{\langle{I,s}\rangle}$
is an $\mathcal{S}$-signature morphism
with the associated tuple function
$\mathrmbfit{tup}_{\mathcal{A}}(I',s')\xleftarrow{\mathrmbfit{tup}_{h}}\mathrmbfit{tup}_{\mathcal{A}}(I,s)$,
we have the following adjoint functions.
%
%
%
\footnote{Recall from the original and foundation papers on {\ttfamily FOLE}
(\cite{kent:iccs2013}, \cite{kent:fole:era:found})
the following definitions.
The set of $\mathcal{A}$-tuples with signature ${\langle{I,s}\rangle}$
is the $\mathrmbf{List}(\mathcal{A})$-extent 
$\mathrmbfit{tup}_{\mathcal{A}}(I,s) = \mathrmbfit{ext}_{\mathrmbf{List}(\mathcal{A})}(I,s)$.
The tuple function
$\mathrmbfit{tup}_{\mathcal{A}}(I',s')\xleftarrow{\mathrmbfit{tup}_{h}}\mathrmbfit{tup}_{\mathcal{A}}(I,s)$
maps ${\langle{I,t}\rangle}{\,\in\,}\mathrmbfit{tup}_{\mathcal{A}}(I,s)$
to ${\langle{I',h{\,\cdot\,}t}\rangle}{\,\in\,}\mathrmbfit{tup}_{\mathcal{A}}(I',s')$.
$\mathrmbf{Rel}_{\mathcal{A}}(I,s)={\wp}\mathrmbfit{tup}_{\mathcal{A}}(I,s)$
is the power set of $\mathcal{A}$-tuples with signature ${\langle{I,s}\rangle}$.}
\footnote{The semantic quantifiers (and substitution)
are used in the definition of the fibered context $\mathrmbf{Rel}(\mathcal{A})$,
which is defined 
in the paper 
on the {\ttfamily FOLE} Table
(\cite{kent:fole:era:tbl}).
An object of $\mathrmbf{Rel}(\mathcal{A})$,
called an $\mathcal{A}$-relation,
is a pair ${\langle{I,s,R}\rangle}$
consisting of an indexing $X$-signature ${\langle{I,s}\rangle}$
and a subset of $\mathcal{A}$-tuples 
$R \in {\wp}\mathrmbfit{tup}_{\mathcal{A}}(I,s) 
= {\wp}\mathrmbfit{ext}_{\mathrmbf{List}(\mathcal{A})}(I,s)$.}
\footnote{Since existential quantification and substitution
are direct/inverse operators 
{\footnotesize{$
\mathrmbf{Rel}_{\mathcal{A}}(I',s')
\xleftrightharpoons[{h}^{{\scriptscriptstyle-}1}\,=\,\mathrmbfit{tup}(h)^{-1}]{\exists_{h}\,=\,{\wp}\mathrmbfit{tup}(h)}
\mathrmbf{Rel}_{\mathcal{A}}(I,s)
$}\normalsize}
along the tuple function
$\mathrmbfit{tup}_{\mathcal{A}}(I',s')\xleftarrow{\mathrmbfit{tup}(h)}\mathrmbfit{tup}_{\mathcal{A}}(I,s)$,
they are clearly adjoint to each other
$\exists_{h}(R){\;\subseteq\;}R'$
\underline{iff}
$R{\;\subseteq\;}{h}^{{\scriptscriptstyle-}1}(R')$.}
%
\begin{table}
\begin{center}
{\footnotesize\setlength{\unitlength}{0.5pt}{$\begin{array}[t]{r@{\hspace{10pt}}r@{\hspace{5pt}}c@{\hspace{5pt}}l}
\text{signature morphism} & {\langle{I',s'}\rangle} & \xrightarrow{h} & {\langle{I,s}\rangle} \\
\text{tuple function} & \mathrmbfit{tup}_{\mathcal{A}}(I',s') & \xleftarrow{\mathrmbfit{tup}_{h}} & \mathrmbfit{tup}_{\mathcal{A}}(I,s) \\
\text{substitution} & \mathrmbf{Rel}_{\mathcal{A}}(I',s') & \xrightarrow{\;{h}^{{\scriptscriptstyle-}1}} & \mathrmbf{Rel}_{\mathcal{A}}(I,s) \\
\text{quantification} & \mathrmbf{Rel}_{\mathcal{A}}(I',s') & \xleftarrow[\;\forall_{h}\,]{\;\exists_{h}\,} & \mathrmbf{Rel}_{\mathcal{A}}(I,s) \\
\multicolumn{1}{l}{} & \multicolumn{3}{c}{\rule{0pt}{12pt}\;\;\;\;\;\exists_{{h}}{\;\dashv\;\;}{h}^{{\scriptscriptstyle-}1}{\dashv\;}\forall_{{h}}} 
\end{array}$}}
\end{center}
\caption{Semantic Flow}
\label{tbl:sem:flow}
\end{table}
\begin{description}
\item[Intuitive explanation:] 
For any tuple subset $R\in\mathrmbf{Rel}_{\mathcal{A}}(I,s)$, 
you can get two tuple subsets $\exists_{h}(R),\forall_{h}(R)\in\mathrmbf{Rel}_{\mathcal{A}}(I',s')$ as follows.  
Given any possible tuple $t'\in\mathrmbfit{tup}_{\mathcal{A}}(I',s')$,
you can ask either an existential or a universal question about it:
``Does there \emph{exist} a tuple $t \in R$ with image $t'$?'' ($t'=\mathrmbfit{tup}_{h}(t)$) or 
``Is it the case that \emph{all possible} tuples $t\in\mathrmbfit{tup}_{\mathcal{A}}(I,s)$ with image $t'$ are present in $R$?''
%
Clearly, the quantification/substitution operators are monotonic.
\end{description}
%

\newpage
\subsubsection{Formula Interpretation.}\label{sub:sub:sec:fml:interp}

\paragraph{Formula Interpretation.}

The formula interpretation function
%
%
\begin{equation}\label{eqn:fmla:interp}
\mathrmbfit{I}_{\mathcal{M}} : \widehat{R} \rightarrow \mathrmbf{Rel}(\mathcal{A}),
\end{equation}
which extends the traditional interpretation function 
$\mathrmbfit{I}_{\mathcal{M}} : R\rightarrow\mathrmbf{Rel}(\mathcal{A})$
(see the foundation paper on {\ttfamily FOLE} \cite{kent:fole:era:found}),
is defined by induction on formulas
in Tbl.~\ref{tbl:fml:int}.
\footnote{The 
function 
$\widehat{R}\xrightarrow{\mathrmbfit{I}_{\mathcal{M}}}\mathrmbf{Rel}_{\mathcal{A}}$
is the parallel combination of its fiber functions
\[\mbox{\footnotesize{$
\biggl\{ \widehat{R}(I,s)\xrightarrow{\mathrmbfit{I}^{\mathcal{M}}_{{\langle{I\!,s}\rangle}}}\mathrmbf{Rel}_{\mathcal{A}}(I,s)
\mid {\langle{I,s}\rangle}{\,\in\,}\mathrmbf{List}(X) \biggr\}
$.}\normalsize}\]
The definition of $\mathrmbfit{I}_{\mathcal{M}}$ is directly in terms of these fiber functions.}
At the base step,
it defines the formula interpretation of an entity type $r\in\widehat{R}$ 
as the traditional interpretation of that type
$\mathrmbfit{I}_{\mathcal{M}}(r)
={\wp}\tau(\mathrmbfit{ext}_{\mathcal{E}}(r))$,
which is 
the set of descriptors for entities of that type.
At the the induction step, 
it represents the logical operations by their associated boolean operations: 
intersection of interpretations for conjunction,
union of interpretations for disjunction,
etc.;
and it represents the syntactic flow operators in Tbl.~\ref{tbl:syn:flow} of \S\ref{sub:sub:sec:fmlism}
by their associated semantic flow operators in Tbl.~\ref{tbl:sem:flow}.
\begin{table}
\begin{center}
{\footnotesize{\setlength{\extrarowheight}{2pt}\begin{tabular}{|@{\hspace{5pt}}r@{\hspace{20pt}}l@{\hspace{5pt}}|}
\multicolumn{2}{l}{\rule[-6pt]{0pt}{1pt}\textsf{fiber:} signature ${\langle{I,s}\rangle}$ 
with extent (tuple set) 
$\mathrmbfit{tup}_{\mathcal{A}}(I,s)
{\,=\,}\prod_{i\in{I}}\,\mathcal{A}_{s_{\!i}}$}
\\ \hline
\textit{operator} 
&
\textit{definition}
\hfill
$\mathrmbfit{I}_{\mathcal{M}}(\varphi)\in\mathrmbf{Rel}_{\mathcal{A}}(I,s)={\wp}\mathrmbfit{tup}_{\mathcal{A}}(I,s)$
\\
entity type & 
$\mathrmbfit{I}_{\mathcal{M}}(r)
={\wp}\tau(\mathrmbfit{ext}_{\mathcal{E}}(r))$
\hfill
$r{\,\in\,}R(I,s)\subseteq\widehat{R}(I,s)$
\\
meet & 
$\mathrmbfit{I}_{\mathcal{M}}(\varphi{\,\wedge\,}\psi)
=\mathrmbfit{I}_{\mathcal{M}}(\varphi){\,\cap\,}\mathrmbfit{I}_{\mathcal{M}}(\psi)$
\hfill
$\varphi,\psi{\,\in\,}\widehat{R}(I,s)$
\\
join &
$\mathrmbfit{I}_{\mathcal{M}}(\varphi{\,\vee\,}\psi)
=\mathrmbfit{I}_{\mathcal{M}}(\varphi){\,\cup\,}\mathrmbfit{I}_{\mathcal{M}}(\psi)$
\\
top & 
$\mathrmbfit{I}_{\mathcal{M}}({\scriptstyle\top_{{\langle{I,s}\rangle}}})
=\mathrmbfit{tup}_{\mathcal{A}}(I,s)$
\\
bottom & 
$\mathrmbfit{I}_{\mathcal{M}}({\scriptstyle\bot_{{\langle{I,s}\rangle}}})
=\emptyset$
\\
negation & 
$\mathrmbfit{I}_{\mathcal{M}}(\neg\varphi)
=\neg\mathrmbfit{I}_{\mathcal{M}}(\varphi)
=
\mathrmbfit{tup}_{\mathcal{A}}(I,s){\,\setminus\,}\mathrmbfit{I}_{\mathcal{M}}(\varphi)$
\\
implication & 
$\mathrmbfit{I}_{\mathcal{M}}(\varphi{\,\rightarrowtriangle\,}\psi)
=\mathrmbfit{I}_{\mathcal{M}}(\varphi){\,\rightarrowtriangle\,}\mathrmbfit{I}_{\mathcal{M}}(\psi)
=(\neg\mathrmbfit{I}_{\mathcal{M}}(\varphi)){\,\cup\,}\mathrmbfit{I}_{\mathcal{M}}(\psi)$
\\
difference & 
$\mathrmbfit{I}_{\mathcal{M}}(\varphi{\,\setminus\,}\psi)
= \mathrmbfit{I}_{\mathcal{M}}(\varphi){\,\setminus\,}\mathrmbfit{I}_{\mathcal{M}}(\psi)
=\mathrmbfit{I}_{\mathcal{M}}(\varphi){\,\cap\,}(\neg\mathrmbfit{I}_{\mathcal{M}}(\psi))$
\rule[-5pt]{0pt}{10pt}
\\ \hline
\multicolumn{2}{l}{\rule{0pt}{20pt}\textsf{flow:} signature morphism
${\langle{I',s'}\rangle}\xrightarrow{h}{\langle{I,s}\rangle}$ 
}
\\
\multicolumn{2}{l}{\rule[-6pt]{0pt}{1pt}with tuple map
$\mathrmbfit{tup}_{\mathcal{A}}(I',s')\xleftarrow{\mathrmbfit{tup}_{\mathcal{A}}(h)}\mathrmbfit{tup}_{\mathcal{A}}(I,s)$}
\\ \hline
\textit{operator} & \multicolumn{1}{l|}{\textit{definition}}
\\
existential
& $\mathrmbfit{I}_{\mathcal{M}}({\scriptstyle\sum}_{h}(\varphi))
={\exists}_{h}(\mathrmbfit{I}_{\mathcal{M}}(\varphi))$
\hfill
$\varphi{\,\in\,}\widehat{R}(I,s)$
\\
universal
& $\mathrmbfit{I}_{\mathcal{M}}({\scriptstyle\prod}_{h}(\varphi))
={\forall}_{h}(\mathrmbfit{I}_{\mathcal{M}}(\varphi))$
\\
substitution
& 
$\mathrmbfit{I}_{\mathcal{M}}({h}^{\ast}(\varphi'))
={h}^{-1}(\mathrmbfit{I}_{\mathcal{M}}(\varphi'))$
\rule[-5pt]{0pt}{10pt}
\hfill
$\varphi'{\,\in\,}\widehat{R}(I',s')$
\\ \hline
\end{tabular}}}
\end{center}
\caption{Formula Interpretation}
\label{tbl:fml:int}
\end{table}
%

\newpage
\paragraph{Formal/Semantics Reflection.}

The logical semantics of a structure $\mathcal{M}$ resides in its core,
which is defined by its formula interpretation function 
$\widehat{R}\xrightarrow{\mathrmbfit{I}_{\mathcal{M}}}\mathrmbf{Rel}_{\mathcal{A}}$.
To respect this,
the formal flow operators 
(${\scriptstyle\sum}_{h}$,${\scriptstyle\prod}_{h}$,${h}^{\ast}$)
for existential/universal quantification and substitution
reflect the semantic flow operators
(${\exists}_{h}$,${\forall}_{h}$,${h}^{-1}$)
via interpretation
(Tbl.~\ref{tbl:fml-sem:refl}).
\footnote{This tabular/relational reflection 
is a special case of Prop.~5.59 in the text {\itshape Topos Theory}~\cite{johnstone:77},
suggesting that all of the development in this paper could be done in an arbitrary topos,
or even in a more general setting.}
\footnote{The morphic aspect of Tbl.~\ref{tbl:fml-sem:refl} anticipates the definition of sequent satisfaction in \S.\ref{sub:sub:sec:seq:sat}: 
An $\mathcal{S}$-structure $\mathcal{M} \in \mathrmbf{Struc}(\mathcal{S})$
\emph{satisfies} an $\mathcal{S}$-sequent $\varphi{\;\vdash\;}\psi$
when the interpretation widening of views asserted by the sequent actually holds in $\mathcal{M}$:
$\mathrmbfit{I}_{\mathcal{M}}(\varphi){\;\subseteq\;}\mathrmbfit{I}_{\mathcal{M}}(\psi)$.}
\begin{table}
\begin{center}
{\scriptsize{\begin{tabular}{@{\hspace{30pt}}c@{\hspace{50pt}}c}
{{\begin{tabular}[b]{c}
\setlength{\unitlength}{0.79pt}
\begin{picture}(120,90)(0,-10)
\put(-2,80){\makebox(0,0){\footnotesize{${\bigl\langle{\widehat{R}(I',s'),\vdash'}\bigr\rangle}$}}}
\put(120,80){\makebox(0,0){\footnotesize{${\bigl\langle{\widehat{R}(I,s),\vdash}\bigr\rangle}$}}}
\put(0,0){\makebox(0,0){\footnotesize{$\mathrmbf{Rel}_{\mathcal{A}}(I',s')$}}}
\put(120,0){\makebox(0,0){\footnotesize{$
\underset{{\wp}\mathrmbfit{tup}_{\mathcal{A}}(I,s)}
{\mathrmbf{Rel}_{\mathcal{A}}(I,s)}$}}}
\put(60,100){\makebox(0,0){\scriptsize{${\scriptstyle\sum}_{{h}}$}}}
\put(62,82){\makebox(0,0){\scriptsize{${h}^{\ast}$}}}
\put(60,58){\makebox(0,0){\scriptsize{${\scriptstyle\prod}_{{h}}$}}}
\put(60,20){\makebox(0,0){\scriptsize{$\exists_{{h}}$}}}
\put(64,2){\makebox(0,0){\scriptsize{${h}^{{\scriptscriptstyle-}1}$}}}
\put(60,-20){\makebox(0,0){\scriptsize{$\forall_{{h}}$}}}
\put(6,40){\makebox(0,0)[l]{\scriptsize{$\mathrmbfit{I}^{\mathcal{M}}_{{\langle{I'\!,s'}\rangle}}$}}}
\put(126,40){\makebox(0,0)[l]{\scriptsize{$\mathrmbfit{I}^{\mathcal{M}}_{{\langle{I\!,s}\rangle}}$}}}
\put(85,92){\vector(-1,0){50}}
\qbezier(35,80)(43,80)(51,80)\qbezier(69,80)(77,80)(85,80)\put(85,80){\vector(1,0){0}}
\put(85,68){\vector(-1,0){50}}
\put(85,12){\vector(-1,0){50}}
\put(85,-12){\vector(-1,0){50}}
\qbezier(35,0)(43,0)(51,0)\qbezier(69,0)(77,0)(85,0)\put(85,0){\vector(1,0){0}}
\put(0,65){\vector(0,-1){50}}
\put(120,65){\vector(0,-1){50}}
\end{picture}
\end{tabular}}}
&
{{\begin{tabular}[b]{l}
{\footnotesize\setlength{\extrarowheight}{2.5pt}$\begin{array}[b]{l}
{\langle{I',s'}\rangle}\xrightarrow{h}{\langle{I,s}\rangle} \\ 
\mathrmbfit{tup}_{\mathcal{A}}(I',s')\xleftarrow{\mathrmbfit{tup}(h)}\mathrmbfit{tup}_{\mathcal{A}}(I,s)
\end{array}$}
\\ \\
{\footnotesize\setlength{\extrarowheight}{2.5pt}$\begin{array}[b]{|@{\hspace{5pt}}l@{\hspace{2pt}}|}
\hline
{\scriptstyle\sum}_{{h}} \dashv {h}^{\ast} \dashv {\scriptstyle\prod}_{{h}} 
\\
\exists_{{h}} \dashv {h}^{{\scriptscriptstyle-}1} \dashv \forall_{{h}} 
\\ \hline
{h}^{\ast} \circ \mathrmbfit{I}^{\mathcal{M}}_{{\langle{I\!,s}\rangle}} 
= \mathrmbfit{I}^{\mathcal{M}}_{{\langle{I'\!,s'}\rangle}} \circ {h}^{{\scriptscriptstyle-}1}	
\\
{\scriptstyle\sum}_{{h}} \circ \mathrmbfit{I}^{\mathcal{M}}_{{\langle{I'\!,s'}\rangle}}
 = \mathrmbfit{I}^{\mathcal{M}}_{{\langle{I\!,s}\rangle}} \circ \exists_{{h}}
\\ \hline
\end{array}$}
\end{tabular}}}
\\ & \\
\end{tabular}}}
\end{center}
\caption{Formal/Semantics Reflection}
\label{tbl:fml-sem:refl}
\end{table}
%

\newpage
\subsubsection{Formula Structures.}\label{sub:sub:sec:fml:struc:obj}

Any structure 
$\mathcal{M} = {\langle{\mathcal{E},{\langle{\sigma,\tau}\rangle},\mathcal{A}}\rangle}$ 
has an associated formula structure
$\mathrmbfit{fmla}(\mathcal{M}) 
= \widehat{\mathcal{M}}  
= {\langle{\widehat{\mathcal{E}},{\langle{\widehat{\sigma},\tau}\rangle},\mathcal{A}}\rangle}$
with the same universe $\mathcal{U} = {\langle{K,\tau,Y}\rangle}$
and type domain $\mathcal{A}={\langle{X,Y,\models_{\mathcal{A}}}\rangle}$,
but with the formula schema 
$\mathrmbfit{fmla}(\mathcal{S}) = \widehat{\mathcal{S}} = {\langle{\widehat{R},\widehat{\sigma},X}\rangle}$
defined in \S\ref{sub:sub:sec:fmlism}
and the formula classification
$\widehat{\mathcal{E}} = {\langle{\widehat{R},K,\models_{\widehat{\mathcal{E}}}}\rangle}$
defined here.
%


\paragraph{Formula Classification.}

The formula classification
$\widehat{\mathcal{E}} = {\langle{\widehat{R},K,\models_{\widehat{\mathcal{E}}}}\rangle}$,
which extends 
the entity classification $\mathcal{E} = {\langle{R,K,\models_{\mathcal{E}}}\rangle}$,
is defined in Tbl.~\ref{tbl:fmla:cls} by induction
on formulas using formula interpretation.
%

%
%
%
%

%
\begin{table}
\begin{center}
{\footnotesize{\setlength{\extrarowheight}{2pt}
\begin{tabular}{|@{\hspace{5pt}}r@{\hspace{20pt}}l@{\hspace{10pt}\underline{when}\hspace{10pt}}c@{\hspace{5pt}}|}
\multicolumn{3}{l}{\rule{0pt}{1pt}\textsf{fiber:} signature ${\langle{I,s}\rangle}$ with extent (tuple set) 
$\mathrmbfit{tup}_{\mathcal{A}}(I,s){\,=\,}\prod_{i\in{I}}\,\mathcal{A}_{s_{\!i}}$}
\\
\multicolumn{3}{l}{\rule[-6pt]{0pt}{1pt}
$k{\;\in\;}K$
and $\varphi,\psi{\;\in\;}\widehat{R}(I,s)$}
\\ \hline
\textit{operator} & \multicolumn{1}{l}{\textit{definiendum}} & \multicolumn{1}{c|}{\textit{definiens}}
\\
entity type
& $k{\;\models_{\widehat{\mathcal{E}}}\;}r$
& $\tau(k){\;\in\;}\mathrmbfit{I}_{\mathcal{M}}(r){\,\subseteq\,}\mathrmbfit{tup}_{\mathcal{A}}(I,s)$
\\
meet
& $k{\;\models_{\widehat{\mathcal{E}}}\;}(\varphi{\,\wedge\,}\psi)$
& $k{\;\models_{\widehat{\mathcal{E}}}\;}\varphi$ and $k{\;\models_{\widehat{\mathcal{E}}}\;}\psi$
\\
join
& $k{\;\models_{\widehat{\mathcal{E}}}\;}(\varphi{\,\vee\,}\psi)$
& $k{\;\models_{\widehat{\mathcal{E}}}\;}\varphi$ or $k{\;\models_{\widehat{\mathcal{E}}}\;}\psi$
\\
top
& $k{\;\models_{\widehat{\mathcal{E}}}\;}{\scriptstyle\top_{{\langle{I,s}\rangle}}}$
& $\tau(k){\;\in\;}\mathrmbfit{tup}_{\mathcal{A}}(I,s)$
\\
bottom
& \multicolumn{2}{l|}{$k{\;\cancel{\models}_{\widehat{\mathcal{E}}}\;}{\scriptstyle\bot_{{\langle{I,s}\rangle}}}$}
\\
negation
& $k{\;\models_{\widehat{\mathcal{E}}}\;}(\neg\varphi)$
& 
$\tau(k){\;\in\;}\mathrmbfit{tup}_{\mathcal{A}}(I,s)$
and
$k{\;\cancel{\models}_{\widehat{\mathcal{E}}}\;}\varphi$
\\
implication
& $k{\;\models_{\widehat{\mathcal{E}}}\;}(\varphi{\,\rightarrowtriangle\,}\psi)$
& if $k{\;\models_{\widehat{\mathcal{E}}}\;}\varphi$ then $k{\;\models_{\widehat{\mathcal{E}}}\;}\psi$
\\
difference
& $k{\;\models_{\widehat{\mathcal{E}}}\;}(\varphi{\,\setminus\,}\psi)$
& $k{\;\models_{\widehat{\mathcal{E}}}\;}\varphi$ but not $k{\;\models_{\widehat{\mathcal{E}}}\;}\psi$
\rule[-5pt]{0pt}{10pt}
\\ \hline
\multicolumn{3}{l}{\rule{0pt}{20pt}\textsf{flow:} signature morphism
${\langle{I',s'}\rangle}\xrightarrow{h}{\langle{I,s}\rangle}$}
\\
\multicolumn{3}{l}{\rule[-6pt]{0pt}{1pt}with tuple map
$\mathrmbfit{tup}_{\mathcal{A}}(I',s')\xleftarrow{\mathrmbfit{tup}_{\mathcal{A}}(h)}\mathrmbfit{tup}_{\mathcal{A}}(I,s)$}
\\ \hline
\textit{operator} & \multicolumn{1}{l}{\textit{definiendum}} & \multicolumn{1}{c|}{\textit{definiens}}
\\
existential
& $k{\;\models_{\widehat{\mathcal{E}}}\;}{\scriptstyle\sum}_{h}(\varphi)$
&
$\tau(k){\,\in\,}{\exists}_{h}(\mathrmbfit{I}_{\mathcal{M}}(\varphi))$
\\
universal
& $k{\;\models_{\widehat{\mathcal{E}}}\;}{\scriptstyle\prod}_{h}(\varphi)$
& 
$\tau(k){\,\in\,}{\forall}_{h}(\mathrmbfit{I}_{\mathcal{M}}(\varphi))$ 
\\
substitution
& $k{\;\models_{\widehat{\mathcal{E}}}\;}{h}^{\ast}(\varphi')$
& 
$\tau(k){\,\in\,}{h}^{-1}(\mathrmbfit{I}_{\mathcal{M}}(\varphi'))$ 
\rule[-5pt]{0pt}{10pt}
\\ \hline
\end{tabular}}}
\end{center}
\caption{Formula Classification}
\label{tbl:fmla:cls}
\end{table}
\begin{proposition}\label{prop:cls:interp}
For any formula $\varphi\in\widehat{R}$,
$\mathrmbfit{ext}_{\widehat{\mathcal{E}}}(\varphi)=\tau^{-1}(\mathrmbfit{I}_{\mathcal{M}}(\varphi))$.
Hence,
\newline
\mbox{}\hfill
$k \models_{\widehat{\mathcal{E}}} \varphi$ 
\underline{iff}
$\tau(k)\in\mathrmbfit{I}_{\mathcal{M}}(\varphi)$
for all $k\in{K}$.
\footnote{Simply put,
an entity is in the view of a query
exactly when
the descriptor of that entity is in the interpretation of the query.
But,
there may be tuples in the interpretation of the query
that are not descriptors of any entity in the view of the query.}
\hfill\mbox{}
\end{proposition}
\begin{proof}
By induction, for all
$k{\;\in\;}K,\;\varphi,\psi{\,\in\,}\widehat{R}(I,s),\;\varphi'{\,\in\,}\widehat{R}(I',s')$:
{\footnotesize{
\begin{itemize}
\item 
[meet:]
$k{\;\models_{\widehat{\mathcal{E}}}\;}(\varphi{\,\wedge\,}\psi)$
\underline{when} 
$k{\;\models_{\widehat{\mathcal{E}}}\;}\varphi$ and $k{\;\models_{\widehat{\mathcal{E}}}\;}\psi$
\underline{iff}
$\tau(k)\in\mathrmbfit{I}_{\mathcal{M}}(\varphi)$ and
$\tau(k)\in\mathrmbfit{I}_{\mathcal{M}}(\psi)$
\underline{iff}
$\tau(k)\in
\mathrmbfit{I}_{\mathcal{M}}(\varphi){\,\cap\,}\mathrmbfit{I}_{\mathcal{M}}(\psi)
=\mathrmbfit{I}_{\mathcal{M}}(\varphi{\,\wedge\,}\psi)$;
%
\item 
[join:]
$k{\;\models_{\widehat{\mathcal{E}}}\;}(\varphi{\,\vee\,}\psi)$
\underline{when} 
$k{\;\models_{\widehat{\mathcal{E}}}\;}\varphi$ or $k{\;\models_{\widehat{\mathcal{E}}}\;}\psi$
\underline{iff}
$\tau(k)\in\mathrmbfit{I}_{\mathcal{M}}(\varphi)$ or
$\tau(k)\in\mathrmbfit{I}_{\mathcal{M}}(\psi)$
\underline{iff}
$\tau(k)\in
\mathrmbfit{I}_{\mathcal{M}}(\varphi){\,\cup\,}\mathrmbfit{I}_{\mathcal{M}}(\psi)
=\mathrmbfit{I}_{\mathcal{M}}(\varphi{\,\vee\,}\psi)$;
%
\item 
[top:]
$k{\;\models_{\widehat{\mathcal{E}}}\;}{\scriptstyle\top_{{\langle{I,s}\rangle}}}$
\underline{when} 
$\tau(k){\;\in\;}\mathrmbfit{tup}_{\mathcal{A}}(I,s)=\mathrmbfit{I}_{\mathcal{M}}({\scriptstyle\top_{{\langle{I,s}\rangle}}})$
%
\item 
[bottom:]
$k{\;\models_{\widehat{\mathcal{E}}}\;}{\scriptstyle\bot_{{\langle{I,s}\rangle}}}$
\underline{when} 
$\tau(k){\;\in\;}\emptyset=\mathrmbfit{I}_{\mathcal{M}}({\scriptstyle\bot_{{\langle{I,s}\rangle}}})$
%
\item 
[negation:]
$k{\;\models_{\widehat{\mathcal{E}}}\;}(\neg\varphi)$
\underline{when}  
$\tau(k){\;\in\;}\mathrmbfit{tup}_{\mathcal{A}}(I,s)$
and
$k{\;\cancel{\models}_{\widehat{\mathcal{E}}}\;}\varphi$
\underline{iff}
$\tau(k){\;\in\;}\mathrmbfit{tup}_{\mathcal{A}}(I,s)$
and
$\tau(k)\not\in\mathrmbfit{I}_{\mathcal{M}}(\varphi)$
\underline{iff}
$\tau(k)\in
\mathrmbfit{tup}_{\mathcal{A}}(I,s){\,\setminus\,}\mathrmbfit{I}_{\mathcal{M}}(\varphi)
=\neg\mathrmbfit{I}_{\mathcal{M}}(\varphi)
=\mathrmbfit{I}_{\mathcal{M}}(\neg\varphi)$;
%
\item 
[implication:]
$k{\;\models_{\widehat{\mathcal{E}}}\;}(\varphi{\,\rightarrowtriangle\,}\psi)$
\underline{when} 
($k{\;\models_{\widehat{\mathcal{E}}}\;}\varphi$ implies $k{\;\models_{\widehat{\mathcal{E}}}\;}\psi$)
\underline{iff}
($\tau(k)\in\mathrmbfit{I}_{\mathcal{M}}(\varphi)$
 implies $\tau(k)\in\mathrmbfit{I}_{\mathcal{M}}(\psi)$)
\underline{iff}
$\tau(k)\in
\bigl(
\mathrmbfit{I}_{\mathcal{M}}(\varphi){\,\rightarrowtriangle\,}\mathrmbfit{I}_{\mathcal{M}}(\psi)
\bigr)
=\mathrmbfit{I}_{\mathcal{M}}(\varphi{\,\rightarrowtriangle\,}\psi)$;
%
\item 
[difference:]
$k{\;\models_{\widehat{\mathcal{E}}}\;}(\varphi{\,\setminus\,}\psi)$
\underline{when} 
$k{\;\models_{\widehat{\mathcal{E}}}\;}\varphi$ but not $k{\;\models_{\widehat{\mathcal{E}}}\;}\psi$
\underline{iff}
($\tau(k)\in\mathrmbfit{I}_{\mathcal{M}}(\varphi)$
and
$\tau(k)\not\in\mathrmbfit{I}_{\mathcal{M}}(\psi)$)
\underline{iff}
$\tau(k)\in
\mathrmbfit{I}_{\mathcal{M}}(\varphi){\,\cap\,}(\neg\mathrmbfit{I}_{\mathcal{M}}(\psi))
=\mathrmbfit{I}_{\mathcal{M}}(\varphi){\,\setminus\,}\mathrmbfit{I}_{\mathcal{M}}(\psi)
=\mathrmbfit{I}_{\mathcal{M}}(\varphi{\,\setminus\,}\psi)$;
%
\item 
[existential:]
$k{\;\models_{\widehat{\mathcal{E}}}\;}{\scriptstyle\sum}_{h}(\varphi)$
\underline{when}
$\tau(k){\,\in\,}
{\exists}_{h}(\mathrmbfit{I}_{\mathcal{M}}(\varphi))
=\mathrmbfit{I}_{\mathcal{M}}({\scriptstyle\sum}_{h}(\varphi))$;
%
\item 
[universal:]
$k{\;\models_{\widehat{\mathcal{E}}}\;}{\scriptstyle\prod}_{h}(\varphi)$
\underline{when} 
$\tau(k){\,\in\,}
{\forall}_{h}(\mathrmbfit{I}_{\mathcal{M}}(\varphi))
=\mathrmbfit{I}_{\mathcal{M}}({\scriptstyle\prod}_{h}(\varphi))$;
and
%
\item 
[substitution:]
$k{\;\models_{\widehat{\mathcal{E}}}\;}{h}^{\ast}(\varphi')$
\underline{when} 
$\tau(k){\,\in\,}
{h}^{-1}(\mathrmbfit{I}_{\mathcal{M}}(\varphi'))
=\mathrmbfit{I}_{\mathcal{M}}({h}^{\ast}(\varphi'))$.
\hfill
\rule{5pt}{5pt}
\end{itemize}}}
%
\end{proof}
\begin{lemma}\label{lem:fmla:struc}
The associated formula structure
$\widehat{\mathcal{M}} 
= {\langle{\widehat{\mathcal{E}},{\langle{\widehat{\sigma},\tau}\rangle},\mathcal{A}}\rangle}$
is well-defined.
\end{lemma}
\begin{proof}
From Prop.~\ref{prop:cls:interp} above,
${\wp}\tau(\mathrmbfit{ext}_{\widehat{\mathcal{E}}}(\varphi)) 
\subseteq\mathrmbfit{I}_{\mathcal{M}}(\varphi)
\subseteq\mathrmbfit{tup}_{\mathcal{A}}(I,s)
{\,=\,}\mathrmbfit{ext}_{\mathrmbf{List}(\mathcal{A})}(I,s)$
for any formula $\varphi\in\widehat{R}(I,s)$.
Hence,
the condition for the list designation 
${\langle{\widehat{\sigma},\tau}\rangle} : \widehat{\mathcal{E}} \rightrightarrows \mathrmbf{List}(\mathcal{A})$
holds:
$k \models_{\widehat{\mathcal{E}}} \varphi$ 
\underline{implies}
$\tau(k){\;\models_{\mathrmbf{List}(\mathcal{A})}\;}\widehat{\sigma}(\varphi)$
for all keys $k\in{K}$.
\hfill\rule{5pt}{5pt}
\end{proof}
For all $\varphi\in\widehat{R}$,
we have the relationships
\begin{equation}\label{ext:interp:equiv}
\text{\fbox{${\setlength{\extrarowheight}{2pt}\begin{array}{c@{\hspace{20pt}}c}
{\wp}\tau(\mathrmbfit{ext}_{\widehat{\mathcal{E}}}(\varphi)){\;\subseteq\;}\mathrmbfit{I}_{\mathcal{M}}(\varphi)
& 
\mathrmbfit{ext}_{\widehat{\mathcal{E}}}(\varphi)=\tau^{-1}(\mathrmbfit{I}_{\mathcal{M}}(\varphi))
\end{array}}$}}
\end{equation}
%
%
Compare these orderings to those in 
Eqn.~\ref{ext:interp:equiv:ent:typ}
from the {\ttfamily FOLE} foundation paper \cite{kent:fole:era:found}. 
\newline
\newline
For all ${r}\in{R}$,
we have the relationships
\begin{equation}\label{ext:interp:equiv:ent:typ}
\text{\fbox{${\setlength{\extrarowheight}{2pt}
\begin{array}{c@{\hspace{20pt}}c}
{\wp}\tau(\mathrmbfit{ext}_{\mathcal{E}}(r)){\;=\;}\mathrmbfit{I}_{\mathcal{M}}(r)
& 
\mathrmbfit{ext}_{\mathcal{E}}(r){\;\subseteq\;}\tau^{-1}(\mathrmbfit{I}_{\mathcal{M}}(r))\,.
\end{array}}$}}
\end{equation}
\begin{definition}\label{assump:extens:struc}
A structure $\mathcal{M}$ is \emph{extensive}
when the right hand expression in (\ref{ext:interp:equiv:ent:typ})
is an equality:
$\mathrmbfit{ext}_{\mathcal{E}}(r){\;=\;}\tau^{-1}(\mathrmbfit{I}_{\mathcal{M}}(r))$
for any entity type ${r} \in {R}$.
Then,
the tuple map 
$K\xrightarrow{\tau}\mathrmbf{List}(Y)$
restricts to a bijection:
%
{{\begin{tabular}[c]{c}
\setlength{\unitlength}{0.6pt}
\begin{picture}(160,0)(-20,0)
\put(0,0){\makebox(0,0){\footnotesize{$\mathrmbfit{ext}_{\mathcal{E}}(r)$}}}
\put(120,0){\makebox(0,0){\footnotesize{$\mathrmbfit{I}_{\mathcal{M}}(r)$}}}
\put(55,12){\makebox(0,0)[l]{\footnotesize{$\tau$}}}
\put(55,-12){\makebox(0,0)[l]{\footnotesize{$\tau^{-1}$}}}
\put(35,5){\vector(1,0){50}}
\put(85,-5){\vector(-1,0){50}}
\end{picture}
\end{tabular}}}.
%
\footnote{For an extensive structure,
extent order is equivalent to interpretation order:
$\mathrmbfit{ext}_{\mathcal{E}}(r){\;\subseteq\;}\mathrmbfit{ext}_{\mathcal{E}}(r')$
\underline{iff}
$\mathrmbfit{I}_{\mathcal{M}}(r){\;\subseteq\;}\mathrmbfit{I}_{\mathcal{M}}(r')$
for entity types $r,r'\in{R}$.}
\end{definition}
\begin{lemma}
A structure $\mathcal{M}$ is extensive
when 
$K\xrightarrow{\tau}\mathrmbf{List}(Y)$ is injective.
\footnote{This corrects an editing error in the {\ttfamily FOLE} foundation paper (Kent~\cite{kent:fole:era:found}).}
\end{lemma}
\begin{proof}
By Eqn.~\ref{ext:interp:equiv:ent:typ},
${\wp}\tau(\mathrmbfit{ext}_{\mathcal{E}}(r)){\;=\;}\mathrmbfit{I}_{\mathcal{M}}(r)$.
Hence,
$\mathrmbfit{ext}_{\mathcal{E}}(r){\;=\;}\tau^{-1}(\mathrmbfit{I}_{\mathcal{M}}(r))$.
\footnote{For any function $A\xrightarrow{f}B$,
direct image is left adjoint to inverse image: 
${\wp}f(X){\,\subseteq\,}Y$ \underline{iff} $X{\,\subseteq\,}f^{-1}(Y)$
for any subsets $X{\,\subseteq\,}A$ and $Y{\,\subseteq\,}B$.
If $f$ is injective and ${\wp}f(X){\,=\,}Y$,
then $X{\,=\,}f^{-1}(Y)$.}
\end{proof}
Any structure $\mathcal{M}$ has an associated extensive structure.
An example is the key-embedding structure $\dot{\mathcal{M}}$ (see \cite{kent:fole:era:found}).
%



\begin{definition}\label{def:compre:struc}
A structure $\mathcal{M}$ is \emph{comprehensive}
\footnote{In the original discussion (Kent~\cite{kent:iccs2013}) about tabular interpretation, 
all {\ttfamily FOLE} structures were assumed to be comprehensive.}
when the left hand expression in Eqn.~\ref{ext:interp:equiv}
is an equality:
${\wp}\tau(\mathrmbfit{ext}_{\widehat{\mathcal{E}}}(\varphi)){\;=\;}\mathrmbfit{I}_{\mathcal{M}}(\varphi)$
for any formula $\varphi\in\widehat{R}$.
Then, the tuple map 
$K\xrightarrow{\tau}\mathrmbf{List}(Y)$
restricts to a bijection:
%
{{\begin{tabular}[c]{c}
\setlength{\unitlength}{0.6pt}
\begin{picture}(160,0)(-20,0)
\put(0,0){\makebox(0,0){\footnotesize{$\mathrmbfit{ext}_{\widehat{\mathcal{E}}}(\varphi)$}}}
\put(120,0){\makebox(0,0){\footnotesize{$\mathrmbfit{I}_{\mathcal{M}}(\varphi)$}}}
\put(55,12){\makebox(0,0)[l]{\footnotesize{$\tau$}}}
\put(55,-12){\makebox(0,0)[l]{\footnotesize{$\tau^{-1}$}}}
\put(35,5){\vector(1,0){50}}
\put(85,-5){\vector(-1,0){50}}
\end{picture}
\end{tabular}}}.
%
\footnote{For a comprehensive structure,
extent order is equivalent to interpretation order:
$\mathrmbfit{ext}_{\widehat{\mathcal{E}}}(\varphi)){\;\subseteq\;}\mathrmbfit{ext}_{\widehat{\mathcal{E}}}(\psi))$
\underline{iff}
$\mathrmbfit{I}_{\mathcal{M}}(\varphi){\;\subseteq\;}\mathrmbfit{I}_{\mathcal{M}}(\psi)$
for formulas $\varphi,\psi\in\widehat{R}$.}
\end{definition}
The condition
${\wp}\tau(\mathrmbfit{ext}_{\widehat{\mathcal{E}}}(\varphi)){\;=\;}\mathrmbfit{I}_{\mathcal{M}}(\varphi)$
means that
the restricted tuple function
$\mathrmbfit{ext}_{\widehat{\mathcal{E}}}(\varphi)
\xrightarrow{\;\tau_{\varphi}\;}\mathrmbfit{I}_{\mathcal{M}}(\varphi)$
is surjective.
Hence, 
we can choose an injective inverse function
$\mathrmbfit{ext}_{\widehat{\mathcal{E}}}(\varphi)
\xleftarrow{\;\gamma_{\varphi}\;}\mathrmbfit{I}_{\mathcal{M}}(\varphi)$
satisfying
$\gamma_{\varphi}{\,\cdot\,}\tau_{\varphi} = \mathrmit{1}_{\mathrmbfit{I}_{\mathcal{M}}(\varphi)}$.
In a comprehensive structure $\mathcal{M}$,
we make this choice for each formula $\varphi\in\widehat{R}$.
\begin{proposition}\label{prop:incl}
Let $\mathcal{M} = {\langle{\mathcal{E},\sigma,\tau,\mathcal{A}}\rangle}$ be a structure, 
whose key set is a subset of $Y$-tuples $K\subseteq\mathrmbf{List}(Y)$ and 
whose tuple map is inclusion $K\xrightarrow{\mathrmit{inc}}\mathrmbf{List}(Y)$.
Then $\mathcal{M}$ is comprehensive.
\end{proposition}
\begin{proof}
Let $\varphi\in\widehat{R}$ be any formula. 
By Prop.~\ref{prop:cls:interp},
$k \models_{\widehat{\mathcal{E}}} \varphi$ 
\underline{iff}
$k=\mathrmit{inc}(k)\in\mathrmbfit{I}_{\mathcal{M}}(\varphi)$
for all $k\in{K}$;
equivalently,
${\wp}\mathrmit{inc}(\mathrmbfit{ext}_{\widehat{\mathcal{E}}}(\varphi))
=\mathrmbfit{ext}_{\widehat{\mathcal{E}}}(\varphi)
=\mathrmbfit{I}_{\mathcal{M}}(\varphi)$.
\end{proof}

\begin{definition}\label{def:img:struc}
A {\ttfamily FOLE} structure $\mathcal{M}$ has 
an associated \emph{image} structure
$\mathring{\mathcal{M}} = {\langle{\mathring{\mathcal{E}},\sigma,\mathring{\tau},\mathcal{A}}\rangle}$
with the same schema $\mathcal{S} = {\langle{R,\sigma,X}\rangle}$
and typed domain $\mathcal{A} = {\langle{X,Y,\models_{\mathcal{A}}}\rangle}$,
but with the trivial universe
$\mathring{\mathcal{U}} = {\langle{\mathrmbf{List}(Y),\mathrmit{1}_{\mathrmbf{List}(Y)},Y}\rangle}$
and the descriptor entity classification
$\mathring{\mathcal{E}} = 
{\exists_{\tau}}(\mathcal{E}) =
{\langle{R,\mathrmbf{List}(Y),\models_{\mathring{\mathcal{E}}}}\rangle}$,
where a tuple serves as its own identifier:
${\langle{I,t}\rangle}{\;\models_{\mathring{\mathcal{E}}}\;}r$
for a tuple ${\langle{I,t}\rangle}\in\mathrmbf{List}(Y)$ and an entity type $r\in{R}$
when ${\langle{I,t}\rangle}$ is the descriptor 
$\tau(k)={\langle{I,t}\rangle}$
for some key $k{\,\in\,}K$
such that
$k{\;\models_{\mathcal{E}}\;}r$.
Thus,
$\mathrmbfit{ext}_{\mathring{\mathcal{E}}}(r)
= {\wp}\tau(\mathrmbfit{ext}_{\mathcal{E}}(r))
{\;\subseteq\;}\mathrmbf{List}(Y)$
is the direct image of 
$\mathrmbfit{ext}_{\mathcal{E}}(r){\;\subseteq\;}K$
along the tuple map $K\xrightarrow{\tau}\mathrmbf{List}(Y)$.
\end{definition}
\begin{center}
{{\begin{tabular}{c}
\setlength{\unitlength}{0.5pt}
\begin{picture}(120,80)(0,0)
\put(0,80){\makebox(0,0){\footnotesize{$R$}}}
\put(0,0){\makebox(0,0){\footnotesize{$K$}}}
\put(87,80){\makebox(0,0)[l]{\footnotesize{$\mathrmbf{List}(X)$}}}
\put(87,0){\makebox(0,0)[l]{\footnotesize{$\mathrmbf{List}(Y)$}}}
\put(8,40){\makebox(0,0)[l]{\scriptsize{$\models_{\mathcal{E}}$}}}
\put(70,40){\makebox(0,0)[l]{\scriptsize{$\models_{\mathring{\mathcal{E}}}$}}}
\put(128,40){\makebox(0,0)[l]{\scriptsize{$\models_{\mathrmbf{List}(\mathcal{A})}$}}}
\put(50,90){\makebox(0,0){\scriptsize{$\sigma$}}}
\put(50,10){\makebox(0,0){\scriptsize{$\tau$}}}
\put(20,80){\vector(1,0){60}}
\put(20,0){\vector(1,0){60}}
\put(0,65){\line(0,-1){50}}
\put(120,65){\line(0,-1){50}}
\put(14,66){\line(3,-2){80}}
\end{picture}
\end{tabular}}}
\end{center}
\begin{proposition}\label{prop:img:compre}
For the image structure,
the formula interpretation in $\mathring{\mathcal{M}}$
is the formula interpretation in $\mathcal{M}$,
and the formula extent in $\mathring{\mathcal{M}}$
is this interpretation:
$\mathrmbfit{ext}_{\widehat{\mathring{\mathcal{E}}}}(\varphi)
{\;=\;}\mathrmbfit{I}_{\mathring{\mathcal{M}}}(\varphi)
{\;=\;}\mathrmbfit{I}_{\mathcal{M}}(\varphi)$
for any formula
$\varphi\in\widehat{R}$.
The image structure $\mathring{\mathcal{M}}$ is comprehensive.
\end{proposition}
\begin{proof}
At the base step in Tbl.~\ref{tbl:fml:int},
$\mathrmbfit{I}_{\mathring{\mathcal{M}}}(r)
= {\wp}\mathrmit{1}_{\mathrmbf{List}(Y)}(\mathrmbfit{ext}_{\mathring{\mathcal{E}}}(r))
= \mathrmbfit{ext}_{\mathring{\mathcal{E}}}(r)
= {\wp}\tau(\mathrmbfit{ext}_{\mathcal{E}}(r))
= \mathrmbfit{I}_{\mathcal{M}}(r)$
for any entity type $r\in{R}$
(Def.~\ref{def:img:struc} above).
By induction,
$\mathrmbfit{I}_{\mathring{\mathcal{M}}}(\varphi)
=\mathrmbfit{I}_{\mathcal{M}}(\varphi)$
for any formula $\varphi\in\widehat{R}$.
By Prop.~\ref{prop:cls:interp},
${\langle{I,t}\rangle} \models_{\widehat{\mathring{\mathcal{E}}}} \varphi$ 
\underline{iff}
${\langle{I,t}\rangle}=\mathrmit{1}_{\mathrmbf{List}(Y)}(I,t)\in\mathrmbfit{I}_{\mathring{\mathcal{M}}}(\varphi)$
for any tuple ${\langle{I,t}\rangle}\in\mathrmbf{List}(Y)$;
equivalently,
${\wp}\mathrmit{1}_{\mathrmbf{List}(Y)}(\mathrmbfit{ext}_{\widehat{\mathring{\mathcal{E}}}}(\varphi))
=\mathrmbfit{ext}_{\widehat{\mathring{\mathcal{E}}}}(\varphi)
=\mathrmbfit{I}_{\mathring{\mathcal{M}}}(\varphi)$.
\end{proof}
\begin{corollary}\label{cor:inj}
A structure $\mathcal{M}$ is comprehensive
when $K\xrightarrow{\tau}\mathrmbf{List}(Y)$ is injective.
\end{corollary}
\begin{proof}
By Prop.~\ref{prop:img:compre},
$
{\wp}\tau(\mathrmbfit{ext}_{\widehat{\mathcal{E}}}(\varphi))
{\;=\;}\mathrmbfit{ext}_{\widehat{\mathring{\mathcal{E}}}}(\varphi)
{\;=\;}\mathrmbfit{I}_{\mathring{\mathcal{M}}}(\varphi)
{\;=\;}\mathrmbfit{I}_{\mathcal{M}}(\varphi)$
for any $\varphi\in\widehat{R}$.
\end{proof}
\begin{proposition}\label{prop:key:inc}
For any structure
$\mathcal{M} = {\langle{\mathcal{E},\sigma,\tau,\mathcal{A}}\rangle}$,
the associated ``key-embedding'' structure
$\dot{\mathcal{M}} = {\langle{\mathcal{E},\dot{\sigma},\dot{\tau},\dot{\mathcal{A}}}\rangle}$
(see the {\ttfamily FOLE} Foundation paper \cite{kent:fole:era:found})
is  comprehensive.
\end{proposition}
\begin{proof}
By Cor.~\ref{cor:inj}, 
since the tuple map $K\xrightarrow{\dot{\tau}}\mathrmbf{List}(\dot{Y})$ 
of the key-embedding structure $\dot{\mathcal{M}}$ is injective.
\end{proof}
\begin{proposition}\label{prop:iter:fmla}
If $\mathcal{M}$ is comprehensive,
then $\widehat{\mathcal{M}}$ is comprehensive
with the same
interpretation and entity extent:
$\mathrmbfit{I}_{\widehat{\mathcal{M}}}=\mathrmbfit{I}_{\mathcal{M}}$
and 
$\mathrmbfit{ext}_{\widehat{\widehat{\mathcal{E}}}} = \mathrmbfit{ext}_{\widehat{\mathcal{E}}}$.
\end{proposition}
\begin{proof}
At the base step in Tbl.~\ref{tbl:fml:int},
$\mathrmbfit{I}_{\widehat{\mathcal{M}}}(\varphi)
= {\wp}\tau(\mathrmbfit{ext}_{\widehat{\mathcal{E}}}(\varphi))
= \mathrmbfit{I}_{\mathcal{M}}(\varphi)$
by comprehension.
Hence,
by induction
$\mathrmbfit{I}_{\widehat{\mathcal{M}}} = \mathrmbfit{I}_{\mathcal{M}}$.
By Eqn.~\ref{ext:interp:equiv},
$\mathrmbfit{ext}_{\widehat{\widehat{\mathcal{E}}}}(\varphi)
= \tau^{-1}(\mathrmbfit{I}_{\widehat{\mathcal{M}}}(\varphi))
= \tau^{-1}(\mathrmbfit{I}_{\mathcal{M}}(\varphi))
= \mathrmbfit{ext}_{\widehat{\mathcal{E}}}(\varphi)$.
\end{proof}
%


\newpage
\subsubsection{Formula Structure Morphisms.}\label{sub:sub:sec:fml:struc:mor}

Let
{\footnotesize{
\begin{equation}\label{eqn:fmla:info}
\mathcal{M}_{2} = {\langle{\mathcal{E}_{2},{\langle{\sigma_{2},\tau_{2}}\rangle},\mathcal{A}_{2}}\rangle} 
\xrightleftharpoons{{\langle{r,k,f,g}\rangle}}
{\langle{\mathcal{E}_{1},{\langle{\sigma_{1},\tau_{1}}\rangle},\mathcal{A}_{1}}\rangle} = \mathcal{M}_{1}
\end{equation}
}\normalsize}
be any structure morphism.
We can define a formula structure morphism with certain qualifications.
\begin{lemma}\label{lem:bool:fmla}
Any of the following equivalent conditions hold for Eqn.~\ref{eqn:fmla:info}
\begin{center}
{\footnotesize\setlength{\extrarowheight}{2pt}{\begin{tabular}{@{\hspace{20pt}}c@{\hspace{20pt}}}
$k(k_{1}){\;\models_{\widehat{\mathcal{E}}_{2}}\;}\varphi_{2}$
 \underline{iff} 
$k_{1}{\;\models_{\widehat{\mathcal{E}}_{1}}\;}\widehat{r}(\varphi_{2})$
\\
$k(k_{1}){\;\in\;}\mathrmbfit{ext}_{\widehat{\mathcal{E}}_{2}}(\varphi_{2})$
 \underline{iff} 
$k_{1}{\;\in\;}\mathrmbfit{ext}_{\widehat{\mathcal{E}}_{1}}(\widehat{r}(\varphi_{2}))$
\\
$k^{-1}(\mathrmbfit{ext}_{\widehat{\mathcal{E}}_{2}}(\varphi_{2}))
 = 
\mathrmbfit{ext}_{\widehat{\mathcal{E}}_{1}}(\widehat{r}(\varphi_{2}))\;\;\;\;$
\end{tabular}}}
\end{center}
for any source boolean formula $\varphi_{2}{\,\in\,}\widehat{R}_{2}$ 
(containing no quantification/substitution) 
and any target key $k_{1}{\,\in\,}K_{1}$.
These express entity informorphism conditions.
\end{lemma}
\begin{proof}
Proved by induction.
This is clearly true for entity types (relation symbols).
Check on all booleans: meets, joins, negations, etc. 
\begin{description}
\item[all:]
${\langle{I_{1},s_{1}}\rangle}
{\;=\;}\sigma_{1}(\widehat{r}(\varphi_{2}))
{\;=\;}{\scriptstyle\sum}_{f}(\sigma_{2}(\varphi_{2}))
{\;=\;}{\scriptstyle\sum}_{f}(I_{2},s_{2})$
\hfill
schema morphism
\newline\rule{100pt}{0.1pt}
\item[meets:] 
\hfill
[entity infomorphism]
\mbox{}\newline
$k_{1}{\;\models_{\widehat{\mathcal{E}}_{1}}\;}\widehat{r}(\varphi_{2}{\,\wedge\,}\psi_{2})
{\;=\;}\widehat{r}(\varphi_{2}){\,\wedge\,}\widehat{r}(\psi_{2})$
\hfill
definition of $\widehat{r}$
\newline
\underline{iff}
$k_{1}{\;\models_{\widehat{\mathcal{E}}_{1}}\,}\widehat{r}(\varphi_{2})$ and
$k_{1}{\;\models_{\widehat{\mathcal{E}}_{1}}\,}\widehat{r}(\psi_{2})$
\hfill
definition of $\models_{\widehat{\mathcal{E}}_{1}}$
\newline
\underline{iff}
$k(k_{1}){\;\models_{\widehat{\mathcal{E}}_{2}}\,}\varphi_{2}$ and
$k(k_{1}){\;\models_{\widehat{\mathcal{E}}_{2}}\,}\psi_{2}$
\hfill
induction
\newline
\underline{iff}
$k(k_{1}){\;\models_{\widehat{\mathcal{E}}_{2}}\,}(\varphi_{2}{\,\wedge\,}\psi_{2})$,
\hfill
definition of $\models_{\widehat{\mathcal{E}}_{2}}$
\newline\rule{100pt}{0.1pt}
\item[top:] 
\hfill
[entity infomorphism]
\mbox{}\newline
$k_{1}{\;\models_{\widehat{\mathcal{E}}_{1}}\,}
\widehat{r}(\top_{{\langle{I_{2},s_{2}}\rangle}})
{\,=\,}\top_{{\langle{I_{1},s_{1}}\rangle}}
{\,=\,}\top_{{\scriptscriptstyle\sum}_{f}(I_{2},s_{2})}$
\hfill
definition of $\widehat{r}$
\newline
\underline{iff}
$\tau_{1}(k_{1}){\,\in\,}\mathrmbfit{tup}_{\mathcal{A}_{1}}({\scriptstyle\sum}_{f}(I_{2},s_{2}))$
\hfill
definition of $\models_{\widehat{\mathcal{E}}_{1}}$
\newline
\underline{iff}
${\scriptstyle\sum}_{g}(\tau_{1}(k_{1})){\,\in\,}\mathrmbfit{tup}_{\mathcal{A}_{2}}(I_{2},s_{2})$
\hfill
type domain morphism 
\newline
\underline{iff}
$\tau_{2}(k(k_{1})){\,\in\,}\mathrmbfit{tup}_{\mathcal{A}_{2}}(I_{2},s_{2})$
\hfill
universe morphism 
\newline
\underline{iff}
$k(k_{1}){\;\models_{\widehat{\mathcal{E}}_{2}}}\top_{{\langle{I_{2},s_{2}}\rangle}}$
\hfill
definition of $\models_{\widehat{\mathcal{E}}_{2}}$
\newline\rule{100pt}{0.1pt}
\item[negation:] 
\hfill
[entity infomorphism]
\mbox{}\newline
$k_{1}{\;\models_{\widehat{\mathcal{E}}_{1}}\;}\widehat{r}(\neg\varphi_{2})
{\;=\;}\neg\widehat{r}(\varphi_{2})$,
\hfill
definition of $\widehat{r}$
\newline
\underline{iff}
$\tau_{1}(k_{1}){\;\in\;}
\mathrmbfit{tup}_{\mathcal{A}_{1}}({\scriptstyle\sum}_{f}(I_{2},s_{2}))$
and
$k_{1}{\;\cancel{\models}_{\widehat{\mathcal{E}_{1}}}\;}\widehat{r}(\varphi_{2})$
\hfill
definition of $\models_{\widehat{\mathcal{E}}_{1}}$
\newline
\underline{iff}
${\scriptstyle\sum}_{g}(\tau_{1}(k_{1})){\,\in\,}\mathrmbfit{tup}_{\mathcal{A}_{2}}(I_{2},s_{2})$
and
$k_{1}{\;\cancel{\models}_{\widehat{\mathcal{E}_{1}}}\;}\widehat{r}(\varphi_{2})$
\hfill
type domain morphism 
\newline
\underline{iff}
$\tau_{2}(k(k_{1})){\,\in\,}\mathrmbfit{tup}_{\mathcal{A}_{2}}(I_{2},s_{2})$
and
$k_{1}{\;\cancel{\models}_{\widehat{\mathcal{E}_{1}}}\;}\widehat{r}(\varphi_{2})$
\hfill
universe morphism
\newline
\underline{iff}
$\tau_{2}(k(k_{1})){\,\in\,}\mathrmbfit{tup}_{\mathcal{A}_{2}}(I_{2},s_{2})$
and
$k(k_{1}){\;\cancel{\models}_{\widehat{\mathcal{E}}_{2}}\,}\varphi_{2}$
\hfill
induction
\newline
\underline{iff}
$k(k_{1}){\;\models_{\widehat{\mathcal{E}}_{1}}\;}\neg\varphi_{2}$
\hfill
definition of $\models_{\widehat{\mathcal{E}}_{1}}$
\end{description}
\rule{5pt}{5pt}
\end{proof}
\begin{proposition}\label{bool:struc:fmla}
There is a boolean formula structure passage 
\[\mbox{\footnotesize{$
\mathrmbf{Struc}\xrightarrow{\mathrmbfit{fmla}_{0}}\mathrmbf{Struc}
$,}\normalsize}\]
which is idempotent
$\mathrmbfit{fmla}_{0}(\mathrmbfit{fmla}_{0}(\mathcal{M})){\;\cong\;}\mathrmbfit{fmla}_{0}(\mathcal{M})$:
a formula of formulas is another formula. 
\end{proposition}

\newpage




%
\begin{lemma}\label{lem:comp:fixdata:fmla}
Assume the structures $\mathcal{M}_{1}$ and $\mathcal{M}_{1}$ are comprehensive
and the data value set $Y$ is fixed.
Any of the following equivalent conditions hold for Eqn.~\ref{eqn:fmla:info}
%
%
\begin{center}
{\footnotesize\setlength{\extrarowheight}{2pt}{\begin{tabular}{@{\hspace{20pt}}c@{\hspace{20pt}}}
$k(k_{1}){\;\models_{\widehat{\mathcal{E}}_{2}}\;}\varphi_{2}$
 \underline{if} 
$k_{1}{\;\models_{\widehat{\mathcal{E}}_{1}}\;}\widehat{r}(\varphi_{2})$
\\
$k(k_{1}){\;\in\;}\mathrmbfit{ext}_{\widehat{\mathcal{E}}_{2}}(\varphi_{2})$
 \underline{if} 
$k_{1}{\;\in\;}\mathrmbfit{ext}_{\widehat{\mathcal{E}}_{1}}(\widehat{r}(\varphi_{2}))$
\\
$k^{-1}(\mathrmbfit{ext}_{\widehat{\mathcal{E}}_{2}}(\varphi_{2}))
 \supseteq 
\mathrmbfit{ext}_{\widehat{\mathcal{E}}_{1}}(\widehat{r}(\varphi_{2}))\;\;\;\;$
\\
$\;\;\;\;\;\;\;\;\;\;\;\;\mathrmbfit{ext}_{\widehat{\mathcal{E}}_{2}}(\varphi_{2})
 \supseteq
{{\wp}k}(\mathrmbfit{ext}_{\widehat{\mathcal{E}}_{1}}(\widehat{r}(\varphi_{2})))$
\end{tabular}}}
\end{center}
for any source boolean formula $\varphi_{2}{\,\in\,}\widehat{R}_{2}$ 
and any target key $k_{1}{\,\in\,}K_{1}$.
These imply that the following condition holds
\[\mbox{\footnotesize{$
\mathrmbfit{I}_{\mathcal{M}_{2}}(\varphi_{2})
 \supseteq
\mathrmbfit{I}_{\mathcal{M}_{1}}(\widehat{r}(\varphi_{2}))
$}\normalsize}\]
holds for any source formula $\varphi_{2}{\,\in\,}\widehat{R}_{2}$, 
since the structures are comprehensive and the value set is fixed.
\footnote{$
\mathrmbfit{I}_{\mathcal{M}_{2}}(\varphi_{2})
=  {{\wp}\tau_{2}}(\mathrmbfit{ext}_{\widehat{\mathcal{E}}_{2}}(\varphi_{2}))
\supseteq {{\wp}\tau_{2}}({{\wp}k}(\mathrmbfit{ext}_{\widehat{\mathcal{E}}_{1}}(\widehat{r}(\varphi_{2}))))
= {{\wp}\tau_{1}}(\mathrmbfit{ext}_{\widehat{\mathcal{E}}_{1}}(\widehat{r}(\varphi_{2})))
= \mathrmbfit{I}_{\mathcal{M}_{1}}(\widehat{r}(\varphi_{2}))$.}
\end{lemma}
\begin{proof}
Proved by induction.
True for booleans by proof analogous to Lem.~\ref{lem:bool:fmla}.
\begin{description}
\item[all:]
${\langle{I_{1},s_{1}}\rangle}
{\;=\;}\sigma_{1}(\widehat{r}(\varphi_{2}))
{\;=\;}{\scriptstyle\sum}_{f}(\sigma_{2}(\varphi_{2}))
{\;=\;}{\scriptstyle\sum}_{f}(I_{2},s_{2})$
\hfill
schema morphism
\newline
$\mathrmbfit{tup}_{\mathcal{A}_{2}}(I_{2},s_{2})
{\;=\;}\mathrmbfit{tup}_{\mathcal{A}_{1}}({\scriptstyle\sum}_{f}(I_{2},s_{2}))
{\;=\;}\mathrmbfit{tup}_{\mathcal{A}_{1}}(I_{1},s_{1})$
\hfill
type domain morphism
\newline\rule{100pt}{0.1pt}
\item[existential:] 
$k_{1}{\;\models_{\widehat{\mathcal{E}}_{1}}\;}\widehat{r}({\scriptstyle\sum}_{h}(\varphi_{2}))$
\newline
\underline{iff}
$k_{1}{\;\models_{\widehat{\mathcal{E}}_{1}}\;}{\scriptstyle\sum}_{h}(\widehat{r}(\varphi_{2}))$
\hfill
definition of $\widehat{r}$
\newline
\underline{iff}
$\tau_{1}(k_{1}){\,\in\,}
{\exists}_{h}(\mathrmbfit{I}_{\mathcal{M}_{1}}(\widehat{r}(\varphi_{2})))$
\hfill
definition of $\models_{\widehat{\mathcal{E}}_{1}}$
\newline
\underline{implies}
$\tau_{1}(k_{1}){\,\in\,}{\exists}_{h}(\mathrmbfit{I}_{\mathcal{M}_{2}}(\varphi_{2}))$
\hfill
induction
\newline
\underline{iff}
$\tau_{2}(k(k_{1}))
{\,\in\,}
{\exists}_{h}(\mathrmbfit{I}_{\mathcal{M}_{2}}(\varphi_{2}))$
\hfill
universe morphism
\newline
\underline{iff}
$k(k_{1}){\;\models_{\widehat{\mathcal{E}}_{2}}\;}{\scriptstyle\sum}_{h}(\varphi_{2})$
\hfill
definition of $\models_{\widehat{\mathcal{E}}_{2}}$
\newline\rule{100pt}{0.1pt}
\item[substitution:] 
$k_{1}{\;\models_{\widehat{\mathcal{E}}_{1}}\;}\widehat{r}({h}^{\ast}(\varphi'_{2}))$
\newline
$k_{1}{\;\models_{\widehat{\mathcal{E}}_{1}}\;}{h}^{\ast}(\widehat{r}(\varphi'_{2}))$
\hfill
definition of $\widehat{r}$
\newline
\underline{iff}
$\tau_{1}(k_{1}){\,\in\,}
{h}^{-1}(\mathrmbfit{I}_{\mathcal{M}_{1}}(\widehat{r}(\varphi'_{2})))$
\hfill
definition of $\models_{\widehat{\mathcal{E}}_{1}}$
\newline
\underline{implies}
$\tau_{1}(k_{1})
{\,\in\,}
{h}^{-1}(\mathrmbfit{I}_{\mathcal{M}_{2}}(\varphi'_{2}))$
\hfill
induction
\newline
\underline{iff}
$\tau_{2}(k(k_{1}))
{\,\in\,}
{h}^{-1}(\mathrmbfit{I}_{\mathcal{M}_{2}}(\varphi'_{2}))$
\hfill
universe morphism
\newline
\underline{iff}
$k(k_{1}){\;\models_{\widehat{\mathcal{E}}_{2}}\;}{h}^{\ast}(\varphi'_{2})$
\hfill
definition of $\models_{\widehat{\mathcal{E}}_{2}}$
\end{description}
\rule{5pt}{5pt}
\end{proof}

\newpage

\begin{lemma}\label{lem:all:fmla}
As in Lem.~\ref{lem:comp:fixdata:fmla},
assume the structures $\mathcal{M}_{1}$ and $\mathcal{M}_{1}$ are comprehensive
and the data value set $Y$ is fixed.
In addition, 
assume the key function $K_{2}\xleftarrow{\;k\;}K_{1}$ is surjective.
Any of the following equivalent conditions hold for Eqn.~\ref{eqn:fmla:info}
\begin{center}
{\footnotesize\setlength{\extrarowheight}{2pt}{\begin{tabular}{@{\hspace{20pt}}c@{\hspace{20pt}}}
$k(k_{1}){\;\models_{\widehat{\mathcal{E}}_{2}}\;}\varphi_{2}$
 \underline{iff} 
$k_{1}{\;\models_{\widehat{\mathcal{E}}_{1}}\;}\widehat{r}(\varphi_{2})$
\\
$k(k_{1}){\;\in\;}\mathrmbfit{ext}_{\widehat{\mathcal{E}}_{2}}(\varphi_{2})$
 \underline{iff} 
$k_{1}{\;\in\;}\mathrmbfit{ext}_{\widehat{\mathcal{E}}_{1}}(\widehat{r}(\varphi_{2}))$
\\
$k^{-1}(\mathrmbfit{ext}_{\widehat{\mathcal{E}}_{2}}(\varphi_{2}))
 = 
\mathrmbfit{ext}_{\widehat{\mathcal{E}}_{1}}(\widehat{r}(\varphi_{2}))\;\;\;\;$
\end{tabular}}}
\end{center}
for any source boolean formula $\varphi_{2}{\,\in\,}\widehat{R}_{2}$ 
and any target key $k_{1}{\,\in\,}K_{1}$.
These imply that the following condition holds
\[\mbox{\footnotesize{$
\mathrmbfit{ext}_{\widehat{\mathcal{E}}_{2}}(\varphi_{2})
= {{\wp}k}(\mathrmbfit{ext}_{\widehat{\mathcal{E}}_{1}}(\widehat{r}(\varphi_{2})))
$}\normalsize}\]
for any source formula $\varphi_{2}{\,\in\,}\widehat{R}_{2}$, 
since the function $K_{2}\xleftarrow{\;k\;}K_{1}$ is surjective.
\footnote{For any function $A\xrightarrow{f}B$,
direct image 
is left adjoint to 
inverse image: 
${\wp}f(X){\,\subseteq\,}Y$ \underline{iff} $X{\,\subseteq\,}f^{-1}(Y)$
for any subsets $X{\,\subseteq\,}A$ and $Y{\,\subseteq\,}B$.
If $f$ is surjective and $X{\,=\,}f^{-1}(Y)$,
then ${\wp}f(X){\,=\,}Y$.}
This implies that the following condition holds
\[\mbox{\footnotesize{$
\mathrmbfit{I}_{\mathcal{M}_{2}}(\varphi_{2})
= \mathrmbfit{I}_{\mathcal{M}_{1}}(\widehat{r}(\varphi_{2}))
$}\normalsize}\]
for any source formula $\varphi_{2}{\,\in\,}\widehat{R}_{2}$, 
since the structures are comprehensive and the value set is fixed.
\footnote{$\mathrmbfit{I}_{\mathcal{M}_{2}}(\varphi_{2})
= {{\wp}\tau_{2}}(\mathrmbfit{ext}_{\widehat{\mathcal{E}}_{2}}(\varphi_{2}))
= {{\wp}\tau_{2}}({{\wp}k}(\mathrmbfit{ext}_{\widehat{\mathcal{E}}_{1}}(\widehat{r}(\varphi_{2}))))
= {{\wp}\tau_{1}}(\mathrmbfit{ext}_{\widehat{\mathcal{E}}_{1}}(\widehat{r}(\varphi_{2})))
= \mathrmbfit{I}_{\mathcal{M}_{1}}(\widehat{r}(\varphi_{2}))$.}
\end{lemma}
\begin{proof}
Proof is similar to that of Lem.~\ref{lem:comp:fixdata:fmla}.
\hfill\rule{5pt}{5pt}
\end{proof}
\begin{lemma}\label{lem:fmla:struc:mor}
With the assumptions of Lem.~\ref{lem:all:fmla},
there is an associated formula structure morphism
between comprehensive formula structures
\[\mbox{\footnotesize{$
\mathrmbfit{fmla}(\mathcal{M}_{2}) = {\langle{\widehat{\mathcal{E}}_{2},{\langle{\widehat{\sigma}_{2},\tau_{2}}\rangle},\mathcal{A}_{2}}\rangle} 
\xrightleftharpoons[{\langle{\widehat{r},k,f,1_{Y}}\rangle}]{\mathrmbfit{fmla}(r,k,f,1_{Y})}
{\langle{\widehat{\mathcal{E}}_{1},{\langle{\widehat{\sigma}_{1},\tau_{1}}\rangle},\mathcal{A}_{1}}\rangle} = \mathrmbfit{fmla}(\mathcal{M}_{1})
$}\normalsize}\]
with schema morphism 
$\mathrmbfit{fmla}(\mathcal{S}_{2})={\langle{\widehat{R}_{2},\widehat{\sigma}_{2},X_{2}}\rangle} 
\xRightarrow[{\langle{\widehat{r},f}\rangle}]{\mathrmbfit{fmla}(r,f)}
{\langle{\widehat{R}_{1},\widehat{\sigma}_{1},X_{1}}\rangle}=\mathrmbfit{fmla}(\mathcal{S}_{1})$
and entity infomorphism
$\widehat{\mathcal{E}_{2}}={\langle{\widehat{R}_{2},K_{2},\models_{\widehat{\mathcal{E}_{2}}}}\rangle}
\xrightleftharpoons{{\langle{\widehat{r},k}\rangle}}
{\langle{\widehat{R}_{1},K_{1},\models_{\widehat{\mathcal{E}_{1}}}}\rangle}=\widehat{\mathcal{E}_{1}}$.
\end{lemma}
\begin{proof}
Source and target formula structures are comprehensive by Prop.~\ref{prop:iter:fmla}.
The entity infomorphism condition
$k(k_{1}){\;\models_{\widehat{\mathcal{E}}_{2}}\;}\varphi_{2}$
\underline{iff}
$k_{1}{\;\models_{\widehat{\mathcal{E}}_{1}}\;}\widehat{r}(\varphi_{2})$
holds for any source formula $\varphi_{2}{\,\in\,}\widehat{R}_{2}$ and target key $k_{1}{\,\in\,}K_{1}$
by Lem.~\ref{lem:all:fmla}.
\hfill\rule{5pt}{5pt}
\end{proof}
Let $\mathring{\mathrmbf{Struc}}(Y)$ denote the subcontext of comprehensive structures
with fixed data value set $Y$
whose structure morphisms have a surjective key function.
\begin{proposition}\label{struc:fmla}
There is a formula structure passage 
\[\mbox{\footnotesize{$
\mathring{\mathrmbf{Struc}}(Y)\xrightarrow{\mathring{\mathrmbfit{fmla}}_{Y}}\mathring{\mathrmbf{Struc}}(Y)
$,}\normalsize}\]
which is idempotent
$\mathring{\mathrmbfit{fmla}}_{Y}(\mathring{\mathrmbfit{fmla}}_{Y}(\mathcal{M})){\;\cong\;}\mathring{\mathrmbfit{fmla}}_{Y}(\mathcal{M})$:
a formula of formulas is another formula. 
\end{proposition}
%

\newpage
\subsection{Satisfaction.}\label{sub:sub:sec:sat} 







Satisfaction is a fundamental classification between formalism and semantics.
The atom of formalism used in satisfaction is the {\ttfamily FOLE} constraint,
whereas the atom of semantics used is the {\ttfamily FOLE} structure.
Satisfaction is defined in terms of formula interpretation 
(Eqn.~\ref{eqn:fmla:interp} of \S\ref{sub:sub:sec:fmla:struc}).
\footnote{Important definitions follow a logical order:
\emph{formula interpretation}
$\;\Rightarrow\;$ \emph{satisfaction}
$\;\Rightarrow\;$ \emph{institution}
$\;\Rightarrow\;$ \emph{structure interpretation}
$\;\Rightarrow\;$ \emph{sound logic interpretation}.}
%

\subsubsection{Sequent Satisfaction.}\label{sub:sub:sec:seq:sat} 

\begin{sloppypar}
An $\mathcal{S}$-structure $\mathcal{M} \in \mathrmbf{Struc}(\mathcal{S})$
\emph{satisfies} an $\mathcal{S}$-sequent $\varphi{\;\vdash\;}\psi$
(\S\ref{sub:sub:sec:seq})
when the interpretation widening of views asserted by the sequent actually holds in $\mathcal{M}$:
$\mathrmbfit{I}_{\mathcal{M}}(\varphi){\;\subseteq\;}\mathrmbfit{I}_{\mathcal{M}}(\psi)$.
Satisfaction is symbolized either by
$\mathcal{M}{\;\models_{\mathcal{S}}\;}(\varphi{\;\vdash\;}\psi)$
or by
$\varphi{\;\vdash_{\mathcal{M}}\;}\psi$.
For each $\mathcal{S}$-signature ${\langle{I,s}\rangle}$,
satisfaction defines the fiber order
$\mathcal{M}^{\mathcal{S}}(I,s)={\langle{\widehat{R}(I,s),\leq_{\mathcal{M}}}\rangle}$,
where
$\varphi{\;\leq_{\mathcal{M}}\;}\psi$
when
$\varphi{\;\vdash_{\mathcal{M}}\;}\psi$
for any two formulas $\varphi,\psi\in\widehat{R}(I,s)$.
Sequent satisfaction can be expressed in terms of implication as
$\top{\;\leq_{\mathcal{M}}\;}(\varphi{\,\rightarrowtriangle\,}\psi)$.
Thus, 
satisfaction of sequents is equivalent to satisfaction of formulas.
\end{sloppypar}
\begin{corollary}
Satisfaction in $\mathcal{M}$
is equivalent to
satisfaction in the image $\mathring{\mathcal{M}}$:
$\mathcal{M}{\;\models_{\mathcal{S}}\;}(\varphi{\;\vdash\;}\psi)$
\underline{iff}
$\mathring{\mathcal{M}}{\;\models_{\mathcal{S}}\;}(\varphi{\;\vdash\;}\psi)$.
\footnote{Satisfaction in $\mathcal{M}$ and its key-embedding structure $\dot{\mathcal{M}}$ should also be connected.}
%
%
\end{corollary}
\begin{proof}
By Prop.~\ref{prop:img:compre},
the formula interpretation in $\mathcal{M}$ and $\mathring{\mathcal{M}}$ are equal.
\end{proof}
For any $\mathcal{S}$-structure $\mathcal{M} \in \mathrmbf{Struc}(\mathcal{S})$,
the formula extent order 
$\mathrmbfit{ord}(\widehat{\mathcal{E}}) = 
{\langle{\widehat{R},\leq_{\widehat{\mathcal{E}}}}\rangle}$
is defined by
$\varphi{\;\leq_{\widehat{\mathcal{E}}}\;}\psi$
when
$\mathrmbfit{ext}_{\widehat{\mathcal{E}}}(\varphi){\;\subseteq\;}\mathrmbfit{ext}_{\widehat{\mathcal{E}}}(\psi)$.
For each signature ${\langle{I,s}\rangle}\in\mathrmbf{List}(X)$,
define the suborder
$\mathrmbfit{ord}_{\widehat{\mathcal{E}}}(I,s)
= {\langle{\widehat{R}(I,s),\leq_{\widehat{\mathcal{E}}}}\rangle}$.
%
\begin{corollary}
For an arbitrary $\mathcal{S}$-structure
the formula interpretation order 
is as strong as or stronger than the extent order 
of the formula classification:
$\varphi{\;\leq_{\mathcal{M}}\;}\psi\;\;\text{\underline{implies}}\;\;\varphi{\;\leq_{\widehat{\mathcal{E}}}\;}\psi$,
but not necessarily the converse.
\end{corollary}
\begin{proof}
Extent is the inverse image of interpretation:
$\mathrmbfit{ext}_{\widehat{\mathcal{E}}}(\varphi)=\tau^{-1}(\mathrmbfit{I}_{\mathcal{M}}(\varphi))$
for any formula $\varphi\in\widehat{R}$ (Eqn.\ref{ext:interp:equiv}).
Since inverse image is monotonic,
$\mathrmbfit{I}_{\mathcal{M}}(\varphi){\;\subseteq\;}\mathrmbfit{I}_{\mathcal{M}}(\psi)$
\underline{implies}
$\mathrmbfit{ext}_{\widehat{\mathcal{E}}}(\varphi){\;\subseteq\;}\mathrmbfit{ext}_{\widehat{\mathcal{E}}}(\psi)$.
\end{proof}
\begin{corollary}
For a comprehensive $\mathcal{S}$-structure,
the entity extent and interpretation fiber orders are identical 
$\mathrmbfit{ord}_{\widehat{\mathcal{E}}}(I,s){\;=\;}\mathcal{M}^{\mathcal{S}}(I,s)$,
since
$\varphi{\;\leq_{\widehat{\mathcal{E}}}\;}\psi$
\underline{iff}
$\varphi{\;\leq_{\mathcal{M}}\;}\psi$.
\end{corollary}
\begin{proof}
See footnote to Def.~\ref{def:compre:struc}.
\end{proof}
\begin{proposition}
Formal 
quantification and substitution are monotonic.
\end{proposition}
\begin{proof}
\begin{sloppypar}
The formal operators 
${\scriptstyle\sum}_{h}$, ${h}^{\ast}$ and ${\scriptstyle\prod}_{h}$
are monotonic,
since the semantic operators 
$\exists_{h}$, ${h}^{{\scriptscriptstyle-}1}$ and $\forall_{{h}}$
are monotonic.
We show the proof for formal existential quantification. 
{\emph{If ${\langle{I',s'}\rangle}\xrightarrow{h}{\langle{I,s}\rangle}$ is a signature morphism,
then the existential operator 
$\mathcal{M}^{\mathcal{S}}(I',s')
\xleftarrow{{\scriptscriptstyle\sum}_{h}}
\mathcal{M}^{\mathcal{S}}(I,s)$
is monotonic:}}
$\varphi{\;\leq_{\mathcal{M}}\;}\psi$
\underline{iff}
$\mathrmbfit{I}_{\mathcal{M}}(\varphi){\;\subseteq\;}\mathrmbfit{I}_{\mathcal{M}}(\psi)$
\underline{implies}
$\mathrmbfit{I}_{\mathcal{M}}({\scriptstyle\sum}_{h}(\varphi))
={\exists}_{h}(\mathrmbfit{I}_{\mathcal{M}}(\varphi))
\subseteq
{\exists}_{h}(\mathrmbfit{I}_{\mathcal{M}}(\psi))
=\mathrmbfit{I}_{\mathcal{M}}({\scriptstyle\sum}_{h}(\psi))$
\underline{iff}
${\scriptstyle\sum}_{h}(\varphi){\;\leq_{\mathcal{M}}\;}{\scriptstyle\sum}_{h}(\psi)$
for any two formulas $\varphi,\psi\in\widehat{R}(I,s)$.
\hfill\rule{5pt}{5pt}
\end{sloppypar}
\end{proof}

\subsubsection{Constraint Satisfaction.}\label{sub:sub:sec:constr:sat} 

\begin{sloppypar}
An $\mathcal{S}$-structure $\mathcal{M}\in\mathrmbf{Struc}(\mathcal{S})$
\emph{satisfies} 
an $\mathcal{S}$-constraint $\varphi'{\;\xrightarrow{h}\;}\varphi$
when $\mathcal{M}$ satisfies the sequent ${\scriptstyle\sum}_{h}(\varphi){\;\vdash\;}\varphi'$
\underline{iff} 
$\varphi'{\;\geq_{\mathcal{M}}\;}{\scriptstyle\sum}_{h}(\varphi)$
\underline{iff} 
$\mathrmbfit{I}_{\mathcal{M}}(\varphi'){\;\supseteq\;}
\mathrmbfit{I}_{\mathcal{M}}({\scriptstyle\sum}_{h}(\varphi))
={\exists}_{h}(\mathrmbfit{I}_{\mathcal{M}}(\varphi))$;
equivalently,
\footnote{See the formal/semantics reflection 
discussed in Sec.~\ref{sub:sub:sec:fml:interp}
and illustrated in Tbl.~\ref{tbl:fml-sem:refl}.}
when
$\mathcal{M}$ satisfies the sequent 
$\varphi{\;\vdash\;}{h}^{\ast}(\varphi')$
\underline{iff} 
${h}^{\ast}(\varphi'){\;\geq_{\mathcal{M}}\;}\varphi$
\underline{iff} 
${h}^{-1}(\mathrmbfit{I}_{\mathcal{M}}(\varphi'))
=\mathrmbfit{I}_{\mathcal{M}}({h}^{\ast}(\varphi'))
{\;\supseteq\;}\mathrmbfit{I}_{\mathcal{M}}(\varphi)$.
%
%
Satisfaction is symbolized by
$\mathcal{M}{\;\models_{\mathcal{S}}\;}(\varphi'{\;\xrightarrow{h}\;}\varphi)$.
Constraint satisfaction can be expressed in terms of implication as
$\top{\;\leq_{\widehat{\mathcal{E}}}\;}({\scriptstyle\sum}_{h}(\varphi){\,\rightarrowtriangle\,}\varphi')$;
equivalently,
$\top{\;\leq_{\widehat{\mathcal{E}}}\;}(\varphi{\,\rightarrowtriangle\,}{h}^{\ast}(\varphi'))$.
Thus, 
satisfaction of constraints is equivalent to satisfaction of formulas.
\end{sloppypar}

%
\begin{lemma}\label{lem:nat:cxt}
A structure $\mathcal{M}$
determines a mathematical context
$\mathcal{M}^{\mathcal{S}}{\;\sqsubseteq\;}\mathrmbf{Cons}(\mathcal{S})$,
called the \emph{conceptual intent} of $\mathcal{M}$,
whose objects are $\mathcal{S}$-formulas and
whose morphisms are $\mathcal{S}$-constraints satisfied by $\mathcal{M}$.
The ${\langle{I,s}\rangle}^{\text{th}}$ fiber is the order
${\mathcal{M}^{\mathcal{S}}(I,s)}^{\text{op}}
={\langle{\widehat{R}(I,s),\geq_{\mathcal{M}}}\rangle}$.
\footnote{The satisfaction relation 
corresponds to the ``truth classification'' in Barwise and Seligman~\cite{barwise:seligman:97},
where the conceptual intent $\mathcal{M}^{\mathcal{S}}$
corresponds to the ``theory of $\mathcal{M}$''.}
\end{lemma}
\begin{proof}
\begin{sloppypar}
$\mathcal{M}^{\mathcal{S}}$ is closed under constraint identities and constraint composition.
\newline
{\bfseries 1:} $\mathcal{M}{\;\models_{\mathcal{S}}\;}(\varphi{\;\xrightarrow{1}\;}\varphi)$,
and
{\bfseries 2:}
if $\mathcal{M}{\;\models_{\mathcal{S}}\;}(\varphi''{\;\xrightarrow{h'}\;}\varphi')$
and $\mathcal{M}{\;\models_{\mathcal{S}}\;}(\varphi'{\;\xrightarrow{h}\;}\varphi)$,
then $\mathcal{M}{\;\models_{\mathcal{S}}\;}(\varphi''{\;\xrightarrow{h'{\,\cdot\,}h}\;}\varphi)$,
since $\varphi''{\;\geq_{\mathcal{M}}\;}{\scriptstyle\sum}_{h'}(\varphi')$
and $\varphi'{\;\geq_{\mathcal{M}}\;}{\scriptstyle\sum}_{h}(\varphi)$
implies
$\varphi''{\;\geq_{\mathcal{M}}\;}
{\scriptstyle\sum}_{h'}({\scriptstyle\sum}_{h}(\varphi))={\scriptstyle\sum}_{h'{\,\cdot\,}h}(\varphi)$
using the monotonicity of the existential operator
$\mathcal{M}^{\mathcal{S}}(I'',s'')
\xleftarrow{{\scriptscriptstyle\sum}_{h'}}
\mathcal{M}^{\mathcal{S}}(I',s')$.
\hfill\rule{5pt}{5pt}
\end{sloppypar}
\end{proof}
There is an intent(ional) order between $\mathcal{S}$-structures.
Structure $\mathcal{M}_{2}$ is more general than structure $\mathcal{M}_{1}$,
symbolically 
$\mathcal{M}_{1}{\;\leq_{\mathcal{S}}\;}\mathcal{M}_{2}$,
when any constraint satisfied by $\mathcal{M}_{2}$ is also satisfied by $\mathcal{M}_{1}$:
$\mathcal{M}_{2}{\;\models_{\mathcal{S}}\;}(\varphi'{\;\xrightarrow{h}\;}\varphi)$
implies 
$\mathcal{M}_{1}{\;\models_{\mathcal{S}}\;}(\varphi'{\;\xrightarrow{h}\;}\varphi)$;
that is,
when
$\mathcal{M}_{1}^{\mathcal{S}} \sqsupseteq \mathcal{M}_{2}^{\mathcal{S}}$.
%
%

\comment{
\begin{corollary}\label{lem:nat:sdsg}
A structure $\mathcal{M}$ determines a signed domain diagram 
(\S\ref{append:sub:sub:sec:rel:db:cons})
${\langle{\mathcal{M}^{\mathcal{S}}\!,\mathrmbfit{S},\mathcal{A}}\rangle}$
(with constant type domain $\mathcal{A} = {\langle{X,Y,\models_{\mathcal{A}}}\rangle}$),
where the signature passage
$\mathcal{M}^{\mathcal{S}}\!\xrightarrow{\mathrmbfit{S}}\mathrmbf{List}(X)$
maps an $\mathcal{S}$-constraint 
$\varphi'{\,\xrightarrow{h\,}\,}\varphi$
to its underlying signature morphism
$\widehat{\sigma}(\varphi')={\langle{I',s'}\rangle}\xrightarrow{h}{\langle{I,s}\rangle}=\widehat{\sigma}(\varphi)$.
\footnote{Recall that the type domain $\mathcal{A}$ determines the tuple passage
$\mathrmbf{List}(X)^{\text{op}}\xrightarrow{\mathrmbfit{tup}_{\mathcal{A}}}\mathrmbf{Set}$.
Hence,
the composite
$(\mathcal{M}^{\mathcal{S}})^{\text{op}}\xrightarrow{\mathrmbfit{S}^{\text{op}}{\!\circ\,}\mathrmbfit{tup}_{\mathcal{A}}}\mathrmbf{Set}$
maps a constraint
$\varphi'{\,\xrightarrow{h\,}\,}\varphi$
to the tuple function
$\mathrmbfit{tup}_{\mathcal{A}}(I',s')\xleftarrow{\mathrmbfit{tup}_{\mathcal{A}}(h)}\mathrmbfit{tup}_{\mathcal{A}}(I,s)$.}
\end{corollary}
}
%


\newpage
\subsubsection{Institutional Aspect.}\label{sub:sub:sec:inst:asp}

%
For any schema $\mathcal{S}$,
the satisfaction classification
$\mathrmbf{Truth}(\mathcal{S})
= {\langle{\mathrmbf{Struc}(\mathcal{S}),\mathrmbf{Cons}(\mathcal{S}),\models_{\mathcal{S}}}\rangle}$
has $\mathcal{S}$-constraints $(\varphi'\xrightarrow{h}\varphi)\in\mathrmbf{Cons}(\mathcal{S})$ as types,
$\mathcal{S}$-structures $\mathcal{M}\in\mathrmbf{Struc}(\mathcal{S})$ as instances,
and satisfaction as the classification relation
$\mathcal{M}{\;\models_{\mathcal{S}}\;}(\varphi'\xrightarrow{h}\varphi)$.
\footnote{The satisfaction classification of the schema $\mathcal{S}$ 
corresponds to the ``truth classification'' of a first-order language 
in Barwise and Seligman~\cite{barwise:seligman:97}.}
In the propositions below,
we give the precise meaning of ``interpretation in first-order logic''
(Barwise and Seligman~\cite{barwise:seligman:97})
in terms of an infomorphism 
between truth classifications
\[\mbox{\footnotesize{$
{{\langle{\mathrmbfit{cons}_{{\langle{r,f}\rangle}},\mathrmbfit{struc}^{\curlywedge}_{{\langle{r,f}\rangle}}}}\rangle}:
\mathrmbf{Truth}(\mathcal{S}_{2})
\rightleftharpoons
\mathrmbf{Truth}(\mathcal{S}_{1})
$.}\normalsize}\]
\begin{proposition}\label{prop:ins}
The triple ${\langle{\mathrmbf{Sch},\mathrmbfit{cons},\mathrmbfit{struc}}\rangle}$ 
forms an institution (Goguen and Burstall \cite{goguen:burstall:92}),
where $\mathrmbf{Sch}$ is the context of schemas,
$\mathrmbf{Sch}{\;\xrightarrow{\mathrmbfit{cons}}\;}\mathrmbf{Cxt}$ is the indexed context of constraints
(\S\ref{sub:sub:sec:constr}), and
$\mathrmbf{Sch}^{\mathrm{op}}{\!\xrightarrow{\mathrmbfit{struc}^{\curlywedge}}\;}\mathrmbf{Cxt}$ is the indexed context of structures
%
(appendix of Kent~\cite{kent:fole:era:found}).
%
\end{proposition}
\begin{proof}
See the paper 
``The First-order Logical Environment'' 
(Kent~\cite{kent:iccs2013}).
\hfill\rule{5pt}{5pt}
\end{proof}
In an institution
``satisfaction is invariant under change of notation'':
for any schema morphism
$\mathcal{S}_{2}\xRightarrow{\;{\langle{r,f}\rangle}\;}\mathcal{S}_{1}$,
the following satisfaction condition holds:
{\footnotesize{
\begin{equation}\label{eqn:ins}
\mathrmbfit{struc}^{\curlywedge}_{{\langle{r,f}\rangle}}(\mathcal{M}_{1})
{\;\;\models_{\mathcal{S}_{2}}\;}
(\varphi_{2}'\xrightarrow{h}\varphi_{2})
\;\;\;\;\;\text{\underline{iff}}\;\;\;\;\;
\mathcal{M}_{1}
\models_{\mathcal{S}_{1}}
\mathrmbfit{cons}_{{\langle{r,f}\rangle}}(\varphi_{2}'{\;\vdash\;}\varphi_{2}).
\end{equation}
}\normalsize}
Equivalently, (see \S\ref{sub:sub:sec:specs} for the definition of specification flow) 
%
{\footnotesize{
\begin{equation}\label{eqn:ins:abs}
\mathrmbfit{struc}^{\curlywedge}_{{\langle{r,f}\rangle}}(\mathcal{M}_{1})^{\mathcal{S}_{2}}=
\overleftarrow{\mathrmbfit{spec}}_{{\langle{r,f}\rangle}}(\mathcal{M}_{1}^{\mathcal{S}_{1}}),
\end{equation}
}\normalsize}
%
%
the intent of the structure image is 
the specification image of the intent. 
\comment{
Hence,
any schema morphism 
$\mathcal{S}_{2}\stackrel{{\langle{r,f}\rangle}}{\Longrightarrow}\mathcal{S}_{1}$
defines an \emph{institution} infomorphism
\[\mbox{\footnotesize{$
\overset{\textstyle\mathrmbfit{ins}(\mathcal{S}_{2})}
{\overbrace
{\langle{\mathrmbfit{cons}(\mathcal{S}_{2}),\mathrmbfit{struc}(\mathcal{S}_{2}),\models_{\mathcal{S}_{2}}}\rangle}} 
\;\;\xrightleftharpoons{{\langle{\mathrmbfit{cons}_{{\langle{r,f}\rangle}},\mathrmbfit{struc}_{{\langle{r,f}\rangle}}}\rangle}}\;\; 
\overset{\textstyle\mathrmbfit{ins}(\mathcal{S}_{1})}
{\overbrace
{\langle{\mathrmbfit{cons}(\mathcal{S}_{1}),\mathrmbfit{struc}(\mathcal{S}_{1}),\models_{\mathcal{S}_{1}}}\rangle}} 
$.}\normalsize}\]
}

\begin{proposition}\label{prop:log:env}
The institution ${\langle{\mathrmbf{Sch},\mathrmbfit{cons},\mathrmbfit{struc}}\rangle}$ is a logical environment.
\end{proposition}
\begin{proof}
See the paper ``The First-order Logical Environment'' (Kent~\cite{kent:iccs2013}).
\hfill\rule{5pt}{5pt}
\end{proof}
A logical environment is an institution in which
``satisfaction respects structure morphisms'':
for any vertical structure morphism 
if $\mathcal{M}_{2}\xrightleftharpoons{{\langle{k,g}\rangle}}\mathcal{M}_{1}$
in the fiber context $\mathrmbf{Struc}(\mathcal{S})$ of a schema $\mathcal{S}$, 
the following satisfaction condition holds:
{\footnotesize{
\begin{equation}\label{eqn:log:env}
\mathcal{M}_{2}{\;\;\models_{\mathcal{S}_{2}}\;}(\varphi_{2}'\xrightarrow{h}\varphi_{2})
{\;\;\;\;\text{\underline{implies}}\;\;\;\;}
\mathcal{M}_{1}{\;\;\models_{\mathcal{S}_{2}}\;}(\varphi_{2}'\xrightarrow{h}\varphi_{2}).
\end{equation}
}\normalsize}
Equivalently,
we have the intent order
{\footnotesize{
\begin{equation}\label{eqn:log:env:abs}
\mathcal{M}_{2}^{\mathcal{S}}{\;\geq_{\mathcal{S}}\;}\mathcal{M}_{1}^{\mathcal{S}}
\end{equation}
}\normalsize}

\newpage

\begin{corollary}\label{cor:nat:pass}
A structure morphism
$\mathcal{M}_{2}\xrightleftharpoons{{\langle{r,k,f,g}\rangle}}\mathcal{M}_{1}$
determines an intent passage
$\mathcal{M}_{2}^{\mathcal{S}_{2}}
\xrightarrow{\;\mathrmbfit{int}_{{\langle{r,k,f,g}\rangle}}\;}
\mathcal{M}_{1}^{\mathcal{S}_{1}}
$,
which is a restriction of the constraint passage
$\mathrmbf{Cons}(\mathcal{S}_{2})\xrightarrow{\mathrmbfit{cons}_{{\langle{r,f}\rangle}}}\mathrmbf{Cons}(\mathcal{S}_{1})
:(\varphi_{2}'{\;\xrightarrow{h}\;}\varphi_{2})
\mapsto
(\widehat{r}(\varphi_{2}'){\;\xrightarrow{h}\;}\widehat{r}(\varphi_{2}))$.
%
\end{corollary}
\begin{proof}
The structure morphism
$\mathcal{M}_{2}\xrightleftharpoons{{\langle{r,k,f,g}\rangle}}\mathcal{M}_{1}$
factors as
\footnote{Proven in the appendix of Kent~\cite{kent:fole:era:found}.}
\[\mbox{\footnotesize{$
\mathcal{M}_{2}
\xrightleftharpoons{{\langle{k,g}\rangle}}\mathrmbfit{struc}^{\curlywedge}_{{\langle{r,f}\rangle}}(\mathcal{M}_{1})
\xrightleftharpoons{{\langle{r,f}\rangle}}\mathcal{M}_{1}
$.
}\normalsize}\]
Let $(\varphi_{2}'\xrightarrow{h}\varphi_{2})$ be any $\mathcal{S}_{2}$-constraint.
From Prop.~\ref{prop:log:env} we know that
\[\mbox{\footnotesize{$
\mathcal{M}_{2}{\;\;\models_{\mathcal{S}_{2}}\;}(\varphi_{2}'\xrightarrow{h}\varphi_{2})
{\;\;\;\;\text{\underline{implies}}\;\;\;\;}
\mathrmbfit{struc}^{\curlywedge}_{{\langle{r,f}\rangle}}(\mathcal{M}_{1}){\;\;\models_{\mathcal{S}_{2}}\;}(\varphi_{2}'\xrightarrow{h}\varphi_{2})
$.}\normalsize}\]
From Prop.~\ref{prop:ins} we know that
\[\mbox{\footnotesize{$
\mathrmbfit{struc}^{\curlywedge}_{{\langle{r,f}\rangle}}(\mathcal{M}_{1}){\;\;\models_{\mathcal{S}_{2}}\;}(\varphi_{2}'\xrightarrow{h}\varphi_{2})
{\;\;\;\;\text{\underline{iff}}\;\;\;\;}
\mathcal{M}_{1}\models_{\mathcal{S}_{1}}\mathrmbfit{cons}_{{\langle{r,f}\rangle}}(\varphi_{2}'{\;\vdash\;}\varphi_{2})
$.
}\normalsize}\]
\rule{5pt}{5pt}
\end{proof}
\begin{definition}\label{conc:int:cxt}
There is a conceptual intent passage 
\[\mbox{\footnotesize{$
\mathrmbf{Struc}\xrightarrow{\mathrmbfit{int}}\mathrmbf{Cxt}
:\mathcal{M}\mapsto\mathcal{M}^{\mathcal{S}}
$}\normalsize}\]
from structures to mathematical contexts.
\end{definition}


%% file: comps.tex
\newpage
\section{Architectural Components}\label{sub:sec:arch:comps}

\subsection{Specifications.}\label{sub:sec:specs}

\subsubsection{Specifications.}\label{sub:sub:sec:specs}

\paragraph{Consequence Relations.}

A {\ttfamily FOLE} consequence relation
(Barwise and Seligman \cite{barwise:seligman:97}) 
is a pair
${\langle{\mathcal{S},\vdash}\rangle}$,
where $\mathcal{S}$ is a schema and 
${\;\vdash\;}{\;\subseteq\;}\widehat{R}{\,\times}\widehat{R}$ 
is a set of $\mathcal{S}$-sequents;
that is, a binary relation on $\mathcal{S}$-formulas.
We want each sequent in a consequence relation
to assert logical entailment between component formulas. 
%
A consequence relation
can be used to represent and express all the subtyping relationships of a data model.
%
In the example illustrated in the {\ttfamily ERA} data model of the {\ttfamily FOLE} foundation paper \cite{kent:fole:era:found},
we might have the subtyping relationships 
$\bigl(\mathtt{Manager}{\;\vdash\,}\mathtt{Employee}\bigr)$ and
$\bigl(\mathtt{Engineering}{\;\vdash\,}\mathtt{Department}\bigr)$.
%

\paragraph{Specifications.}

Consequence relations only connect formulas within fibers: 
due to the common signature requirement on sequent components,
a {\ttfamily FOLE} consequence relation ${\langle{\mathcal{S},\vdash}\rangle}$
partitions into a collection of fiber consequence relations
$\bigl\{ \vdash_{{\langle{I,s}\rangle}}{\!\subseteq\;}\widehat{R}(I,s){\,\times}\widehat{R}(I,s) 
\mid {\langle{I,s}\rangle} \in \mathrmbf{List}(X) \bigr\}$
indexed by $\mathcal{S}$-signatures.
We now define a useful notion that connects formulas across fibers.

Given a schema $\mathcal{S}$,
an $\mathcal{S}$-specification is a subgraph $\mathrmbf{T}{\;\sqsubseteq\;}\mathrmbf{Cons}(\mathcal{S})$,
whose nodes are $\mathcal{S}$-formulas and whose edges are $\mathcal{S}$-constraints.
A consequence relation is a specification in which all constraints are sequents.
\comment{To be specific,
we assume that the node-set of a specification is minimal;
that is,
the node-set is the collection all source/target formulas of edges.
This is called the (resultant) node-set.
If we allowed larger node-sets,
we would get equivalent specifications.
{\fbox{{\bfseries{or}} we could regard isolated formulas as identity constraints.}}}
Let $\mathrmbf{Spec}(\mathcal{S})={\wp}\mathrmbf{Cons}(\mathcal{S})$ denote the set of all $\mathcal{S}$-specifications.
\footnote{For any graph $\mathcal{G}$,
${\wp}\mathcal{G} = {\langle{{\wp}\mathcal{G},\sqsubseteq}\rangle}$ 
denotes the power preorder of all subgraphs of $\mathcal{G}$.}
%
Although implicit,
we usually include the schema (language) in the symbolism,
so that a {\ttfamily FOLE} specification (presentation) 
$\mathcal{T}={\langle{\mathcal{S},\mathrmbf{T}}\rangle}$ is an indexed notion
consisting of a schema $\mathcal{S}$ and a $\mathcal{S}$-specification $\mathrmbf{T}{\;\in\;}\mathrmbf{Spec}(\mathcal{S})$.
We can place axiomatic restrictions on specifications
(and consequence relations)
in various manners.
A  {\ttfamily FOLE} specification requires entailment to be a preorder,
satisfying reflexivity and transitivity.
It also requires satisfaction of sufficient axioms 
(Tbl.~\ref{tbl:axioms})
to described the various logical operations 
(connectives, quantifiers, etc.) 
used to build formulas in first-order logic.
%

\paragraph{Specification Satisfaction.}

An $\mathcal{S}$-structure $\mathcal{M} \in \mathrmbf{Struc}(\mathcal{S})$
\emph{satisfies} (is a model of) 
an $\mathcal{S}$-specification $\mathrmbf{T}$,
symbolized
$\mathcal{M}{\;\models_{\mathcal{S}}\;}\mathrmbf{T}$,
when it satisfies every constraint in the specification:
$\mathcal{M}{\;\models_{\mathcal{S}}\;}\mathrmbf{T}$
\underline{iff}
$\mathcal{M}^{\mathcal{S}}\sqsupseteq\mathrmbf{T}$.
Hence,
the intent $\mathcal{M}^{\mathcal{S}}$ is the largest 
and most specialized
$\mathcal{S}$-specification satisfied by $\mathcal{M}$.
\footnote{$\mathcal{M}^{\mathcal{S}}$ is not just a mathematical context, but also an $\mathcal{S}$-specification.}
When specification order is defined below,
$\mathcal{M}_{1} \leq_{\mathcal{S}} \mathcal{M}_{2}$ (intentional order) 
is equivalent to
$\mathcal{M}_{1}^{\mathcal{S}} \leq_{\mathcal{S}} \mathcal{M}_{2}^{\mathcal{S}}$ (intent specification order).

\newpage
\subsubsection{Entailment and Consequence.}\label{sub:sub:sec:entail:cons}


Let $\mathcal{S} = {\langle{R,\sigma,X}\rangle}$ be a schema.
An $\mathcal{S}$-specification $\mathrmbf{T}$ entails an $\mathcal{S}$-constraint $(\varphi'{\;\xrightarrow{h}\;}\varphi)$,
symbolized by $\mathrmbf{T}{\;\vdash_{\mathcal{S}}\;}(\varphi'{\xrightarrow{h}\,}\varphi)$,
when any model of the specification satisfies the constraint:
$\mathcal{M}{\;\models_{\mathcal{S}}\;}\mathrmbf{T}$
implies
$\mathcal{M}{\;\models_{\mathcal{S}}\;}(\varphi'{\xrightarrow{h}\,}\varphi)$
for any $\mathcal{S}$-structure $\mathcal{M}$;
that is,
when
$\mathcal{M}^{\mathcal{S}}{\,\sqsupseteq\;}\mathrmbf{T}$
implies
$\mathcal{M}^{\mathcal{S}}{\,\ni\,}(\varphi'{\xrightarrow{h}\,}\varphi)$
for any $\mathcal{S}$-structure $\mathcal{M}$.
\footnote{In particular,
the conceptual intent entails a constraint iff it satisfies the constraint:
$\mathcal{M}^{\mathcal{S}}{\;\vdash_{\mathcal{S}}\;}(\varphi'{\xrightarrow{h}\,}\varphi)$
\underline{iff}
$\mathcal{M}{\;\models_{\mathcal{S}}\;}(\varphi'{\xrightarrow{h}\,}\varphi)$.}
The graph 
\[\mbox{\footnotesize{$
\mathrmbf{T}^{\scriptstyle\bullet} 
= \Bigl\{ \varphi'{\xrightarrow{h}\,}\varphi 
\mid \mathrmbf{T}{\;\vdash_{\mathcal{S}}\;}(\varphi'{\xrightarrow{h}\,}\varphi) \Bigr\}
= \bigsqcap_{\mathcal{S}} \Bigl\{ \mathcal{M}^{\mathcal{S}} \mid 
\mathcal{M}{\;\in\;}\mathrmbf{Struc}(\mathcal{S}), \mathcal{M}^{\mathcal{S}}{\,\sqsupseteq\;}\mathrmbf{T} \Bigr\}
$}\normalsize}\]
of all constraints entailed by a specification $\mathrmbf{T}$ is called its consequence.
The consequence 
$\mathrmbf{T}^{\scriptstyle\bullet}$
is a mathematical context,
since each conceptual intent 
$\mathcal{M}^{\mathcal{S}}$
is a mathematical context.
The consequence operator $(\mbox{-})^{\scriptstyle\bullet}$
is a closure operator
on specifications:
(increasing) $\mathrmbf{T}\sqsubseteq\mathrmbf{T}^{\scriptstyle\bullet}$;
(monotonic)  $\mathrmbf{T}_{1}\sqsubseteq\mathrmbf{T}_{2}$ 
implies $\mathrmbf{T}_{1}^{\scriptstyle\bullet}\sqsubseteq\mathrmbf{T}_{2}^{\scriptstyle\bullet}$; and
(idempotent) $\mathrmbf{T}^{\scriptstyle\bullet\bullet}=\mathrmbf{T}^{\scriptstyle\bullet}$.
%
%
Closure operators can be alternatively described as entailment relations
(Mossakowski, Diaconescu and Tarlecki~\cite{mossakowski:diaconescu:tarlecki:wlt}).
${\langle{\mathrmbf{Cons}(\mathcal{S}),\vdash_{\mathcal{S}}}\rangle}$
forms an entailment relation:
(reflexive) $\bigl\{(\varphi'{\xrightarrow{h}\,}\varphi)\bigr\}{\;\vdash_{\mathcal{S}}\;}(\varphi'{\xrightarrow{h}\,}\varphi)$
for any $\mathcal{S}$-constraint $\varphi'{\xrightarrow{h}\,}\varphi$;
(monotone) 
$\mathrmbf{T}_{1}{\;\sqsupseteq\;}\mathrmbf{T}_{2}$ and
$\mathrmbf{T}_{2}{\;\vdash_{\mathcal{S}}\;}(\varphi'{\xrightarrow{h}\,}\varphi)$
implies
$\mathrmbf{T}_{1}{\;\vdash_{\mathcal{S}}\;}(\varphi'{\xrightarrow{h}\,}\varphi)$; and
(transitive) 
if $\mathrmbf{T}_{1}^{\scriptstyle\bullet}\sqsupseteq\mathrmbf{T}_{2}$
and $\mathrmbf{T}_{1}{\,\sqcup\;}\mathrmbf{T}_{2}{\;\vdash_{\mathcal{S}}\;}(\varphi'{\xrightarrow{h}\,}\varphi)$,
then $\mathrmbf{T}_{1}{\;\vdash_{\mathcal{S}}\;}(\varphi'{\xrightarrow{h}\,}\varphi)$.
%

There is an intentional (concept lattice) entailment order between specifications that is implicit in satisfaction:
$\mathrmbf{T}_{1}{\;\leq_{\mathcal{S}}\;}\mathrmbf{T}_{2}$ when 
$\mathrmbf{T}_{1}^{\scriptstyle\bullet}{\;\sqsupseteq\;}\mathrmbf{T}_{2}^{\scriptstyle\bullet}$;
equivalently,
$\mathrmbf{T}_{1}^{\scriptstyle\bullet}{\;\sqsupseteq\;}\mathrmbf{T}_{2}$.
This is a specialization-generalization order;
$\mathrmbf{T}_{1}$ is more specialized than $\mathrmbf{T}_{2}$, 
and $\mathrmbf{T}_{2}$ is more generalized than $\mathrmbf{T}_{1}$.
We symbolize this preorder by 
$\mathrmbf{Spec}(\mathcal{S})={\langle{\mathrmbf{Spec}(\mathcal{S}),\leq_{\mathcal{S}}}\rangle}$.
Intersections and unions define joins and meets,
with the bottom specification being the empty join $\bot_{\mathcal{S}} = \bigcap\emptyset = \mathrmbf{Cons}(\mathcal{S})$
and the top specification being the empty meet $\top_{\mathcal{S}} = \bigcup\emptyset = \emptyset$.
\footnote{The nodes (formulas) in a specification can be identified with identity constraints.
Identity constraints added to or subtracted from a specification give an equivalent specification.
Hence,
we can assume the node-set of formulas of a specification is included in the edge-set of constraints of a specification.
With that assumption,
boolean operations need only work on the edge-set of a specification.}
%
Any specification $\mathrmbf{T}$ is entailment equivalent to its consequence 
$\mathrmbf{T}\equiv\mathrmbf{T}^{\scriptstyle\bullet}$.
A specification $\mathrmbf{T}$ is said to be closed when it is equal to its consequence 
$\mathrmbf{T}=\mathrmbf{T}^{\scriptstyle\bullet}$.
An $\mathcal{S}$-specification $\mathrmbf{T}$ is \emph{consistent} when
some $\mathcal{S}$-structure $\mathcal{M}$ satisfies $\mathrmbf{T}$:
$\mathcal{M}{\;\models_{\mathcal{S}}\;}\mathrmbf{T}$
or
$\bot_{\mathcal{S}}{\;<_{\mathcal{S}}\;}\mathcal{M}^{\mathcal{S}}{\;\leq_{\mathcal{S}}\;}\mathrmbf{T}$.
It is inconsistent otherwise.
Hence,
an $\mathcal{S}$-specification $\mathrmbf{T}$ is inconsistent
when
$\mathrmbfit{T}^{\scriptstyle\bullet} 
= \mathrmbf{Cons}(\mathcal{S})
= \bot_{\mathcal{S}}$. 


\newpage
\subsubsection{Specification Flow.}\label{sub:sub:sec:spec:flow}

\paragraph{Specification Flow.}

Specifications can be moved along schema morphisms.
Given 
$\mathcal{S}_{2}\xRightarrow{{\langle{r,f}\rangle}}\mathcal{S}_{1}$,
{\em direct flow} is the direct image operator
\[\mbox{\footnotesize{$
\overrightarrow{\mathrmbfit{spec}}_{{\langle{r,f}\rangle}} = 
{\wp}\mathrmbfit{cons}_{{\langle{r,f}\rangle}}
: \mathrmbf{Spec}(\mathcal{S}_{2})
\rightarrow 
\mathrmbf{Spec}(\mathcal{S}_{1})
$}\normalsize}\]
and 
{\em inverse flow} is the inverse image operator (with consequence)
\[\mbox{\footnotesize{$
\overleftarrow{\mathrmbfit{spec}}_{{\langle{r,f}\rangle}} = 
\mathrmbfit{cons}_{{\langle{r,f}\rangle}}^{-1}((\mbox{-})^{\scriptscriptstyle\bullet})
: \mathrmbf{Spec}(\mathcal{S}_{2}) 
\leftarrow 
\mathrmbf{Spec}(\mathcal{S}_{1})
$}\normalsize}\]
along the constraint passage
$\mathrmbf{Cons}(\mathcal{S}_{2})\xrightarrow{\mathrmbfit{cons}_{{\langle{r,f}\rangle}}}\mathrmbf{Cons}(\mathcal{S}_{1})$.
%
\comment{
\footnote{Let
$\mathcal{G}_{2}={\langle{E_{2},s_{2},t_{2},N_{2}}\rangle}
\xrightarrow{{\langle{f,g}\rangle}}
{\langle{E_{1},s_{1},t_{1},N_{1}}\rangle}=\mathcal{G}_{1}$
be a graph morphism.
The direct image monotonic function
${\wp}\mathcal{G}_{2}\xrightarrow{{\wp}{\langle{f,g}\rangle}}{\wp}\mathcal{G}_{1}$
is defined component-wise:
it maps
a source subgraph
$\mathcal{G}={\langle{E,s_{2},t_{2},N}\rangle}\sqsubseteq\mathcal{G}_{2}$
to the target subgraph
${\wp}{\langle{f,g}\rangle}(\mathcal{G})\doteq{\langle{{\wp}f(E),s_{1},t_{1},{\wp}g(N)}\rangle}\sqsubseteq\mathcal{G}_{1}$.
The inverse image monotonic function
${\wp}\mathcal{G}_{2}\xleftarrow{{\langle{f,g}\rangle}^{\scriptscriptstyle{-\!1}}}{\wp}\mathcal{G}_{1}$
is also defined component-wise:
it maps
a target subgraph
$\mathcal{G}={\langle{E,s_{1},t_{1},N}\rangle}\sqsubseteq\mathcal{G}_{1}$
to the source subgraph
${\langle{f,g}\rangle}^{-1}(\mathcal{G})\doteq{\langle{f^{-1}(E),s_{2},t_{2},g^{-1}(N)}\rangle}\sqsubseteq\mathcal{G}_{2}$.
Inverse image maps contexts to contexts.
Apply to the formula passage
$\mathrmbf{Cons}(\mathcal{S}_{2})\xrightarrow{\mathrmbfit{cons}_{{\langle{r,f}\rangle}}}\mathrmbf{Cons}(\mathcal{S}_{1})$.}
}
Properties satisfied:
\begin{itemize}
{\footnotesize{
\item 
inverse images are closed,
\hfill
$\overleftarrow{\mathrmbfit{spec}}_{{\langle{r,f}\rangle}}(\mathrmbf{T}_{1})^{\scriptstyle\bullet}
{\;=\;\;\,}\overleftarrow{\mathrmbfit{spec}}_{{\langle{r,f}\rangle}}(\mathrmbf{T}_{1})$;
%
\item 
direct image commutes with consequence,
\hfill
$\overrightarrow{\mathrmbfit{spec}}_{{\langle{r,f}\rangle}}(\mathrmbf{T}_{2})^{\scriptstyle\bullet} 
=\overrightarrow{\mathrmbfit{spec}}_{{\langle{r,f}\rangle}}(\mathrmbf{T}_{2}^{\scriptstyle\bullet})^{\scriptstyle\bullet}$; 
%
\item 
direct image
is monotonic,
\hfill
$\mathrmbf{T}_{2}{\;\leq_{\mathcal{S}_{2}}}\mathrmbf{T}_{2}'$ 
 \underline{implies} 
$\overrightarrow{\mathrmbfit{spec}}_{{\langle{r,f}\rangle}}(\mathrmbf{T}_{2}){\;\leq_{\mathcal{S}_{1}}}\overrightarrow{\mathrmbfit{spec}}_{{\langle{r,f}\rangle}}(\mathrmbf{T}_{2}')$. 
}}
\end{itemize}
These are adjoint monotonic functions w.r.t. specification order:
\[\mbox
{\footnotesize{
$\overrightarrow{\mathrmbfit{spec}}_{{\langle{r,f}\rangle}}(\mathrmbf{T}_{2})\geq_{\mathcal{S}_{1}}\mathrmbf{T}_{1} 
\text{ \underline{iff} } 
\mathrmbf{T}_{2}\geq_{\mathcal{S}_{2}}\overleftarrow{\mathrmbfit{spec}}_{{\langle{r,f}\rangle}}(\mathrmbf{T}_{1})$
}\normalsize}
\]
so that 
direct image $\overrightarrow{\mathrmbfit{spec}}_{{\langle{r,f}\rangle}}$ preserves all lattice meets $\bigwedge=\bigcup$ and
inverse image $\overleftarrow{\mathrmbfit{spec}}_{{\langle{r,f}\rangle}}$ preserves all lattice joins $\bigvee=\bigcap$.
%

\begin{sloppypar}
The constraint passage 
$\mathrmbf{Cons}(\mathcal{S}_{2})\xrightarrow{\mathrmbfit{cons}_{{\langle{r,f}\rangle}}}\mathrmbf{Cons}(\mathcal{S}_{1})$
is a closure operator morphism,
since direct image commutes with consequence:
$\overrightarrow{\mathrmbfit{spec}}_{{\langle{r,f}\rangle}}(\mathrmbf{T}_{2}^{\scriptstyle\bullet})
\subseteq
\overrightarrow{\mathrmbfit{spec}}_{{\langle{r,f}\rangle}}(\mathrmbf{T}_{2})^{\scriptstyle\bullet}$
for any $\mathcal{S}_{2}$-specification $\mathrmbf{T}_{2}$.
Morphisms of closure operators can be alternatively described as morphisms of entailment relations
(Mossakowski, Diaconescu and Tarlecki~\cite{mossakowski:diaconescu:tarlecki:wlt}).
The constraint passage 
$\mathrmbf{Cons}(\mathcal{S}_{2})\xrightarrow{\mathrmbfit{cons}_{{\langle{r,f}\rangle}}}\mathrmbf{Cons}(\mathcal{S}_{1})$
forms a morphism of entailment relations,
since 
$\mathrmbf{T}_{2}{\;\vdash_{\mathcal{S}_{2}}\;}(\varphi_{2}'{\xrightarrow{h}\,}\varphi_{2})$
implies
$\overrightarrow{\mathrmbfit{spec}}_{{\langle{r,f}\rangle}}(\mathrmbf{T}_{2}){\;\vdash_{\mathcal{S}_{1}}\;}\mathrmbfit{cons}_{{\langle{r,f}\rangle}}(\varphi_{2}'{\xrightarrow{h}\,}\varphi_{2})$
for any $\mathcal{S}_{2}$-specification $\mathrmbf{T}_{2}$
and any $\mathcal{S}_{2}$-constraint $(\varphi_{2}'{\xrightarrow{h}\,}\varphi_{2})$
by direct image monotonicity.
%
\end{sloppypar}

\paragraph{Specification Morphisms.}

A specification morphism
${\langle{\mathcal{S}_{2},\mathrmbf{T}_{2}}\rangle}\xrightarrow{{\langle{r,f}\rangle}}{\langle{\mathcal{S}_{1},\mathrmbf{T}_{1}}\rangle}$
is a schema morphism 
$\mathcal{S}_{2}\xRightarrow{{\langle{r,f}\rangle}}\mathcal{S}_{1}$
that preserves entailment:
\newline\mbox{}\hfill$
\mathrmbf{T}_{2}{\;\vdash_{\mathcal{S}_{2}}\;}(\varphi_{2}'{\xrightarrow{h}\,}\varphi_{2})
\;\text{implies}\;
\mathrmbf{T}_{1}{\;\vdash_{\mathcal{S}_{1}}\;}\mathrmbfit{cons}_{{\langle{r,f}\rangle}}(\varphi_{2}'{\xrightarrow{h}\,}\varphi_{2})
$\hfill\mbox{}\newline
%
for any $\mathcal{S}_{2}$-constraint $(\varphi_{2}'{\xrightarrow{h}\,}\varphi_{2})$;
or more concisely,
%
\[\mbox{\footnotesize{
$\overrightarrow{\mathrmbfit{spec}}_{{\langle{r,f}\rangle}}(\mathrmbf{T}_{2}^{\scriptstyle\bullet}){\;\sqsubseteq\;}\mathrmbf{T}_{1}^{\scriptstyle\bullet}
\;\;\text{\underline{iff}}\;\;
\overrightarrow{\mathrmbfit{spec}}_{{\langle{r,f}\rangle}}(\mathrmbf{T}_{2}){\;\sqsubseteq\;}\mathrmbf{T}_{1}^{\scriptstyle\bullet}$.
}\normalsize}\]
%
%
Equivalently,
that maps the source specification to a generalization of the target specification
$\overrightarrow{\mathrmbfit{spec}}_{{\langle{r,f}\rangle}}(\mathrmbf{T}_{2})\geq_{\mathcal{S}_{1}}\mathrmbf{T}_{1}$
\underline{or}
that maps the target specification to a specialization of the source specification
$\mathrmbf{T}_{2}\geq_{\mathcal{S}_{2}}\overleftarrow{\mathrmbfit{spec}}_{{\langle{r,f}\rangle}}(\mathrmbf{T}_{1})$.
%

%
%
The fibered mathematical context of specifications $\mathrmbf{Spec}$ 
has specifications as objects and specification morphisms as morphisms.
Thus,
the fibered context of specifications $\mathrmbf{Spec}$ is defined in terms of formal information flow.
There is an underlying schema passage ${\langle{\mathcal{S},\mathrmbf{T}}\rangle}\mapsto\mathcal{S}$
from specifications to schemas 
\newline\newline\mbox{}\rule{30pt}{0pt}\hfill{
$\mathrmbf{Spec}\xrightarrow{\mathrmbfit{sch}}\mathrmbf{Sch}$.
}\hfill{
(Fig.~\ref{fbr:ctx})
}\newline

\paragraph{Extending Order.}

We regard a schema morphism 
$\mathcal{S}_{2}\xRightarrow{{\langle{r,f}\rangle}}\mathcal{S}_{1}$
to be a translation device:
any $\mathcal{S}_{2}$-formula $\varphi$ with signature ${\langle{I,s}\rangle}$ 
is translated to
an $\mathcal{S}_{1}$-formula $\hat{r}(\varphi)$ with signature ${\scriptstyle\sum}_{f}(I,s)$.
This notion of a translation device is embodied in the direct image operator
$\mathrmbf{Spec}(\mathcal{S}_{2})
\xrightarrow{\overrightarrow{\mathrmbfit{spec}}_{{\langle{r,f}\rangle}}}
\mathrmbf{Spec}(\mathcal{S}_{1})$.
A specification morphism
$\mathcal{T}_{2}={\langle{\mathcal{S}_{2},\mathrmbf{T}_{2}}
\rangle}\xrightarrow{{\langle{r,f}\rangle}}
{\langle{\mathcal{S}_{1},\mathrmbf{T}_{1}}\rangle}=\mathcal{T}_{2}$
from source specification
$\mathrmbf{T}_{2}\in\mathrmbf{Spec}(\mathcal{S}_{2})$
to target specification
$\mathrmbf{T}_{1}\in\mathrmbf{Spec}(\mathcal{S}_{1})$
is a schema morphism 
$\mathcal{S}_{2}\xRightarrow{{\langle{r,f}\rangle}}\mathcal{S}_{1}$
that translates the source specification to a generalization of the target specification
$\overrightarrow{\mathrmbfit{spec}}_{{\langle{r,f}\rangle}}(\mathrmbf{T}_{2})\geq_{\mathcal{S}_{1}}\mathrmbf{T}_{1}$.
Thus,
we interpret a specification morphism as a link between two specifications,
where the source specification is more general than the target specification,
symbolized as 
$\mathrmbf{T}_{2}{\,\succcurlyeq\,}\mathrmbf{T}_{1}$.
In this way,
we interpret 
the context of specifications $\mathrmbf{Spec}$
to be an extension of $\mathrmbfit{fbr}(\mathcal{S})=\mathrmbf{Spec}(\mathcal{S})^{\text{op}}$,
the opposite of the specification order at some schema $\mathcal{S}$.
This idea is used in later papers 
(mentioned in \S\ref{sec:intro})
to motivate the extension of information systems along indexing passages,
which is an integral component in the definition of system morphisms
and in the extension of the ideas of conservative extension and modularity to the level of systems.


%
\begin{lemma}\label{lem:nat:spec:mor}
A structure morphism
$\mathcal{M}_{2}\xrightleftharpoons{{\langle{r,k,f,g}\rangle}}\mathcal{M}_{1}$
with  schema morphism
$\mathcal{S}_{2}\xRightarrow{{\langle{r,f}\rangle}}\mathcal{S}_{1}$
determines an intent specification morphism
${\langle{\mathcal{S}_{2},\mathcal{M}_{2}^{\mathcal{S}_{2}}}\rangle}
\xrightarrow{\langle{r,f}\rangle}
{\langle{\mathcal{S}_{1},\mathcal{M}_{1}^{\mathcal{S}_{1}}}\rangle}
$,
which is a restriction of the constraint passage
$\mathrmbf{Cons}(\mathcal{S}_{2})\xrightarrow{\mathrmbfit{cons}_{{\langle{r,f}\rangle}}}\mathrmbf{Cons}(\mathcal{S}_{1})
:(\varphi_{2}'{\;\xrightarrow{h}\;}\varphi_{2})
\mapsto
(\widehat{r}(\varphi_{2}'){\;\xrightarrow{h}\;}\widehat{r}(\varphi_{2}))$ 
(compare this to Cor.~\ref{cor:nat:pass}).
\end{lemma}
\begin{proof}
The structure morphism
$\mathcal{M}_{2}\xrightleftharpoons{{\langle{r,k,f,g}\rangle}}\mathcal{M}_{1}$
factors as
\footnote{Proven in the appendix of Kent~\cite{kent:fole:era:found}.}
\[\mbox{\footnotesize{$
\mathcal{M}_{2}
\xrightleftharpoons{{\langle{k,g}\rangle}}\mathrmbfit{struc}^{\curlywedge}_{{\langle{r,f}\rangle}}(\mathcal{M}_{1})
\xrightleftharpoons{{\langle{r,f}\rangle}}\mathcal{M}_{1}
$.
}\normalsize}\]
%
Hence,
\[\mbox{\footnotesize{$
\mathcal{M}^{\mathcal{S}_{2}}_{2}
\stackrel{
\overset{\text{Eqn.}}{\scriptscriptstyle{\text{\ref{eqn:log:env:abs}}}}
}{\;\geq_{\mathcal{S}_{2}}\;}
\mathrmbfit{struc}^{\curlywedge}_{{\langle{r,f}\rangle}}(\mathcal{M}_{1})^{\mathcal{S}_{2}}
\stackrel{
\overset{\text{Eqn.}}{\scriptscriptstyle{\text{\ref{eqn:ins:abs}}}}
}{=}
\overleftarrow{\mathrmbfit{spec}}_{{\langle{r,f}\rangle}}(\mathcal{M}_{1}^{\mathcal{S}_{1}})
$.}\normalsize}\]
\rule{5pt}{5pt}
\end{proof}

\comment{
Alternate proof:
$\overrightarrow{\mathrmbfit{spec}}_{{\langle{r,f}\rangle}}(\mathcal{M}_{2}^{\mathcal{S}_{2}})
\geq_{\mathcal{S}_{1}}
\mathcal{M}_{1}^{\mathcal{S}_{1}}$,
since
$(\varphi'{\;\xrightarrow{h}\;}\varphi)\in\mathcal{M}_{2}^{\mathcal{S}_{2}}$
\underline{iff}
$\mathcal{M}_{2}{\;\models_{\mathcal{S}_{2}}\;}(\varphi'{\;\xrightarrow{h}\;}\varphi)$
\underline{iff}
$\varphi'{\;\leq_{\widehat{\mathcal{R}}_{2}}\;}{\scriptstyle\sum}_{h}(\varphi)$
\underline{$\stackrel{\widehat{r}\;\text{mono}}{\text{implies}}$}
$\widehat{r}(\varphi'){\;\leq_{\widehat{\mathcal{R}}_{1}}\;}
\widehat{r}({\scriptstyle\sum}_{h}(\varphi))
={\scriptstyle\sum}_{h}(\widehat{r}(\varphi))$
\underline{iff}
$\mathcal{M}_{1}{\;\models_{\mathcal{S}_{1}}\;}(\widehat{r}(\varphi'){\;\xrightarrow{h}\;}\widehat{r}(\varphi))$
\underline{iff}
$\mathrmbfit{cons}_{{\langle{r,f}\rangle}}(\varphi'{\;\xrightarrow{h}\;}\varphi)=
(\widehat{r}(\varphi'){\;\xrightarrow{h}\;}\widehat{r}(\varphi))\in\mathcal{M}_{1}^{\mathcal{S}_{1}}$.
}

\begin{definition}\label{conc:int:spec}
There is a conceptual intent passage
\footnote{Compare this to the conceptual intent passage 
$\mathrmbf{Struc}\xrightarrow{\mathrmbfit{int}}\mathrmbf{Cxt}$
in Def.~\ref{conc:int:cxt}.}
\newline\newline\mbox{}\rule{30pt}{0pt}\hfill{
$\mathrmbf{Struc}\xrightarrow{\mathrmbfit{int}}\mathrmbf{Spec}
:\mathcal{M}\mapsto{\langle{\mathcal{S},\mathcal{M}^{\mathcal{S}}}\rangle}$
}\hfill{
(Fig.~\ref{fbr:ctx})
}\newline\newline
from structures to specifications,
which respects schema $\mathrmbfit{int}{\;\circ\;}\mathrmbfit{sch} = \mathrmbfit{sch}$. 
\footnote{This defines the lower part of the {\ttfamily FOLE} superstructure
(Fig.~\ref{fbr:ctx}).}
\end{definition}
%

\newpage
\subsubsection{Legacy Notions.}\label{sub:sub:sec:legacy}

\paragraph{Conservative Extensions.}
%
In mathematical logic, 
theory $T_{1}$ is a 
conservative extension of theory $T_{2}$ 
when 
(1) the language of $T_{1}$ extends the language of $T_{2}$,
(2) every theorem of $T_{2}$ is a theorem of $T_{1}$, and 
(3) any theorem of $T_{1}$ that is in the language of $T_{2}$ is already a theorem of $T_{2}$.
%
We translate this definition to {\ttfamily FOLE} morphisms. 
Given 
(1) 
a schema morphism
$\mathcal{S}_{2}\xRightarrow{\,\sigma\;}\mathcal{S}_{1}$
extending the language of $\mathcal{S}_{2}$ to the language of $\mathcal{S}_{1}$,
a target specification $\mathrmbf{T}_{1}\in\mathrmbf{Spec}(\mathcal{S}_{1})$
is a conservative extension of 
a source specification $\mathrmbf{T}_{2}\in\mathrmbf{Spec}(\mathcal{S}_{2})$ 
when the following conditions hold:
(2)
the image of any constraint entailed by $\mathrmbf{T}_{2}$ is entailed by $\mathrmbf{T}_{1}$
with
$\overrightarrow{\mathrmbfit{spec}}_{\sigma}(\mathrmbf{T}_{2})\geq_{\mathcal{S}_{1}}\mathrmbf{T}_{1}$,
so that
$\mathrmbf{T}_{2}\xrightarrow{\sigma}\mathrmbf{T}_{1}$
is a specification morphism;
and
(3)
any constraint 
whose image is entailed by $\mathrmbf{T}_{1}$ 
is already entailed by $\mathrmbf{T}_{2}$
with
$\mathrmbf{T}_{2}\leq_{\mathcal{S}_{2}}\overleftarrow{\mathrmbfit{spec}}_{\sigma}(\mathrmbf{T}_{1})$,
so that 
$\mathrmbf{T}_{2}$ is closed along $\mathcal{S}_{2}\xRightarrow{\,\sigma\;}\mathcal{S}_{1}$ 
with
$\mathrmbf{T}_{2}^{\scriptscriptstyle\blacklozenge_{\sigma}}{\;\equiv_{\mathcal{S}_{2}}\;}\mathrmbf{T}_{2}$.
\footnote{Morphic closure is defined at the system level in Kent \cite{kent:iccs2009}.}
Conservative extensions are closed under composition.

\paragraph{Consistency.}

A conservative extension of a consistent specification is consistent,
since the inverse image operator preserves concept lattice joins,
mapping the empty join to the empty join:
$\mathrmbf{T}_{1}{\;\equiv_{\mathcal{S}_{1}}\;}\bot_{\mathcal{S}_{1}}$ 
implies
$\mathrmbf{T}_{2}{\;\equiv_{\mathcal{S}_{2}}\;}\overleftarrow{\mathrmbfit{spec}}_{\sigma}(\mathrmbf{T}_{1}) 
{\;\equiv_{\mathcal{S}_{2}}\;} \bot_{\mathcal{S}_{2}}$.
In contrast,
the specification component of \underline{any} 
structure morphism
preserves soundness (specific consistency).
\comment{
\begin{quotation}
{\bfseries{
So how much like conservative extensions
are the (intent) specification components of structure morphisms, or
the specification components of sound logic morphisms?}}
\end{quotation}
}
\begin{corollary}\label{cor:pres:cons}
Structure morphisms preserve soundness (specific consistency) (\S\ref{sub:sec:snd:log}):
for any structure morphism
$\mathcal{M}_{2}\xrightleftharpoons{{\langle{r,k,f,g}\rangle}}\mathcal{M}_{1}$
with  schema morphism
$\mathcal{S}_{2}\xRightarrow{{\langle{r,f}\rangle}}\mathcal{S}_{1}$,
if $\mathcal{S}_{2}$-structure $\mathcal{M}_{2}$ satisfies $\mathrmbf{T}_{2}$,
then $\mathcal{S}_{1}$-structure $\mathcal{M}_{1}$ satisfies $\overrightarrow{\mathrmbfit{spec}}_{{\langle{r,f}\rangle}}(\mathrmbf{T}_{2})$.
\end{corollary}
\begin{proof}
\mbox{}
\vspace{-10pt}
\begin{center}
$\begin{array}{@{\hspace{25pt}}c@{\hspace{-5pt}}c@{\hspace{-5pt}}c}
\mathcal{S}_{2}
&
\xRightarrow{{\langle{r,f}\rangle}}
&
{\hspace{-80pt}}\mathcal{S}_{1}
\\
\mathcal{M}_{2}^{\mathcal{S}_{2}}
\stackrel{\scriptscriptstyle{\text{sound}}}{\;\leq_{\mathcal{S}_{2}}\,}
\mathrmbf{T}_{2}
&&
\mathcal{M}_{1}^{\mathcal{S}_{1}}
\stackrel{\scriptscriptstyle{\text{Lem.~\ref{lem:nat:spec:mor}}}}{\;\leq_{\mathcal{S}_{1}}\,}
\overrightarrow{\mathrmbfit{spec}}_{{\langle{r,f}\rangle}}(\mathcal{M}_{2}^{\mathcal{S}_{2}})
\stackrel{\overset{\text{mono}}{\scriptscriptstyle{\text{tonic}}}}{\;\leq_{\mathcal{S}_{1}}\;}
\overrightarrow{\mathrmbfit{spec}}_{{\langle{r,f}\rangle}}(\mathrmbf{T}_{2})
\\\\
\end{array}$
\end{center}
\end{proof}
\begin{corollary}\label{cor:refl:cons}
Reductions reflect consistency:
if $\mathcal{M}_{2}=\mathrmbfit{struc}^{\curlywedge}_{{\langle{r,f}\rangle}}(\mathcal{M}_{1})$
is the reduct of $\mathcal{M}_{1}$
with underlying schema morphism
$\mathcal{S}_{2}\xRightarrow{{\langle{r,f}\rangle}}\mathcal{S}_{1}$ and 
$\mathcal{S}_{1}$-structure $\mathcal{M}_{1}$ satisfies $\mathrmbf{T}_{1}$,
then $\mathcal{S}_{2}$-structure $\mathcal{M}_{2}$ satisfies $\overleftarrow{\mathrmbfit{spec}}_{{\langle{r,f}\rangle}}(\mathrmbf{T}_{1})$.
\end{corollary}
\begin{proof}
\[\mbox{\footnotesize{$
\mathcal{M}^{\mathcal{S}_{2}}_{2}
\overset{\text{reduct}}{=}
\mathrmbfit{struc}^{\curlywedge}_{{\langle{r,f}\rangle}}(\mathcal{M}_{1})^{\mathcal{S}_{2}}
\stackrel{\overset{\text{Eqn.}}{\scriptscriptstyle{\text{\ref{eqn:ins:abs}}}}}{=}
\overleftarrow{\mathrmbfit{spec}}_{{\langle{r,f}\rangle}}(\mathcal{M}_{1}^{\mathcal{S}_{1}})
\stackrel{\overset{\text{mono}}{\scriptscriptstyle{\text{tonic}}}}{\;\leq_{\mathcal{S}_{2}}\;}
\overleftarrow{\mathrmbfit{spec}}_{{\langle{r,f}\rangle}}(\mathrmbf{T}_{1})
$.}\normalsize}\]
\end{proof}
\newpage
\subsection{Logics}\label{sub:sec:log}

\subsubsection{Logics.}\label{sub:sub:sec:log:obj}

A logic $\mathcal{L}={\langle{\mathcal{S},\mathcal{M},\mathrmbf{T}}\rangle}$
consists of 
a structure $\mathcal{M}$ and
a specification $\mathcal{T}={\langle{\mathcal{S},\mathrmbf{T}}\rangle}$
that share 
a common schema $\mathcal{S}$.
%
%
For any fixed structure $\mathcal{M}$,
the set of all logics $\mathrmbf{Log}(\mathcal{M})$
with that structure is a preordered set under the specification order:
${\langle{\mathcal{S},\mathcal{M},\mathrmbf{T}_{1}}\rangle}{\;\leq_{\mathcal{M}}\;}{\langle{\mathcal{S},\mathcal{M},\mathrmbf{T}_{2}}\rangle}$ 
when 
$\mathrmbf{T}_{1}\leq_{\mathcal{S}}\mathrmbf{T}_{2}$.
%
Note that
\[\mbox{\footnotesize{$
\bot_{\mathcal{M}}
={\langle{\mathcal{S},\mathcal{M},\mathrmbf{Cons}(\mathcal{S})}\rangle}
{\;\leq_{\mathcal{M}}\;}\mathcal{L}
{\;\leq_{\mathcal{M}}\;}{\langle{\mathcal{S},\mathcal{M},\emptyset}\rangle}
=\top_{\mathcal{M}}
$}\normalsize}\]
%
for any logic $\mathcal{L}={\langle{\mathcal{S},\mathcal{M},\mathrmbf{T}}\rangle}$,
since
$\mathrmbf{Cons}(\mathcal{S})
=
\mathrmbf{Cons}(\mathcal{S})^{\scriptscriptstyle\bullet}{\;\sqsupseteq\;}\mathrmbf{T}\sqsupseteq\emptyset$.
For any structure $\mathcal{M}$ with underlying schema
$\mathrmbfit{sch}(\mathcal{M})=\mathcal{S}$,
the logic order over $\mathcal{M}$
is isomorphic to 
the specification order over $\mathcal{S}$, 
$\mathrmbf{Log}(\mathcal{M})\cong\mathrmbf{Spec}(\mathcal{S})$. 
%
%
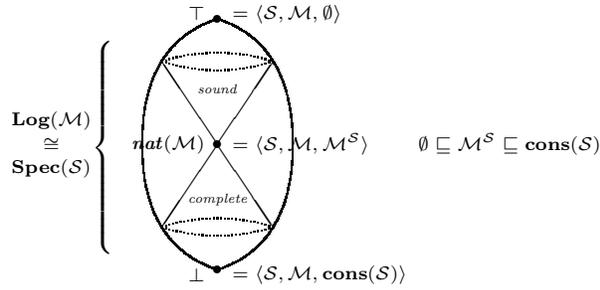
\begin{figure}
\begin{center}
{{\begin{tabular}{@{\hspace{15pt}}c}
\\
\setlength{\unitlength}{0.95pt}
\begin{picture}(150,90)(-45,5)
\put(15,102){\makebox(0,0)[r]{\scriptsize{$\top$}}}
\put(26,102){\makebox(0,0)[l]{\scriptsize{$={\langle{\mathcal{S},\mathcal{M},\emptyset}\rangle}$}}}
\put(15,50){\makebox(0,0)[r]{\scriptsize{$\mathrmbfit{nat}(\mathcal{M})$}}}
\put(26,50){\makebox(0,0)[l]{\scriptsize{$={\langle{\mathcal{S},\mathcal{M},\mathcal{M}^{\mathcal{S}}}\rangle}$}}}
\put(15,-2){\makebox(0,0)[r]{\scriptsize{$\bot$}}}
\put(26,-2){\makebox(0,0)[l]{\scriptsize{$={\langle{\mathcal{S},\mathcal{M},\mathrmbf{cons}(\mathcal{S})}\rangle}$}}}
\put(100,50){\makebox(0,0)[l]{\scriptsize{$\emptyset\sqsubseteq\mathcal{M}^{\mathcal{S}}\sqsubseteq\mathrmbf{cons}(\mathcal{S})$}}}
\put(20,72){\makebox(0,0){\tiny{$\emph{sound}$}}}
\put(20,28){\makebox(0,0){\tiny{$\emph{complete}$}}}
\put(-62,50){\makebox(0,0)[l]{\scriptsize{\shortstack{$\mathrmbf{Log}(\mathcal{M})$\\$\cong$\\$\mathrmbf{Spec}(\mathcal{S})$}}}}
\put(-30,50){\makebox(0,0)[l]{\small{$\left\{ \rule[15pt]{0pt}{30pt} \right.$}}}
\put(20,100){\circle*{3}}
\put(20,50){\circle*{3}}
\put(20,0){\circle*{3}}
\qbezier(-10,50)(-10,90)(20,100)
\qbezier(-10,50)(-10,10)(20,0)
\qbezier(50,50)(50,90)(20,100)
\qbezier(50,50)(50,10)(20,0)
\put(20,50){\line(-2,3){22}}
\put(20,50){\line(2,3){22}}
\qbezier[30](-2,83)(20,90)(42,83)
\qbezier[30](-2,83)(20,76)(42,83)
\put(20,50){\line(-2,-3){22}}
\put(20,50){\line(2,-3){22}}
\qbezier[30](-2,17)(20,24)(42,17)
\qbezier[30](-2,17)(20,10)(42,17)
\end{picture}
\end{tabular}}}
\end{center}
\caption{Logic Order}
\label{fig:log:ord}
\end{figure}
%

\subsubsection{Logic Flow.}\label{sub:sub:sec:log:flow}

The semantic molecules (logics) can be moved along structure morphisms.
For any structure morphism       
$\mathcal{M}_{2}
\xrightleftharpoons{{\langle{r,k,f,g}\rangle}}
\mathcal{M}_{1}$
with underlying schema morphism 
$\mathcal{S}_{2}\stackrel{{\langle{r,f}\rangle}}{\Longrightarrow}\mathcal{S}_{1}$,
define the {\em direct/inverse flow} operators 
\begin{center}
{\footnotesize{$\begin{array}{l@{\hspace{4pt}:\hspace{4pt}}l}
\mathrmbf{Log}(\mathcal{M}_{2})
\xrightarrow{\overrightarrow{\mathrmbfit{log}}_{{\langle{r,k,f,g}\rangle}}}
\mathrmbf{Log}(\mathcal{M}_{1})
&
{\langle{\mathcal{S}_{2},\mathcal{M}_{2},\mathrmbf{T}_{2}}\rangle}
\mapsto
{\langle{\mathcal{S}_{1},\mathcal{M}_{1},\overrightarrow{\mathrmbfit{spec}}_{{\langle{r,f}\rangle}}(\mathrmbf{T}_{2})}\rangle}
\\
\mathrmbf{Log}(\mathcal{M}_{2})
\xleftarrow{\overleftarrow{\mathrmbfit{log}}_{{\langle{r,k,f,g}\rangle}}}
\mathrmbf{Log}(\mathcal{M}_{1})
&
{\langle{\mathcal{S}_{2},\mathcal{M}_{2},\overleftarrow{\mathrmbfit{spec}}_{{\langle{r,f}\rangle}}(\mathrmbf{T}_{1})}\rangle}\mapsfrom
{\langle{\mathcal{S}_{1},\mathcal{M}_{1},\mathrmbf{T}_{1}}\rangle}
\end{array}$}}
\end{center}
These are adjoint monotonic functions w.r.t. logic order:
{\footnotesize{
\begin{equation}\label{eqn:log:mor}
\overrightarrow{\mathrmbfit{log}}_{{\langle{r,k,f,g}\rangle}}(\mathcal{L}_{2})
{\;\geq_{\mathcal{M}_{1}}\;}
\mathcal{L}_{1}
\;\;\text{\underline{iff}}\;\; 
\mathcal{L}_{2}
{\;\geq_{\mathcal{M}_{2}}\;}
\overleftarrow{\mathrmbfit{log}}_{{\langle{r,k,f,g}\rangle}}(\mathcal{L}_{1})
\end{equation}
}\normalsize}
for all target logics $\mathcal{L}_{1}$ 
and source logics $\mathcal{L}_{2}$.

\paragraph{Logic Morphisms}

A logic morphism
$\mathcal{L}_{2}={\langle{\mathcal{S}_{2},\mathcal{M}_{2},\mathrmbf{T}_{2}}\rangle}
\xrightleftharpoons{{\langle{r,k,f,g}\rangle}}
{\langle{\mathcal{S}_{1},\mathcal{M}_{1},\mathrmbf{T}_{1}}\rangle}=\mathcal{L}_{1}$
consists of 
a structure morphism
$\mathcal{M}_{2}
\xrightarrow{{\langle{r,k,f,g}\rangle}}
\mathcal{M}_{1}$ and
a specification morphism
$\mathcal{T}_{2}={\langle{\mathcal{S}_{2},\mathrmbf{T}_{2}}\rangle}
\xrightleftharpoons{{\langle{r,f}\rangle}}
{\langle{\mathcal{S}_{1},\mathrmbf{T}_{1}}\rangle}=\mathcal{T}_{1}$
that share a common schema morphism  
$\mathcal{S}_{2}\xRightarrow{{\langle{r,f}\rangle}}\mathcal{S}_{1}$.
%
A logic morphism
$\mathcal{L}_{2}
\xrightarrow{{\langle{r,k,f,g}\rangle}}
\mathcal{L}_{1}$
is a structure morphism 
$\mathcal{M}_{2}\xrightleftharpoons{{\langle{r,k,f,g}\rangle}}\mathcal{M}_{1}$
that maps the source logic to a generalization of the target logic
$\overrightarrow{\mathrmbfit{log}}_{{\langle{r,k,f,g}\rangle}}(\mathcal{L}_{2})\geq_{\mathcal{M}_{1}}\mathcal{L}_{1}$,
or equivalently,
that maps the target logic to a specialization of the source logic
$\mathcal{L}_{2}\geq_{\mathcal{M}_{2}}\overleftarrow{\mathrmbfit{log}}_{{\langle{r,k,f,g}\rangle}}(\mathcal{L}_{1})$.

%
The context of logics $\mathrmbf{Log}$ 
has logics as objects and logic morphisms as morphisms.
It is the fibered product
\newline\newline\mbox{}\rule{30pt}{0pt}\hfill{
$\mathrmbf{Struc}\xleftarrow{\mathrmbfit{struc}}\mathrmbf{Log}\xrightarrow{\mathrmbfit{spec}}\mathrmbf{Spec}$,
}\hfill{
(Fig.~\ref{fbr:ctx})
}\newline\newline
of the contexts of structures and specifications 
$\mathrmbf{Struc}\xrightarrow{\mathrmbfit{sch}}\mathrmbf{Sch}\xleftarrow{\mathrmbfit{sch}}\mathrmbf{Spec}$.
The projective passages 
satisfy the condition 
$\mathrmbfit{struc}{\;\circ\;}\mathrmbfit{sch} = \mathrmbfit{spec}{\;\circ\;}\mathrmbfit{sch}$.
\footnote{This defines the upper-right part of the {\ttfamily FOLE} superstructure
(Fig.~\ref{fbr:ctx}).}
\comment{
\footnote{Hence, a logic morphism is
(1) a structure morphism
$\mathcal{M}_{2}\xrightleftharpoons{{\langle{r,k,f,g}\rangle}}\mathcal{M}_{1}$,
consisting of a universe morphism 
$\mathcal{U}_{2}\xLeftarrow{{\langle{k,g}\rangle}}\mathcal{U}_{1}$
and a $\mathrmbf{Struc}(\mathcal{U}_{1})$-morphism 
$\mathrmbfit{struc}_{{\langle{k,g}\rangle}}(\mathcal{M}_{2})\xrightleftharpoons{{\langle{r,f}\rangle}}\mathcal{M}_{1}$,
and
(2) a specification morphism
${\langle{\mathcal{S}_{2},\mathrmbf{T}_{2}}\rangle} \xrightarrow{{\langle{r,f}\rangle}} {\langle{\mathcal{S}_{1},\mathrmbf{T}_{1}}\rangle}$
consisting of a schema morphism 
$\mathcal{S}_{2}\xRightarrow{{\langle{r,f}\rangle}}\mathcal{S}_{1}$
and a $\mathrmbf{Spec}(\mathcal{S}_{1})^{\mathrm{op}}$ ordering
$\overrightarrow{\mathrmbfit{spec}}_{{\langle{r,f}\rangle}}(\mathrmbf{T}_{2})\geq_{\mathcal{S}_{1}}\mathrmbf{T}_{1}$
\underline{or} a $\mathrmbf{Spec}(\mathcal{S}_{2})^{\mathrm{op}}$ ordering
$\mathrmbf{T}_{2}\geq_{\mathcal{S}_{2}}\overleftarrow{\mathrmbfit{spec}}_{{\langle{r,f}\rangle}}(\mathrmbf{T}_{1})$.}
}

\subsubsection{Legacy Notions.}\label{sub:sub:sec:log:cons:ext}

\paragraph{Conservative Extensions}

%
Given a structure morphism
$\mathcal{M}_{2}\xrightleftharpoons{{\langle{r,k,f,g}\rangle}}\mathcal{M}_{1}$,
a target logic 
$\mathcal{L}_{1}
\in\mathrmbf{Log}(\mathcal{M}_{1})$
is a conservative extension of 
a source logic 
$\mathcal{L}_{2}
\in\mathrmbf{Log}(\mathcal{M}_{2})$
when the following conditions hold:
(1)
the image of any constraint entailed by $\mathcal{L}_{2}$ is entailed by $\mathcal{L}_{1}$
with
$\overrightarrow{\mathrmbfit{log}}_{{\langle{r,k,f,g}\rangle}}(\mathcal{L}_{2})\geq_{\mathcal{M}_{1}}\mathcal{L}_{1}$,
so that
$\mathcal{L}_{2}\xrightleftharpoons{{\langle{r,k,f,g}\rangle}}\mathcal{L}_{1}$
is a logic morphism;
and
(2)
any constraint 
whose image is entailed by $\mathcal{L}_{1}$ 
is already entailed by $\mathcal{L}_{2}$
with
$\mathcal{L}_{2}\leq_{\mathcal{M}_{2}}\overleftarrow{\mathrmbfit{log}}_{{\langle{r,k,f,g}\rangle}}(\mathcal{L}_{1})$,
so that 
$\mathcal{L}_{2}$ is closed along $\mathcal{M}_{2}\xrightleftharpoons{{\langle{r,k,f,g}\rangle}}\mathcal{M}_{1}$ with
$\mathcal{L}_{2}^{\scriptscriptstyle\blacklozenge_{{\langle{r,k,f,g}\rangle}}}{\;\equiv_{\mathcal{M}_{2}}\;}\mathcal{L}_{2}$.
Hence,
logic $\mathcal{L}_{1}\in\mathrmbf{Log}(\mathcal{M}_{1})$
is a conservative extension of 
logic $\mathcal{L}_{2}\in\mathrmbf{Log}(\mathcal{M}_{2})$
along structure morphism $\mathcal{M}_{2}\xrightleftharpoons{{\langle{r,k,f,g}\rangle}}\mathcal{M}_{1}$
when the underlying specifications satisfy this property along the underlying schema morphisms.
Conservative extensions are closed under composition.


\newpage
\subsection{Sound Logics.}\label{sub:sec:snd:log}

A logic $\mathcal{L}={\langle{\mathcal{S},\mathcal{M},\mathrmbf{T}}\rangle}$ is sound
when the component structure $\mathcal{M}$
satisfies 
the component specification $\mathrmbf{T}$:
$\mathcal{M}{\;\models_{S}\;}\mathrmbf{T}$ or
$\mathcal{M}^{\mathcal{S}}{\;\leq_{\mathcal{S}}\;}\mathrmbf{T}$.
Associated with any $\mathcal{S}$-structure 
$\mathcal{M}$
is the natural logic 
$\mathrmbfit{nat}(\mathcal{M}) = {\langle{\mathcal{S},\mathcal{M},\mathcal{M}^{\mathcal{S}}}\rangle}$,
whose specification is the conceptual intent of $\mathcal{M}$.
The natural logic is the least sound logic: 
$\mathrmbfit{nat}(\mathcal{M}){\;\leq_{\mathcal{M}}\;}\mathcal{L}$
for any sound logic $\mathcal{L}={\langle{\mathcal{S},\mathcal{M},\mathrmbfit{T}}\rangle}$,
since soundness means
$\mathcal{M}^{\mathcal{S}}\leq_{\mathcal{S}}\mathrmbf{T}$
(Fig.~\ref{fig:log:ord}).
\comment{
\begin{center}
$\begin{array}{@{\hspace{25pt}}c@{\hspace{-5pt}}c@{\hspace{-5pt}}c}
\mathcal{M}_{2}
&
\xrightleftharpoons{{\langle{r,k,f,g}\rangle}}
&
{\hspace{-80pt}}\mathcal{M}_{1}
\\
\mathcal{M}_{2}^{\mathcal{S}_{2}}
\stackrel{\scriptscriptstyle{\text{sound}}}{\;\leq_{\mathcal{S}_{2}}\,}
\mathrmbf{T}_{2}
&&
\mathcal{M}_{1}^{\mathcal{S}_{1}}
\stackrel{\scriptscriptstyle{\text{sound}}}{\;\leq_{\mathcal{S}_{2}}\,}
\mathrmbf{T}_{1}
\stackrel{\overset{\text{spec}}{\scriptscriptstyle{\text{morph}}}}{\;\leq_{\mathcal{S}_{2}}\;}
\overrightarrow{\mathrmbfit{spec}}_{{\langle{r,f}\rangle}}
(\mathrmbf{T}_{2})
\\\\
\end{array}$
\end{center}
}
Any structure morphism
$\mathcal{M}_{2}\xrightarrow{\langle{r,k,f,g}\rangle}\mathcal{M}_{1}$
induces the natural logic morphism 
$\mathrmbfit{nat}(\mathcal{M})\xrightarrow{{\langle{r,k,f,g}\rangle}}\mathrmbfit{nat}(\mathcal{M})$,
since
${\langle{\mathcal{S}_{2},\mathcal{M}_{2}^{\mathcal{S}_{2}}}\rangle}
\xrightarrow{\langle{r,f}\rangle}
{\langle{\mathcal{S}_{1},\mathcal{M}_{1}^{\mathcal{S}_{1}}}\rangle}$
is a specification morphism
(Lem.~\ref{lem:nat:spec:mor}).
Hence,
there is a natural logic passage
\newline\newline\mbox{}\rule{30pt}{0pt}\hfill{
$\mathrmbf{Struc}\xrightarrow{\mathrmbfit{nat}}\mathrmbf{Snd}$
}\hfill{
(Fig.~\ref{fbr:ctx})
}\newline\newline
to the subcontext 
$\mathrmbf{Snd}\xhookrightarrow{\mathrmbfit{inc}}\mathrmbf{Log}$
of sound logics.
Structures form a reflective subcontext of sound logics,
since the pair 
${\langle{\mathrmbfit{struc},\mathrmbfit{nat}}\rangle} : \mathrmbf{Snd}\rightleftarrows\mathrmbf{Struc}$
forms an adjunction 
\footnote{An adjunction (generalized pair) of passages; 
that is,
a pair of oppositely-directed passages 
that satisfy inverse equations up to morphism.
Any ``canonical construction from one species of structure to another'' 
is represented by an adjunction between corresponding categories of the two species (Goguen~\cite{goguen:cm91}).}
with $\mathcal{L}{\;\geq\;}\mathrmbfit{nat}(\mathrmbfit{struc}(\mathcal{L}))$
and $\mathrmbfit{nat}{\;\circ\;}\mathrmbfit{struc} = 1_{\mathrmbf{Struc}}$.
\footnote{This defines half of the upper-left part of the {\ttfamily FOLE} superstructure
(Fig.~\ref{fbr:ctx}).}
%

\subsubsection{Residuation.}\label{sub:sub:sec:resid}

Associated with any logic 
$\mathcal{L}={\langle{\mathcal{S},\mathcal{M},\mathrmbf{T}}\rangle}$
is its restriction
$\mathrmbfit{res}_{\mathcal{M}}(\mathcal{L})	
= \mathcal{L}{\;\vee_{\mathcal{M}}\;}\mathrmbfit{nat}(\mathcal{M})
={\langle{\mathcal{S},\mathcal{M},\mathcal{M}^{\mathcal{S}}{\,\vee_{\mathcal{S}}\,}\mathrmbf{T}^{\scriptstyle\bullet}}\rangle}
= {\langle{\mathcal{S},\mathcal{M},\mathcal{M}^{\mathcal{S}}{\,\cap\,}\mathrmbf{T}^{\scriptstyle\bullet}}\rangle}$,
which is the conceptual join 
in $\mathrmbf{Log}(\mathcal{M})$
(Fig.~\ref{fig:log:ord})
of the logic with the natural logic of its structure component.
Clearly,
the restriction is a sound logic and
$\mathrmbfit{res}_{\mathcal{M}}(\mathcal{L})\geq_{\mathcal{S}}\mathcal{L}$.
There is a restriction passage 
%
\newline\newline\mbox{}\rule{30pt}{0pt}\hfill{
$\mathrmbf{Log}\xrightarrow{\mathrmbfit{res}}\mathrmbf{Snd}$,
}\hfill{
(Fig.~\ref{fbr:ctx})
}\newline\newline
which maps a logic $\mathcal{L}$ to the sound logic
$\mathrmbfit{res}(\mathcal{L})$ 
and maps a logic morphism 
$\mathcal{L}_{2}
={\langle{\mathcal{S}_{2},\mathcal{M}_{2},T_{2}}\rangle}
\xrightarrow{{\langle{r,k,f,g}\rangle}}
{\langle{\mathcal{S}_{2},\mathcal{M}_{2},T_{2}}\rangle}=
\mathcal{L}_{2}$
to the morphism of sound logics
\[\mbox{\footnotesize{$
\mathrmbfit{res}(\mathcal{L}_{2})
=
{\langle{\mathcal{S}_{2},\mathcal{M}_{2},
\mathcal{M}_{2}^{\mathcal{S}_{2}}{\vee_{\mathcal{S}_{2}}}\mathrmbf{T}_{2}^{\scriptstyle\bullet}
}\rangle}
\xrightarrow{{\langle{r,k,f,g}\rangle}}
{\langle{\mathcal{S}_{1},\mathcal{M}_{1},
\mathcal{M}_{1}^{\mathcal{S}_{1}}{\vee_{\mathcal{S}_{1}}}\mathrmbf{T}_{1}^{\scriptstyle\bullet}
}\rangle}
=
\mathrmbfit{res}(\mathcal{L}_{1})
$}\normalsize}\]
This is well-defined, 
since it just couples the specification morphism conditions
for the theories of $\mathcal{L}$ and $\mathrmbfit{nat}(\mathcal{M})$
\begin{center}
{\footnotesize{\setlength{\extrarowheight}{3pt}$\begin{array}{c}
\mathcal{M}_{2}^{\mathcal{S}_{2}}
\geq_{\mathcal{S}_{2}}
\overleftarrow{\mathrmbfit{spec}}_{{\langle{r,f}\rangle}}(\mathcal{M}_{1}^{\mathcal{S}_{1}})
\text{ and }
\mathrmbf{T}_{2}^{\scriptstyle\bullet}
\geq_{\mathcal{S}_{2}}
\overleftarrow{\mathrmbfit{spec}}_{{\langle{r,f}\rangle}}(\mathrmbf{T}_{1}^{\scriptstyle\bullet})
\text{ implies }
\\
\mathcal{M}_{2}^{\mathcal{S}_{2}}{\,\vee_{\mathcal{S}_{1}}\,}\mathrmbf{T}_{2}^{\scriptstyle\bullet}
\geq_{\mathcal{S}_{2}}
\overleftarrow{\mathrmbfit{spec}}_{{\langle{r,f}\rangle}}(\mathcal{M}_{1}^{\mathcal{S}_{1}})
{\,\vee_{\mathcal{S}_{2}}\,}
\overleftarrow{\mathrmbfit{spec}}_{{\langle{r,f}\rangle}}(\mathrmbf{T}_{1}^{\scriptstyle\bullet})
=
\overleftarrow{\mathrmbfit{spec}}_{{\langle{r,f}\rangle}}(\mathcal{M}_{1}^{\mathcal{S}_{1}}
{\,\vee_{\mathcal{S}_{1}}\,}
\mathrmbf{T}_{1}^{\scriptstyle\bullet}).
\end{array}$}}
\end{center}
The context of sound logics forms a coreflective subcontext of the context of logics,
since the pair 
${\langle{\mathrmbfit{inc},\mathrmbfit{res}}\rangle} : \mathrmbf{Log} \rightarrow \mathrmbf{Snd}$
forms an adjunction
with $\mathrmbfit{inc}{\;\circ\;}\mathrmbfit{res}{\;\cong\;}\mathrmbfit{1}_{\mathrmbf{Snd}}$
and $\mathrmbfit{inc}(\mathrmbfit{res}(\mathcal{L})){\;\geq_{\mathcal{M}}\;}\mathcal{L}$
for any logic $\mathcal{L}$.
For any structure $\mathcal{M}$,
restriction and inclusion on fibers are adjoint monotonic functions
${\langle{\mathrmbfit{res}_{\mathcal{M}},\mathrmbfit{inc}_{\mathcal{M}}}\rangle}
: \mathrmbf{Log}(\mathcal{M})\rightarrow\mathrmbf{Snd}(\mathcal{M})$,
where
$\mathrmbf{Log}(\mathcal{M})$
is the opposite fiber of logics over $\mathcal{M}$ 
and
$\mathrmbf{Snd}(\mathcal{M})$ is the opposite fiber of sound logics
(Fig.~\ref{fig:log:ord}).


\subsubsection{Sound Logic Flow.}\label{sub:sub:sec:snd:log:flow}

%
The movement of sound logics is a modification of logic flow.
Direct flow preserves soundness (Cor.~\ref{cor:pres:cons}).
Hence, there is no change.
%
Augment inverse flow by restricting to sound logics, 
via residuation,
by joining with structure-intent.
For any structure morphism       
$\mathcal{M}_{2}
\xrightleftharpoons{{\langle{r,k,f,g}\rangle}}
\mathcal{M}_{1}$
with underlying schema morphism 
$\mathcal{S}_{2}\stackrel{{\langle{r,f}\rangle}}{\Longrightarrow}\mathcal{S}_{1}$,
define the {\em direct/inverse flow} operators 
\begin{center}
{\footnotesize{$\begin{array}{l@{\hspace{4pt}:\hspace{4pt}}l}
\mathrmbf{Snd}(\mathcal{M}_{2})
\xrightarrow{\overrightarrow{\mathrmbfit{snd}}_{{\langle{r,k,f,g}\rangle}}}
\mathrmbf{Snd}(\mathcal{M}_{1})
&
{\langle{\mathcal{S}_{2},\mathcal{M}_{2},\mathrmbf{T}_{2}}\rangle}
\mapsto
\overset{\textstyle{\overrightarrow{\mathrmbfit{log}}_{{\langle{r,k,f,g}\rangle}}(\mathcal{L}_{2})}}
{\overbrace{{\langle{\mathcal{S}_{1},\mathcal{M}_{1},\overrightarrow{\mathrmbfit{spec}}_{{\langle{r,f}\rangle}}(\mathrmbf{T}_{2})}\rangle}}}
\\
\mathrmbf{Snd}(\mathcal{M}_{2})
\xleftarrow{\overleftarrow{\mathrmbfit{snd}}_{{\langle{r,k,f,g}\rangle}}}
\mathrmbf{Snd}(\mathcal{M}_{1})
\newline
&
\underset{\textstyle{\mathrmbfit{res}_{\mathcal{M}_{2}}(\overleftarrow{\mathrmbfit{log}}_{{\langle{r,k,f,g}\rangle}}(\mathcal{L}_{1}))}}
{\underbrace{{\langle{\mathcal{S}_{2},\mathcal{M}_{2},
\mathcal{M}_{2}^{\mathcal{S}_{2}}{\,\vee_{\mathcal{M}_{2}}\,}
\overleftarrow{\mathrmbfit{spec}}_{{\langle{r,f}\rangle}}(\mathrmbf{T}_{1})}\rangle}}}
\mapsfrom
{\langle{\mathcal{S}_{1},\mathcal{M}_{1},\mathrmbf{T}_{1}}\rangle}
\end{array}$}}
\end{center}
These are adjoint monotonic functions w.r.t. sound logic order:
{\footnotesize{
\begin{equation}\label{eqn:snd:log:mor}
\overrightarrow{\mathrmbfit{snd}}_{{\langle{r,k,f,g}\rangle}}(\mathcal{L}_{2}){\;\geq_{\mathcal{M}_{1}}\;}\mathcal{L}_{1}
\;\;\text{\underline{iff}}\;\; 
\mathcal{L}_{2}{\;\geq_{\mathcal{M}_{2}}\;}\overleftarrow{\mathrmbfit{snd}}_{{\langle{r,k,f,g}\rangle}}(\mathcal{L}_1)
\end{equation}
}\normalsize}
for all sound target logics $\mathcal{L}_{1}$ 
and sound source logics $\mathcal{L}_{2}$.
\begin{corollary}
Any logic morphism 
$\mathcal{L}_{2}
\xrightarrow{{\langle{r,k,f,g}\rangle}}
\mathcal{L}_{2}$
between sound logics 
satisfies sound logic flow adjointness:
Eqn.~\ref{eqn:log:mor} implies Eqn.~\ref{eqn:snd:log:mor}.
\end{corollary}
\begin{proof}
$\overrightarrow{\mathrmbfit{snd}}_{{\langle{r,k,f,g}\rangle}}(\mathcal{L}_{2}){\;\geq_{\mathcal{M}_{1}}\;}\mathcal{L}_{1}$
\underline{iff}
$\overrightarrow{\mathrmbfit{log}}_{{\langle{r,k,f,g}\rangle}}(\mathcal{L}_{2}){\;\geq_{\mathcal{M}_{1}}\;}\mathcal{L}_{1}$
\underline{iff}
$\mathcal{L}_{2}{\;\geq_{\mathcal{M}_{2}}\;}\overleftarrow{\mathrmbfit{log}}_{{\langle{r,k,f,g}\rangle}}(\mathcal{L}_1)$
\newline
\underline{iff}
$\mathcal{L}_{2}{\;\geq_{\mathcal{M}_{2}}\;}\mathrmbfit{nat}(\mathrmbfit{struc}(\mathcal{L}_{2}))
{\;\vee_{\mathcal{M}_{2}}\;}\overleftarrow{\mathrmbfit{log}}_{{\langle{r,k,f,g}\rangle}}(\mathcal{L}_1)$
\underline{iff}
$\mathcal{L}_{2}{\;\geq_{\mathcal{M}_{2}}\;}\overleftarrow{\mathrmbfit{snd}}_{{\langle{r,k,f,g}\rangle}}(\mathcal{L}_1)$.
\rule{5pt}{5pt}
\end{proof}
%
%
\begin{corollary}
For any structure morphism 
$\mathcal{M}_{2}\xrightleftharpoons{{\langle{r,k,f,g}\rangle}}\mathcal{M}_{1}$,
the restriction-inclusion adjunction on fibers is compatible with the inverse-direct flow adjunction.
This means that the following composite adjoint pairs are equal:
\[\mbox{\footnotesize{$
{\langle{\mathrmbfit{res}_{\mathcal{M}_{1}},\mathrmbfit{inc}_{\mathcal{M}_{1}}}\rangle}{\;\cdot\;}{\langle{\overleftarrow{\mathrmbfit{snd}}_{{\langle{r,k,f,g}\rangle}},\overrightarrow{\mathrmbfit{snd}}_{{\langle{r,k,f,g}\rangle}}}\rangle}
= 
{\langle{\overleftarrow{\mathrmbfit{log}}_{{\langle{r,k,f,g}\rangle}},\overrightarrow{\mathrmbfit{log}}_{{\langle{r,k,f,g}\rangle}}}\rangle}{\;\cdot\;}{\langle{\mathrmbfit{res}_{\mathcal{M}_{2}},\mathrmbfit{inc}_{\mathcal{M}_{2}}}\rangle}
$.}\normalsize}\]
\end{corollary}
%
\begin{center}
{{\begin{tabular}{c}
\setlength{\unitlength}{1.0pt}
\begin{picture}(80,90)(0,-10)
\put(0,60){\makebox(0,0){\footnotesize{$\mathrmbf{Snd}(\mathcal{M}_{2})$}}}
\put(80,60){\makebox(0,0){\footnotesize{$\mathrmbf{Log}(\mathcal{M}_{2})$}}}
\put(0,0){\makebox(0,0){\footnotesize{$\mathrmbf{Snd}(\mathcal{M}_{1})$}}}
\put(80,0){\makebox(0,0){\footnotesize{$\mathrmbf{Log}(\mathcal{M}_{1})$}}}
\put(40,74){\makebox(0,0){\scriptsize{$\mathrmbfit{inc}_{\mathcal{M}_{2}}$}}}
\put(42,46){\makebox(0,0){\scriptsize{$\mathrmbfit{res}_{\mathcal{M}_{2}}$}}}
\put(40,14){\makebox(0,0){\scriptsize{$\mathrmbfit{int}_{\mathcal{M}_{1}}$}}}
\put(42,-14){\makebox(0,0){\scriptsize{$\mathrmbfit{res}_{\mathcal{M}_{1}}$}}}
\put(15,28){\makebox(0,0)[r]{\scriptsize{$\overleftarrow{\mathrmbfit{snd}}_{{\langle{r,k,f,g}\rangle}}$}}}
\put(10,32){\makebox(0,0)[l]{\scriptsize{$\overrightarrow{\mathrmbfit{snd}}_{{\langle{r,k,f,g}\rangle}}$}}}
\put(95,28){\makebox(0,0)[r]{\scriptsize{$\overleftarrow{\mathrmbfit{log}}_{{\langle{r,k,f,g}\rangle}}$}}}
\put(90,32){\makebox(0,0)[l]{\scriptsize{$\overrightarrow{\mathrmbfit{log}}_{{\langle{r,k,f,g}\rangle}}$}}}
\put(22,66){\vector(1,0){36}}
\put(58,54){\vector(-1,0){36}}
\put(22,6){\vector(1,0){36}}
\put(58,-6){\vector(-1,0){36}}
\put(-6,12){\vector(0,1){36}}
\put(6,48){\vector(0,-1){36}}
\put(74,12){\vector(0,1){36}}
\put(86,48){\vector(0,-1){36}}
\end{picture}
\end{tabular}}}
\end{center}
%
\begin{proof}
We show
$\mathrmbfit{res}_{\mathcal{M}_{1}}{\;\cdot\;}\overleftarrow{\mathrmbfit{snd}}_{{\langle{r,k,f,g}\rangle}}
= \overleftarrow{\mathrmbfit{log}}_{{\langle{r,k,f,g}\rangle}}{\;\cdot\;}\mathrmbfit{res}_{\mathcal{M}_{2}}$.
\newline
Since
$\mathcal{M}_{2}^{\mathcal{S}_{2}}
{\,\geq_{\mathcal{M}_{2}}\,}
\overleftarrow{\mathrmbfit{spec}}_{{\langle{r,f}\rangle}}(\mathcal{M}_{1}^{\mathcal{S}_{1}})$,
we have
$\overleftarrow{\mathrmbfit{snd}}_{{\langle{r,k,f,g}\rangle}}(\mathrmbfit{res}_{\mathcal{M}_{1}}(\mathcal{L}_{1}))
\newline
=
{\langle{\mathcal{S}_{2},\mathcal{M}_{2},
\mathcal{M}_{2}^{\mathcal{S}_{2}}{\,\vee_{\mathcal{M}_{2}}\,}
\overleftarrow{\mathrmbfit{spec}}_{{\langle{r,f}\rangle}}(
\mathcal{M}_{1}^{\mathcal{S}_{1}}{\,\vee_{\mathcal{M}_{1}}\,}
\mathrmbf{T}_{1})}\rangle}
\newline
=
{\langle{\mathcal{S}_{2},\mathcal{M}_{2},
\mathcal{M}_{2}^{\mathcal{S}_{2}}
{\,\vee_{\mathcal{M}_{2}}\,}
\overleftarrow{\mathrmbfit{spec}}_{{\langle{r,f}\rangle}}(\mathcal{M}_{1}^{\mathcal{S}_{1}})
{\,\vee_{\mathcal{M}_{2}}\,}
\overleftarrow{\mathrmbfit{spec}}_{{\langle{r,f}\rangle}}(\mathrmbf{T}_{1})}\rangle}
\newline
=
{\langle{\mathcal{S}_{2},\mathcal{M}_{2},
\mathcal{M}_{2}^{\mathcal{S}_{2}}
{\,\vee_{\mathcal{M}_{2}}\,}
\overleftarrow{\mathrmbfit{spec}}_{{\langle{r,f}\rangle}}(\mathrmbf{T}_{1})}\rangle}
=
\mathrmbfit{res}_{\mathcal{M}_{2}}(\overleftarrow{\mathrmbfit{log}}_{{\langle{r,k,f,g}\rangle}}(\mathcal{L}_{1}))$.
\rule{5pt}{5pt}
\end{proof}
%


%
\begin{figure}
\begin{center}
\begin{tabular}[t]{@{\hspace{50pt}}c@{\hspace{60pt}}c}
& \\
{{\begin{tabular}[t]{c}
\\
\setlength{\unitlength}{0.95pt}
\begin{picture}(80,120)(-10,-10)
\put(0,120){\makebox(0,0){\normalsize{$\mathrmbf{Snd}$}}}
\put(80,120){\makebox(0,0){\normalsize{$\mathrmbf{Log}$}}}
\put(0,60){\makebox(0,0){\normalsize{$\mathrmbf{Struc}$}}}
\put(40,5.5){\makebox(0,0){\normalsize{$\mathrmbf{Sch}$}}}
\put(80,60){\makebox(0,0){\normalsize{$\mathrmbf{Spec}$}}}
%
\put(42,70){\makebox(0,0){\footnotesize{$\mathrmbfit{int}$}}}
\put(17,29){\makebox(0,0)[r]{\footnotesize{$\mathrmbfit{sch}$}}}
\put(64,29){\makebox(0,0)[l]{\footnotesize{$\mathrmbfit{sch}$}}}
\put(-8,90){\makebox(0,0)[r]{\footnotesize{$\mathrmbfit{nat}$}}}
\put(8,90){\makebox(0,0)[l]{\footnotesize{$\mathrmbfit{struc}$}}}
\put(46,90){\makebox(0,0)[l]{\footnotesize{$\mathrmbfit{struc}$}}}
\put(83,90){\makebox(0,0)[l]{\footnotesize{$\mathrmbfit{spec}$}}}
\put(40,133){\makebox(0,0){\footnotesize{$\mathrmbfit{inc}$}}}
\put(40,107){\makebox(0,0){\footnotesize{$\mathrmbfit{res}$}}}
%
\qbezier(40,37)(45,32)(50,27)
\qbezier(40,37)(35,32)(30,27)
\put(-6,72){\vector(0,1){36}}
\put(6,108){\vector(0,-1){36}}
\put(80,108){\vector(0,-1){36}}
\put(22,126){\vector(1,0){36}}
\put(58,114){\vector(-1,0){36}}
\put(72,48){\vector(-2,-3){22}}
\put(22,60){\vector(1,0){40}}
\put(8,48){\vector(2,-3){22}}
\put(65,107){\vector(-4,-3){48}}
\put(-20,50){\oval(20,20)[l]}
\put(-10,50){\oval(20,20)[br]}
\qbezier(-20,40)(-15,40)(-10,40)
\put(0,55){\vector(0,1){0}}
\put(-14,46){\makebox(0,0){\scriptsize{$\mathrmbfit{fmla}$}}}
\put(-14,35){\makebox(0,0){\scriptsize{$\mathrmbfit{im}$}}}
%
\put(25,-5){\oval(20,20)[l]}
\put(30,-5){\oval(20,20)[br]}
\qbezier(25,-15)(27.5,-15)(30,-15)
\put(40,0){\vector(0,1){0}}
\put(28,-8){\makebox(0,0){\scriptsize{$\mathrmbfit{fmla}$}}}
\end{picture}
\\
\end{tabular}}}
&
{\setlength{\extrarowheight}{2pt}\footnotesize{$\begin{array}[t]{l}

\mathrmbfit{inc}{\;\circ\;}\mathrmbfit{res}{\;\cong\;}\mathrmbfit{1}_{\mathrmbf{Snd}}
\\
\mathrmbfit{inc}(\mathrmbfit{res}(\mathcal{L})){\;\geq_{\mathcal{M}}\;}\mathcal{L}
\\ \hline
\mathrmbfit{nat}{\;\circ\;}\mathrmbfit{struc}{\;=\;}\mathrmbfit{1}_{\mathrmbf{Struc}}
\\
\mathcal{L}{\;\geq_{\mathcal{M}}\;}\mathrmbfit{nat}(\mathrmbfit{struc}(\mathcal{L}))
\\ \hline
\mathrmbfit{int}{\;\circ\;}\mathrmbfit{sch}{\;=\;}\mathrmbfit{sch}
\end{array}$}}
\\ & \\
\end{tabular}
\end{center}
\caption{{\ttfamily FOLE} Superstructure}
\label{fbr:ctx}
\end{figure}
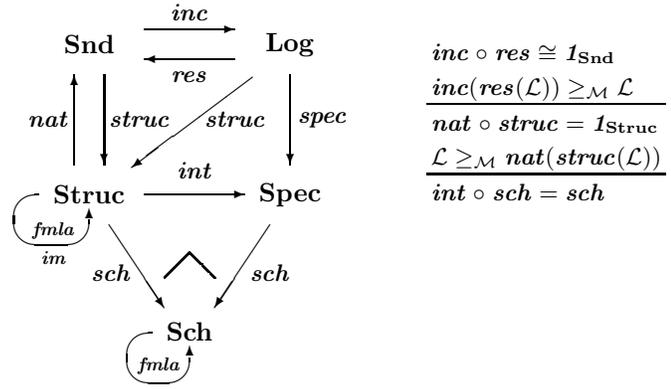
%

%% file: conclu.tex
\newpage
\section{Conclusion and Future Work}\label{sec:conclu}

The work in this paper consisted of two parts:
development of the {\ttfamily FOLE} logical environment and
presentation of the {\ttfamily FOLE} superstructure.
The development of the {\ttfamily FOLE} logical environment centered on 
the satisfaction relation between the polar opposition (formalism\rule[2.5pt]{5pt}{0.4pt}\,semantics).
At the upper pole, we defined the formalism of formulas, sequents and constraints;
the latter two allow us to specify ontological hierarchies.
At the lower pole, we developed semantics through the interpretation and  classification of formulas;
here we defined the valuable concept of comprehension. 
Bridging the poles
is the satisfaction relation between a structure and a formalism (sequent or constraint).
Finally,
to finish the work on the {\ttfamily FOLE} logical environment, 
we expressed {\ttfamily FOLE} as an institution;
and more particularly, as a logical environment. 

The presentation of the {\ttfamily FOLE} superstructure 
involved the mathematical contexts, passages and adjunctions
illustrated in the {\ttfamily FOLE} architectural diagram (Fig.~\ref{fbr:ctx}).
%
This diagram
consists of four components:
structures, specifications, logics and sound logics.
Structures, which represent the semantic aspect of {\ttfamily FOLE}, 
were handled in the {\ttfamily FOLE} foundation paper \cite{kent:fole:era:found}.
In this paper, we present the remaining architectural components: specifications, logics and sound logics.
Specifications represent the formal aspect of {\ttfamily FOLE}; 
here, we define the notions of entailment, consequence and flow of formalism.
Logics combine the formal and semantic aspects of {\ttfamily FOLE}.
Logics are sound when semantics satisfies formalism.



As outlined in the introduction \S\ref{sec:intro},
this paper 
is one of a series of papers
that provide a rigorous mathematical representation
for ontologies within the first-order logical environment {\ttfamily FOLE}.
The {\ttfamily FOLE} representation can be expressed in two forms: 
a classification form and interpretative form.
The foundation paper \cite{kent:fole:era:found} and 
the current superstructure paper 
develop the classification form of {\ttfamily FOLE}.
The paper 
\cite{kent:fole:era:tbl}
and
the paper 
\cite{kent:fole:era:db}
develop the interpretative form of {\ttfamily FOLE} 
as a transformational passage from sound logics \cite{kent:iccs2013},
thereby defining the formalism and semantics of first-order logical/relational database systems 
\cite{kent:db:sem}.

System interoperability, 
in the general setting of institutions and logical environments,
was defined in the paper ``System Consequence'' (Kent~\cite{kent:iccs2009}). 
This was inspired by the channel theory of information flow presented in 
the book 
{\itshape Information Flow: The Logic of Distributed Systems} (Barwise and Seligman \cite{barwise:seligman:97}).
Since {\ttfamily FOLE} is a logical environment (\S\ref{sub:sub:sec:inst:asp}),
in two further papers we apply this approach to interoperability
for information systems based on first-order logic and relational databases:
one paper discusses integration over a fixed type domain and
the other paper discusses integration over a fixed universe.
%